\documentclass[letter,11pt]{article}
\usepackage{fullpage}
\usepackage{microtype}
\usepackage{mathtools,amsmath} \usepackage{amssymb,amsthm}
\usepackage{graphicx}
\usepackage{hyperref,color}
\usepackage{lineno}
\usepackage{romannum}
 \usepackage[margin=1in]{geometry}

\newcommand{\rdist}{\ensuremath{\textsf{rdist}}}
\newcommand{\ring}{\ensuremath{\textsf{Ring}}}
\newcommand{\oring}{\ensuremath{\textsf{Ring}_i^\textsf{o}}}
\newcommand{\xring}{\ensuremath{\textsf{Ring}_i^\textsf{x}}}
\newcommand{\gparallel}{\ensuremath{G^\textsf{p}}}
\newcommand{\radgraph}{\ensuremath{G^\textsf{rad}}}
\newcommand{\windnum}{\textsf{WindNum}}
\newcommand{\gcontract}{G_\textsf{contract}}
\newcommand{\gmod}{G_\textsf{mod}}
\newcommand{\wmod}{\mathcal W_\textsf{mod}}
\newbox\ProofSym
\setbox\ProofSym=\hbox{%
	\unitlength=0.18ex%
	\begin{picture}(10,10)
	\put(0,0){\framebox(9,9){}}
	\put(0,3){\framebox(6,6){}}
	\end{picture}}

\def\ccheck#1{{\textcolor{red}{#1}}}

\graphicspath{{figures/}}

\newtheorem{theorem}{Theorem} 
\newtheorem{lemma}[theorem]{Lemma}

\newtheorem{corollary}[theorem]{Corollary}
\newtheorem{definition}[theorem]{Definition}
\newtheorem{observation}[theorem]{Observation}

\begin{document}
\pagenumbering{arabic}
\title{Parameterized Algorithm for the Planar Disjoint Paths Problem: Exponential in $k^2$, and Linear in $n${\footnote{This work was supported by the National Research Foundation of Korea (NRF) grant funded by the Korea government (MSIT) (No.2020R1C1C1012742).}}}

%
%

\author{
Kyungjin Cho\footnote{Pohang University of Science and Technology,	Korea.
		Email: {\tt{kyungjincho@postech.ac.kr}}}
\and
Eunjin Oh\footnote{Pohang University of Science and Technology,	Korea.
		Email: {\tt{eunjin.oh@postech.ac.kr}}}
\and
Seunghyeok Oh\footnote{Pohang University of Science and Technology,	Korea.
		Email: {\tt{seunghyeokoh@postech.ac.kr}}}}

\maketitle              
\begin{abstract}
In this paper, we study the \textsf{Planar Disjoint Paths} problem:
Given an undirected planar graph $G$ with $n$ vertices and a set $T$ of $k$ pairs $(s_i,t_i)_{i=1}^k$ of vertices,
the goal is to find a set $\mathcal P$ of $k$ pairwise vertex-disjoint paths connecting $s_i$ and $t_i$ for all indices $i\in\{1,\ldots,k\}$. 
We present a $2^{O(k^2)}n$-time algorithm for the \textsf{Planar Disjoint Paths} problem. This improves the two previously best-known algorithms: $2^{2^{O(k)}}n$-time algorithm [Discrete Applied Mathematics 1995] and $2^{O(k^2)}n^6$-time algorithm [STOC 2020].
\end{abstract}
\newpage
\setcounter{page}{1}
\section{Introduction}
In this paper, we study the \textsf{Disjoint Paths} problem on planar graphs: 
Given an undirected graph $G$ with $n$ vertices and a set $T$ of $k$ pairs $(s_i,t_i)_{i=1}^k$ of vertices,
the goal is to find a set $\mathcal P$ of $k$ pairwise vertex-disjoint paths connecting $s_i$ and $t_i$ for all indices $i\in\{1,\ldots,k\}$. 
This problem has been studied extensively due to its numerous applications such as VLSI layout and circuit routing~\cite{frank1990packing}. 
However, this problem is NP-complete even for grid graphs~\cite{chuzhoy2018almost}. 
This motivates  
	the study of this problem from the viewpoint of parameterized algorithms and approximation algorithms. 
	For approximation algorithms,  
	we wish to connect as many terminal pairs as possible using vertex-disjoint paths in polynomial time.
The best known approximation algorithm has approximation ratio $O(\sqrt{n})$~\cite{kolliopoulos2004approximating}. 
In the case that the input graph is restricted to  be planar, the best known approximation ratio is $O(n^{9/19}\log^{O(1)}n)$~\cite{chuzhoy2016improved}. 
On the other hand, under reasonable complexity-theoretic assumptions, no polynomial-time algorithm
has approximation factor better than $2^{\Omega(\sqrt{\log n})}$~\cite{chuzhoy2020new}.
	
	The \textsf{Disjoint Paths} problem has been also studied extensively from the viewpoint of parameterized algorithms. Here, the goal is to design an  algorithm 
	which runs in $f(k)\cdot n^{O(1)}$ time for a function $f$, where $k$ is the number of terminal pairs. 
    The seminal work of Roberson and Seymour~\cite{robertson1995graph} 
    showed that this problem can be solved in $f(k)\cdot n^3$ time for a computable function $f$.
    Although this algorithm is powerful, the original proof of Roberson and Seymour requires the full power of the Graph Minor Thoery.
    For this reason, Kawarabayashi and Wollan~\cite{kawarabayashi2010shorter} presented a simpler proof of correctness of the algorithm in~\cite{robertson1995graph}. 
    Later, this algorithm was improved by Kawarabayashi et al.~\cite{kawarabayashi2012disjoint} 
    to run in $h(k)\cdot n^2$ time for a computable function $h$. 
    However, the dependence on $k$ is still huge both in~\cite{kawarabayashi2012disjoint,robertson1995graph}. 	
    Moreover, the explicit bound on the dependence on $k$ is  not known.  
	
	For this reason, the \textsf{Disjoint Paths} problem has been studied extensively for planar graphs~\cite{adler2017irrelevant,cygan2013planar,lokshtanov2020exponential,reed1995rooted,schrijver1994finding}.  In this case, we call the problem the \textsf{Planar Disjoint Paths} problem. 
	Note that there are two components of the running times of parameterized algorithms: the dependency on the parameter $k$ and 
	the dependency on the input size $n$.  
	Reed~\cite{reed1995rooted} 
	focused on the dependence on the input size, and gave a $2^{2^{O(k)}}n$-time algorithm for the \textsf{Planar Disjoint Paths} problem. 
	On the other hand, Lokshtanov et al.~\cite{lokshtanov2020exponential} focused on the dependence on the parameter $k$, and gave a  $2^{O(k^2)}n^6$-time algorithm for this problem.\footnote{The two papers \cite{reed1995rooted} and \cite{lokshtanov2020exponential} 
	do not give explicit bounds on the dependent on $k$ and dependent on $n$, respectively. But it is not difficult to see that
	these algorithms run in $2^{2^{\Theta(k)}}n$-time and $2^{\Theta(k^2)}n^6$-time, respectively.}
	A natural question posed here is to achieve 
	both best dependency on $k$ and the best dependency on $n$ simultaneously.  
	
As mentioned in~\cite{lokshtanov2018linear}, this direction of research is indeed as old as the field of  parameterized algorithms.
The result of Bodlaender~\cite{bodlaender1996linear}
on on the problem for computing a treewidth of a graph 
belongs to this class of results. Also, the  $f(k)\cdot n^3$-time and  $h(k)\cdot n^2$-time algorithms for the \textsf{Disjoint Paths} problem
by Roberson and Seymour~\cite{robertson1995graph} and by Kawarabayashi et al.~\cite{kawarabayashi2012disjoint}, respectively,
belong to this class of results. 
There are more recent results in this direction, for instance, \textsf{Odd Cycle Transversal}~\cite{iwata2014linear,kawarabayashi2010almost}, 
\textsf{Subgraph Isomorphism}~\cite{dorn2010planar}, \textsf{Planarization}~\cite{jansen2014near}, and \textsf{Treewidth}~\cite{korhonen2022single}.

\medskip
In this paper, we present an algorithm for the \textsf{Planar Disjoint Paths} problem which runs in $2^{O(k^2)}n$ time. This improves both the $2^{2^{O(k)}}n$-time algorithm by Reed~\cite{reed1995rooted}
and the $2^{O(k^2)}n^6$-time algorithm by Lokshtanov et al.~\cite{lokshtanov2020exponential}. 
Our contribution in this paper can be summarized as follows. 

\begin{theorem} 
The \textsf{Planar Disjoint Paths} problem can be solved in $2^{O(k^2)}n$ time.
\end{theorem}

\section{Preliminaries}
An instance of the \textsf{Planar Disjoint Paths} problem is a tuple $(G, T, k)$ where $G$ is a plane graph, $T$ is a set $\{(s_1,t_1),\ldots, (s_k,t_k)\}$  of $k$ vertex pairs. Here, the embedding of $G$ is fixed. 
We let $\bar{T}$ be the set of all vertices in the pairs  of $T$, and we call such vertices the \emph{terminals}. 

A \emph{walk} $\omega$ is a sequence of edges which joins a sequence of vertices. 
If all vertices and all edges of $\omega$ are distinct, we call $\omega$ a \emph{path}. 
A \emph{$T$-linkage} is an ordered family $\langle P_1,\ldots, P_k\rangle $ of $k$ \emph{vertex-disjoint} paths  in $G$ such that $P_i$ connects $s_i$ and $t_i$ for $i\in \{1,\ldots,k\}$. 
 We say two paths $P$ and $P'$ are \emph{crossing} 
 if there are four edges $e,f,e',f'$ of $G$ sharing a common endpoint such that $e, f$ are consecutive edges of $P$, $e',f'$ are consecutive edges of $P'$, and $e, f, e'$ and $f'$ lie in  clockwise order around their common endpoint. 
A \emph{weak $T$-linkage} $\mathcal W$ is  an ordered family $\langle P_1,\ldots, P_k\rangle $ of $k$ pairwise \emph{non-crossing} walks  in $G$ such that $P_i$ connects $s_i$ and $t_i$ for $i\in\{1,\ldots,k\}$. 
If all edges of the walks of a weak $T$-linkage $\mathcal W$ are distinct, we say $\mathcal W$ is \emph{edge-disjoint}. 
We sometimes call a $T$-linkage and a weak $T$-linkage for a set $T$ of terminal pairs
simply a \emph{linkage} and a \emph{weak linkage}, respectively, when we do not need to specify $T$. 

We use standard notions and terms for graphs as used in~\cite{bondy1976graph}. For instance, 
for a graph $G$, we let $V(G)$ be the vertex set of $G$, and $E(G)$ be the edge set of $G$. 
Also, throughout this paper, 
we use $[n]$ to denote the set $\{1,2,\ldots, n\}$.

\begin{definition}
A \emph{tree decomposition} of an undirected graph $G=(V,E)$
is defined as a pair $(T,\beta)$, where $T$ is a tree and $\beta$ is a mapping from nodes of $T$ to subsets of $V$ (called bags) satisfying the following properties. Let $\mathcal {B}:=\{\beta(t) : t\in V(T)\}$ be the set of bags of $T$.
	\begin{itemize}\setlength\itemsep{0.1em}
		\item  For any vertex $u\in V$, there is at least one bag in $\mathcal {B}$ which contains $u$.
		\item For any edge $(u,v)\in E$, there is at least one bag in $\mathcal {B}$ which contains both $u$ and $v$.
		\item For any vertex $u\in V$, the nodes of $T$ containing $u$ in their bags are connected in $T$.
	\end{itemize}
	The $\emph{width}$ of a tree decomposition is defined as the size of its largest bag minus one, and the $\emph{treewidth}$ of $G$ is
	the minimum width of a tree decomposition of $G$.
\end{definition}


\subsection{Radial Distance and Radial Curves} 
The \emph{radial distance} in $G$ between two faces $f_1$ and $f_2$ of $G$ is defined as the minimum length of 
a sequence of faces starting from $f_1$ to $f_2$, such that
every two consecutive faces of this sequence share a common vertex. 
Similarly, the radial distance between two vertices $u$ and $v$, denoted by $\textsf{rdist}(u,v)$,
is defined as the minimum radial distance between two faces $F_u$ and $F_v$ incident to $u$ and $v$, respectively. 
Note that there is a curve connecting $u$ and $v$ 
which meets $G$ only at $\textsf{rdist}(u,v)$ vertices {(excluding $u$ and $v$)} and does not intersect any edges of $G$. 
We call such a curve a \emph{radial curve} connecting $u$ and $v$. See Figure~\ref{fig:face_operation}(a). 
For two subsets $X$ and $Y$ of $V$, we define their radial distance
as the Hausdorff distance: $\rdist(X,Y)=\max\{ \max_{x\in X}\min_{y\in Y}\rdist(x,y),  \max_{y\in Y}\min_{x\in X}\rdist(x,y)\}$.
For two subgraphs $P$ and $P'$ of $G$, we define their radial distance as $\rdist(V(P), V(P'))$.

An \emph{innermost} face $F^*$ of $G$ is defined as a face farthest from the outer face with respect to the radial distance. 
We can compute $F^*$ in linear time using a variant of the breadth-first search on $G$. 
Also, we call a simple closed curve a \emph{noose} if it intersects $G$ only at vertices of $G$, and does not intersect any edge of $G$. 
The complexity of a noose is defined as the number of vertices of $G$ intersected by the noose. 
A similar notion we use in this paper is a \emph{face-edge path} in $G$: an alternating sequence 
of faces and edges 
such that every face is incident (in $G$) to the edges neighboring in the sequence, and the two end elements of the sequence are faces.

The \emph{radial completion} of $G$, denoted by $G^\textsf{rad}$, is a subgraph of $G$ constructed as follows. We add one vertex for each face $F$ of $G$, and connect this vertex and all vertices lying on $F$. Note that every face of $G^\textsf{rad}$ is a triangle. 
See Figure~\ref{fig:face_operation}(b). 
One can consider a radial curve of length $N$
as a path in $G^\textsf{rad}$ of length $\Theta(N)$. 
We use the radial completion to define the discrete homotopy in the following subsection.

\subsection{Homology and Flow Function}
Schrijver~\cite{schrijver1994finding} used the language of flows and homology to deal with linkages and weak-linkages.
As the algorithm of~\cite{schrijver1994finding} works with directed graphs, we consider $G$ as a directed graph by replacing each edge of $G$
with two directed edges of $G$ with opposite directions. Also, we consider a walk of a $T$-linkage of $G$ as oriented from $s_i$ to $t_i$. 
Let $\Sigma=\{1,2,\ldots, k\}$ be an alphabet consisting of $k$ symbols, and let $\Sigma^{-1}=\{1^{-1}, 2^{-1},\ldots, k^{-1}\}$. 
The group $\Sigma^*$ generated by $\Sigma$ consists of all strings $b_1b_2\ldots b_t$ with $t\geq 0$ such that $b_1,\ldots,b_t\in \Sigma \cup \Sigma^{-1}$ and
$b_i\neq {b_{i+1}^{-1}}$ for all indices $i\in[t)$. 
Then the product $x\cdot y$ of two strings $x$ and $y$ in $\Sigma^*$ is defined as the string obtained from the concatenation $xy$ by deleting iteratively all occurrences of $s^{-1}s$ and $ss^{-1}$ for $s\in\Sigma$. The empty string is denoted by  $\epsilon$.
Note that $xx^{-1}=x^{-1}x=\epsilon$ for any string $x\in\Sigma^*$. 

  \paragraph{Flow Function and Linkages.}
  A flow function  $\phi:E(G) \rightarrow \Sigma^*$  is defined as follows. 
  For a vertex $v$, let $e_1,\dots,e_\ell$ be its incident edges in  $G$ sorted in the clockwise direction.
  Then $h(v)$ be the product of  $\phi(e_r)^{\mu(e_r)}$'s for all $r=1,2,\ldots,\ell$,
  where the sign $\mu(e_r)$ is positive if $e_r$ is an outgoing edge from $v$, and negative, otherwise.\footnote{Note that $h(v)$ depends on the choice of $e_1$. To avoid the ambiguity, we choose $e_1$ so that $|h(v)|$ is minimized.} 
  A function $\phi$ is called a \emph{flow function}
  if $h(s_i)=i$ for all $i\in[k]$, $h(t_i)=i^{-1}$ for all $i\in[k]$,
  and $h(v)=\epsilon$ for $v\notin \bar{T}$.  
  This is an algebraic interpretation of the standard flow-conservation constraint. 
  
  An edge-disjoint weak $T$-linkage $\mathcal W$ can be considered as a \emph{flow function} 
   $\phi : E \to \Sigma^*$. 
   Let $W_i$ be the path of $\mathcal W$ connecting $s_i$ and $t_i$ for $i\in[k]$. 
    For each edge $e\in E(G)$, let $\phi(e)=i$ if it is used in the path connecting $s_i$ to $t_i$, and 
  let $\phi(e)=\epsilon$ if it is not used by any path of $\mathcal W$. 
   We say $\phi$ is the flow function \emph{representing} $\mathcal W$.
  But it is possible that a flow function $\phi$ represents more than two weak $T$-linkages of $G$. 
  For illustration, refer to Figure~2 of~\cite{lokshtanov2020exponential}.
  Schrijver~\cite{schrijver1994finding} showed that a flow function $\phi$ represents a $T$-linkage of $G$ 
  if and only if  $\phi(\pi)\in \Sigma\cup \Sigma^{-1}$ for every face-edge path consisting of faces and edges sharing a common vertex. 
  Here, a face-edge path in $G$ is an alternating sequence of faces and edges 
  such that every face is incident (in $G$) to the edges neighboring in the sequence.

  \paragraph{Homology.}
  For a directed edge $e$ of $G$, we let $L_e$ denote the face lying to the left of $e$, and let $R_e$ denote
  the face lying to the right of $e$. 
  We say two flow functions $\phi$ and $\psi$ are \emph{homologous} if there exists a \emph{homology} function $f: \mathcal{F}\to \Sigma^*$,
  where $\mathcal F$ denotes the set of all faces of $G$, such that 
  \begin{itemize}
      \item {$f(F^*)=\epsilon$, and }
      \item {$(f(L_e))^{-1}\cdot \phi(e)\cdot f(R_e)=\psi(e)$ for every directed edge $e$.}
  \end{itemize}
  
  \subsection{ Discrete Homotopy and Linkages} 
Instead of working with homology, 
Lokshtanov et al.~\cite{lokshtanov2020exponential} introduced a notion of \emph{discrete homotopy}, 
which is a variant of the standard homotopy. When we deal with discrete homotopy, we always work with the radial completion $\radgraph$ of $G$. 
Two weak linkages $\mathcal W$ and $\mathcal W'$ 
are \emph{homotopic} to each other if one can obtained from the other weak linkage by a sequence of \emph{face operations}.
There are three types of face operations: \textsf{Face Move}, \textsf{Face Pull}, and \textsf{Face Push}.
Since our argument uses \textsf{Face Move} and \textsf{Face Pull} only, we give the definitions of \textsf{Face Move} and \textsf{Face Pull} only. The definition of \textsf{Face Push} can be found in~{\cite[Definition 5.3]{lokshtanov2020exponential}}. Each operation is applied to a face $F$ of $\radgraph$ and a walk $W$ of $\mathcal W$. Let $\partial F$ be the cycle consisting of the three edges incident to $F$. 
For illustration, see Figure~\ref{fig:face_operation}(c--d). 

\begin{figure}
    \centering
    \includegraphics[width=0.8\textwidth]{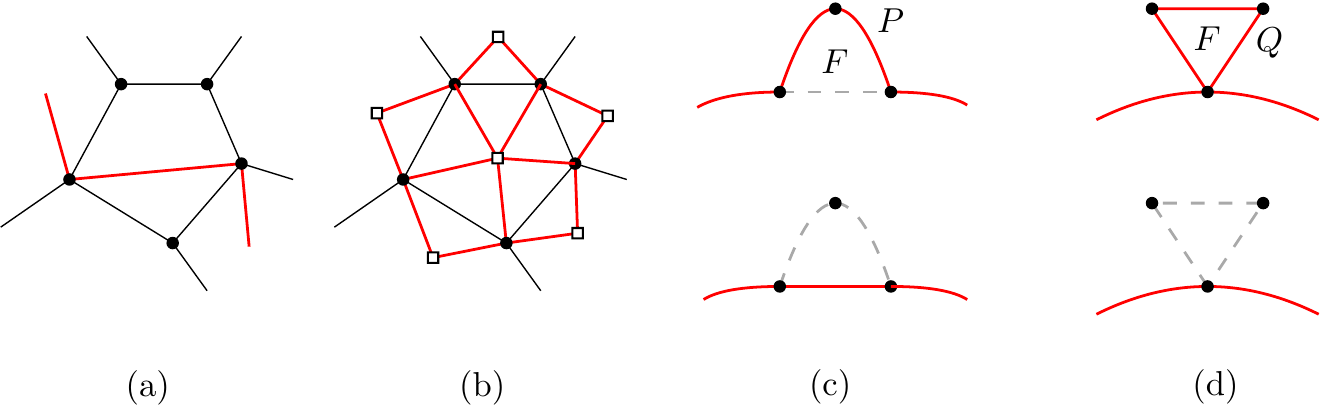}
    \caption{\small 
    (a) The red curve is a part of a radial curve.
    (b) The edges colored red and the points marked with boxes are added to construct $\radgraph$. 
    (c) The face move operation is applied to $(F,P)$. 
    (d) the face pull operation is applied to $(F,Q)$.
    }
    \label{fig:face_operation}
\end{figure}

\begin{itemize}
    \item \textsf{Face Move.} The move operation is applicable to $(F,W)$ only when there is a subpath $P$ of $\partial F$ of length at most two
    such that $P$ is a subwalk of $W$, and no edge of $E(\partial F)\setminus E(P)$ belongs to any walk of $\mathcal W$.
    In this case, we replace $P$ in $W$ by the unique subpath of $\partial F$ between the endpoints of $P$ that is edge-disjoint from $P$.
    \item \textsf{Face Pull.} The pull operation is applicable to $(F,W)$ only when $\partial F$ is a subwalk $Q$ of $W$. 
    In this case, we replace $Q$ in $W$ by a single occurrences of the first vertex in $Q$.
\end{itemize}

The following lemma connects the concept of discretely homotopy and the concept of homology. 
\begin{lemma}[{\cite[Lemma 5.2]{lokshtanov2020exponential}}]\label{lem:homology-homotopy}
Let $\mathcal W$ and $\mathcal W'$ be two weak linkages discretely homotopic to each other. 
Then the flow function representing $\mathcal W$ is homologous to the flow function representing $\mathcal W'$. 
\end{lemma} 

\section{Overview}
In this section, we present an overview of our algorithm. To do this, we first give a sketch of~\cite{lokshtanov2020exponential} in Section~\ref{sec:sketch}, and we then give a sketch of our algorithm in Section~\ref{sec:methods}. 
The two algorithms in~\cite{lokshtanov2020exponential} and~\cite{schrijver1994finding} (and ours) consist of the following two steps. Here, $h_1,h_2$ and $h_3$ are functions on $n$ and $k$. 
\begin{enumerate}
    \item Enumerate $h_1$ weak $T$-linkages of $G$ one of which is homotopic to a $T$-linkage of $G$ in $h_2$ time. 
    \item Assuming that we have a \emph{correct} weak $T$-linkage $\mathcal W$, we compute a $T$-linkage of $G$ in $h_3$ time.\footnote{In fact, we apply this algorithm for all weak $T$-linkages we have computed (since we do not know which one is a correct weak $T$-linkage), and we can obtain a $T$-linkage of $G$ when we apply this algorithm to a correct weak $T$-linkage.}
\end{enumerate}

In the case of~\cite{schrijver1994finding},  the bounds of $h_1, h_2$ and $h_3$ are $n^{O(\sqrt{k})}$, $n^{O(\sqrt{k})}$, and $O(n^6)$, respectively.\footnote{But this algorithm also works for directed graphs.} 
To deal with the second step, the algorithm of~\cite{lokshtanov2020exponential} uses the algorithm in~\cite{schrijver1994finding}, 
and thus they have the same bound on $h_3$. But using clever ideas, they significantly improves the bounds of $h_1$ and $h_2$ to  $2^{O({k^2})}$ and  $2^{O({k^2})}n+2^{O(k)}n^2$, respectively. 

We also use the same high level structure. But we improve the bounds of $h_1, h_2$ and $h_3$ to $2^{O(k^2)}$, $2^{O(k^2)}n$, and $2^{O(k)}n$, respectively. Therefore, 
the total running time of our algorithm is $2^{O(k^2)}n$. 

\subsection{Sketch of~\cite{lokshtanov2020exponential} and Obstacles in Designing a $2^{O(k^2)}n$-time Algorithm}\label{sec:sketch}
We first give a brief sketch of the $2^{O(k^2)}n^6$-time algorithm for the \textsf{Planar Disjoint Paths} problem
as we also use several tools used in~\cite{lokshtanov2020exponential}.
At the end of this subsection, we give three obstacles in improving this algorithm to run in $2^{O(k^2)}n$ time.
Then we briefly show how we overcome each of the obstacles in this paper. 

\subsubsection{Enumerating Homology Classes}
The first step is to enumerate $2^{O(k^2)}$ homology classes. Before doing that, they remove several \emph{irrelevant} vertices from $G$
as a preprocessing step. 
Given an instance $(G,T,k)$ of \textsf{Planar Disjoint Paths}, 
this algorithm first removes several \emph{irrelevant} vertices from $G$ so that
the answer of $(G,T,k)$ remains the same, and the treewidth of $G$ decreases to $2^{O(k)}$. 
Indeed, it is a common ingredient of all previous algorithms for 
\textsf{Planar Disjoint Paths}, which is called the \emph{irrelevant vertex technique}. 
More specifically, a vertex $v$ of $G$ is said to be \emph{irrelevant} for $(G,T,k)$ if 
$(G,T,k)$ is a \textsf{YES}-instance if and only if $(G-v, T, k)$ is a \textsf{YES}-instance,
where $G-v$ is the graph obtained from $G$ by removing $v$ and its incident edges. 
Roberson and Seymour~\cite{robertson1995graph} showed
that a planar graph having treewidth $g(k)$
has an irrelevant vertex for a specific function $g(k)$. For this, they showed that 
a planar graph with treewidth $O(w)$ has a $w\times w$ grid as a minor. 
If $w\geq g(k)$ for a specific function $g(k)$, there is a grid minor of $G$ of size $g(k)\times g(k)$ 
such that the innermost vertex of the grid minor is irrelevant. 
Later, the bound of $g(k)$ was improved to $O(k^{1.5}2^k)$ by Adler et al.~\cite{adler2017irrelevant}. 
Therefore, by repeatedly removing irrelevant vertices from $G$, we can obtain an equivalent instance $(G',T,k)$ where the treewidth of $G'$ is $O(g(k))$. 
A single irrelevant vertex can be found in linear time, and thus a naive implementation of this technique yields a quadratic-time algorithm.

\begin{figure}
    \centering
    \includegraphics[width=0.7\textwidth]{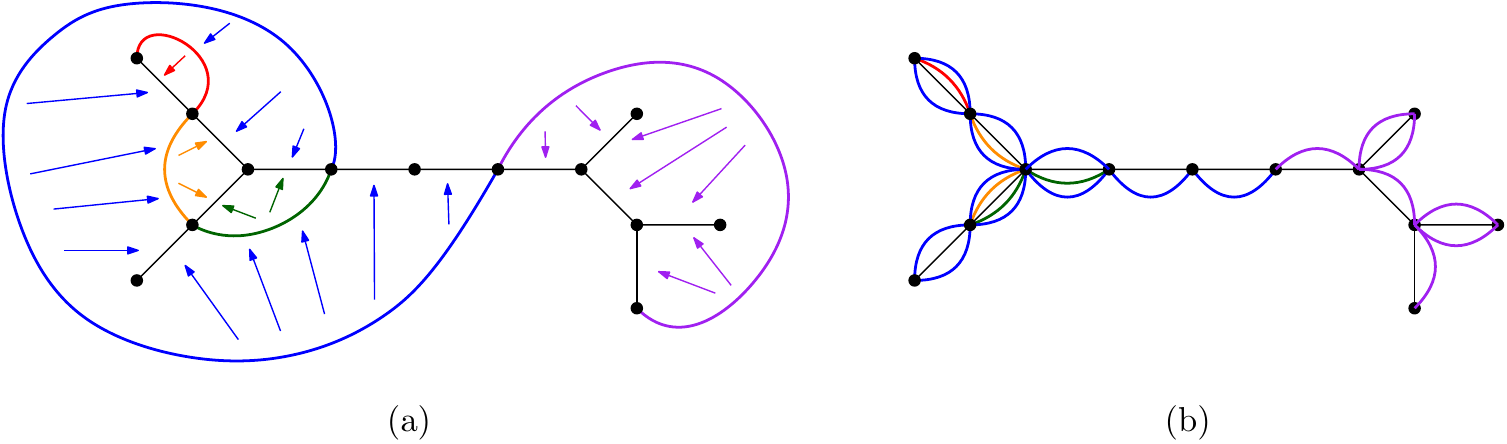}
    \caption{\small After pushing the paths to the tree in (a),
    we can obtain the the walks in (b). 
    }
    \label{fig:steiner}
\end{figure}

\medskip
Now we assume that $G$ has treewidth $2^{O(k)}$. Using this, Lokshtanov et al.~\cite{lokshtanov2020exponential} enumerated $2^{O(k^2)}$ weak $T$-linkages one of which
is discretely homotopic to a $T$-linkage of $G$. 
To do this, they used a Steiner tree $R$ of $\bar{T}$ in $G$, which they call the \emph{backbone tree}, having a certain property.
The backbone tree has $O(k)$ non-degree-2 vertices and $O(k)$ maximal paths consisting of degree-2 vertices only. 
A maximal degree-2 path is said to be \emph{long} if its length is at least $2^{10k}$, and it is said to be \emph{short}, otherwise. 
Then the total complexity of the short maximal degree-2 paths is $2^{O(k)}$. 
They showed that 
for every long maximal degree-2 path $\pi$ of $R$, 
there are two nooses $S_v$ and $S_u$ of complexity $2^{O(k)}$ such that
$\rdist(u,u')$ and $\rdist(v,v')$ is $2^{O(k)}$ for the vertices $v'$ and $u'$ of $S_v$ and $S_u$, respectively, intersected by $\pi$,
where $u$ and $v$ are the endpoints of $\pi$. 
This property holds due to the treewidth of $G$ is $2^{O(k)}$. 
In the following, we call the subpath of $\pi$ lying between $u$ and $u'$ (and between $v$ and $v'$)
the \emph{prefixes} of $\pi$. The subpath of $\pi$ excluding its prefixes is called the \emph{middle subpath} of $\pi$.

If $(G,T,k)$ is a \textsf{YES}-instance of \textsf{Planar Disjoint Paths}, 
there exists a $T$-linkage $\mathcal  P$ of $G$ which can be pushed onto the edges of $R$ so that the resulting walks on $R$
traverse each edge of $R$ $2^{O(k)}$ times. See Figure~\ref{fig:steiner}. 
More specifically, Lokshtanov et al.~\cite{lokshtanov2020exponential} gave a canonical way to  push 
the paths of $\mathcal P$ onto the edges of $R$. Then $\mathcal P$ becomes a weak $T$-linkage of $G$ whose edges lie on the backbone tree $R$. 
For this purpose, they cut each path of $\mathcal P$ into several \emph{segments} at the vertices intersected by the short maximal degree-2 paths 
and by the prefixes of the long maximal degree-2 paths. 
The number of segments is $2^{O(k)}$, and each segment $s$ forms a walk $\omega_s$ on $R$ when it is pushed onto $R$. 
Each edge of $R$ appears exactly once on $\omega_s$ unless $s$ intersects the middle path of a long maximal degree-2 path. 
Also, by a clever choice of $\mathcal P$, they showed that 
$\omega_s$ also traverses a single edge $2^{O(k)}$ times even in the case that $s$ intersects the middle path of a long maximal degree-2 path. 

This shows that, if $(G,T,k)$ is a \textsf{YES}-instance of \textsf{Planar Disjoint Paths}, 
there is a weak $T$-linkage $\mathcal W$ of $G$ homotopic to a $T$-linkage of $G$ 
such that the edges of the walks of $\mathcal W$ lie on $R$, and each walk of $\mathcal W$ traverses a single edge $2^{O(k)}$ times in total. 
Using this, they enumerate all weak $T$-linkages on $R$ such that each weak $T$-linkage traverses a single edge $2^{O(k)}$ times. 
The number of all such weak $T$-linkages is $2^{O(k^2)}$. 

In summary, they first compute a Steiner tree $R$ with a certain property in $2^{O(k)}n^2$ time, and then enumerate all weak $T$-linkages
on $R$ in $2^{O(k^2)}n$ time. 

\subsubsection{Finding a Solution Using a Given Weak Linkage}\label{sec:overview-homology}
For each weak $T$-linkage computed from the previous stage, they apply the algorithm by Schrijver~\cite{schrijver1994finding}. 
Given a weak $T$-linkage $\mathcal W$ of $G$, Schrijver~\cite{schrijver1994finding} presented an $O(n^6)$-time algorithm 
for computing a $T$-linkage $\mathcal P$ of $G$ such that $\psi_{\mathcal W}$ is homologous to $\psi_{\mathcal P}$, 
where $\phi_{\mathcal W}$ and $\phi_{\mathcal P}$ are flow functions representing $\mathcal W$ and $\mathcal P$, respectively. 
Due to Lemma~\ref{lem:homology-homotopy}, this algorithm always returns a $T$-linkage $\mathcal P'$ of $G$ 
if there is a $T$-linkage $\mathcal P$ discretely homotopic to a given weak $T$-linkage $\mathcal W$. 

Although Lokshtanov et al.~\cite{lokshtanov2020exponential} used this algorithm as a black box,
we are required to open up this black box to 
analyze the dependence on $n$ in the running time of~\cite{lokshtanov2020exponential}. 
  The goal here is to compute a flow function $\psi$ homologous to a given flow function $\phi$ such that 
  $\psi$ represents a $T$-\emph{linkage} of $G$. 
  To control the number of edges with $\psi(\cdot)\neq \epsilon$ incident to a common vertex, 
  Schrijver~\cite{schrijver1994finding} introduced a more general problem setting: the \emph{homology feasibility problem}. 
  
  In the homology feasiblity problem, we are given a flow function $\phi$, and a candidate function $\Gamma: {\Pi} \to 2^{\Sigma^*}$, 
  where $\Pi$ is a set of face-edge paths. 
  For a flow function $\psi$ and a face-edge path $\pi$, 
  we let $\psi(\pi)$ denote the product of $\psi(e)^{\mu(e)}$'s for all edges $e$ in $\pi$,
  where  the sign $\mu(e)$ is positive if $L_e$ is the face lying previous to $e$ in $\pi$, and negative, otherwise.  
  Then the goal is to compute a flow function $\psi$ homologous to a given flow function $\phi$ such that 
  $\psi(\pi)\in \Gamma(\pi)$. 
  In this case, the homology function $f$ for $\phi$ and $\psi$ is called a \emph{feasible homology function}. 
  Schrijver showed that a flow function $\psi$ with $\psi(\pi)\in \Gamma(\pi)$ for all $\pi\in\Pi$ corresponds to a $T$-linkage
  if \textsf{(i)} $\Pi$ is the set of all face-edge paths each consisting of faces and edges sharing a common vertex, 
  and \textsf{(ii)} $\Gamma(\pi)=\Sigma\cup\Sigma^{-1}$ for all $\pi\in \Pi$. 
  Therefore, it suffices to present an algorithm for the homology feasibility problem.  
  
  \medskip
  To solve the homology feasibility problem with $\Gamma$, the algorithm of~\cite{schrijver1994finding} first computes several \emph{initial}  functions $f$'s, computes a specific function $\bar{f}$ for each initial function $f$, and then takes the \emph{join} of $\bar{f}$'s.  
  Schrijver showed how to choose 
  such initial functions $f$'s so that the join of $\bar{f}$'s is a feasible homology function. 
  Given an initial function $f$, Schrijver showed how to compute $\bar{f}$ in $O(n\eta |\Pi|)$ time, where $\eta$ denotes the maximum length of $\bar{f}$ for all initial functions $f$.  
  Moreover, the initial functions can be computed in  time quadratic in the number $\chi$ of face-edge paths $\pi$ of $\Pi$ with $\phi(\pi)\notin \Gamma(\pi)$, and the number of initial functions is linear in $\chi$. 
  Therefore, the total running time is $O(\chi^2 + n\cdot \eta \cdot \chi\cdot |\Pi|)$.
  
  In the case of the weak $T$-linkages obtained from~\cite{lokshtanov2020exponential},
  $|\Pi|$ and $\chi$ are $\Theta(n^2)$ since the number of vertices of $R$ is $O(n)$ and
  the degree (in $G$) of each vertex of $R$ is $\Theta(n)$ in the worst case.
  Also, $\eta$ is $\Theta(n)$ in the worst case since the maximum radial distance between the outer face of $G$ and
  the vertices of $R$ is $\Theta(n)$ in the worst case.  
  Therefore, the total running time of~\cite{lokshtanov2020exponential} is $2^{O(k^2)}n^6$.

  \subsubsection{Obstacles in Designing a $2^{O(k^2)}n$-time Algorithm and Our Methods}
  There are three places where a super-linear dependence on $n$ appears: 
  the application of the irrelevant vertex technique takes $2^{O(k)}n^2$ time, 
  the computation of the backbone tree takes $O(n^2)$ time, and the algorithm of~\cite{schrijver1994finding} takes $O(n^6)$ time. 
  Although a fine-grained analysis of the algorithm of~\cite{schrijver1994finding} gives the bound of 
  $O(\chi^2 + n\cdot \eta \cdot \chi\cdot |\Pi|)$, the values of $\chi, \eta$ and $|\Pi|$ are too large in the case of the weak $T$-linkages $\mathcal W$ 
  constructed from~\cite{lokshtanov2020exponential}. 
  This is because the number of edges of $R$ is $\Theta(n)$, and the maximum radial distance from the outer face to
  the vertices of $R$ is $\Theta(n)$ in the worst case. 
  
  One possible direction for obtaining a faster algorithm for the \textsf{Planar Disjoint Paths} problem is 
  to construct weak $T$-linkages with respect to some other backbone structure of complexity of $2^{O(k)}$ instead of constructing them with respect to the backbone tree. Then (after modifying $G$ slightly), we can show that this decreases $\chi$ to $2^{O(k)}$, and thus this leads to an algorithm running in $2^{O(k^2)}n^5$ time. 
  This was our starting point. 

  \medskip
  \paragraph{New Backbone Structures.}
  We first construct a backbone structure (a subgraph of $\radgraph$) of complexity of $2^{O(k)}$ in $2^{O(k)}n$ time, and then show that all paths of $\mathcal P$ can be pushed onto the backbone structure (and several precomputed paths) for a $T$-linkage $\mathcal P$. In this paper, we call this backbone structure the \emph{frames} and \emph{skeleton forests}, which will be defined later. 
  In this way, each edge on the backbone structure is traversed by the resulting walks $2^{O(k)}$ times in total.
  Also, each edge on the precomputed paths is traversed by the resulting walks at most once in total. 
  This allows us to only consider the weak $T$-linkages using each edge of the backbone structure $2^{O(k)}$ times and each edge of the precomputed paths exactly once. 
  There are $2^{O(k^2)}$ such weak $T$-linkages $\mathcal W$, and we can enumerate each of them in $2^{O(k)}n$ time. 
  By slightly modifying $G$ with respect to a given weak $T$-linkage, we can make each vertex of the backbone structure 
  have degree $2^{O(k)}$ (without changing the answer of the instance). 
  Then the number $\chi$ of face-edge paths $\pi$ of $\Pi$ with $\phi(\pi)\notin \Gamma(\pi)$ decreases to $2^{O(k)}$. 
  In this way, we can overcome the second obstacle: we can construct the backbone structure in $2^{O(k)}n$ time (although we define the backbone structure in a different way.)
  Also, we can partially overcome the third obstacle: now the algorithm of~\cite{schrijver1994finding} takes $O(n^5)$ time in our case.

  \medskip 
  
  \paragraph{Maintaining a Data Structure.}
  There remain obstacles to be handled to obtain a $2^{O(k^2)}n$ time algorithm. 
  First, due to properties of the backbone structure we defined, we can show that 
  the length of $\bar{f}(F)$ of all initial functions $f$ constructed from~\cite{schrijver1994finding} is $2^{O(k)}$.
  But even in our case, $|\Pi|$ can be $2^{\Theta(k)}n$ as the total number of edges of the walks of $\mathcal W$ can be $2^{\Theta(k)}n$. 
  We resolve this issue by designing a suitable data structure and implementing the algorithm of~\cite{schrijver1994finding}
  using this data structure. 
  In this way, we can fully overcome the third obstacle: Given a \emph{correct} weak $T$-linkage of $G$ defined with respect to the backbone structure (frames and skeleton forests), 
  we can compute a $T$-linkage of $G$ in $2^{O(k)}n$ time. 

  \medskip 
  
  \paragraph{Irrelevant Vertex Technique.}
  The remaining obstacle is to remove irrelevant vertices so that $G$ has treewidth $2^{O(k)}$ in $2^{O(k)}n$ time in total.
  We overcome this obstacle by removing irrelevant vertices while subdividing $G$ into several pieces. 
  This idea was already used in~\cite{reed1995rooted}, but this work focused on designing $2^{O(k^2)}n$ time algorithm for the
  \textsf{Planar Disjoint Paths} problem. We apply the algorithm~\cite{reed1995rooted} first, and then we are given
  $2^{O(k)}$ pieces (subgraphs of $G$) enclosed by a closed curve such that no terminal is contained in the \emph{interior} of each piece.  
  Then we remove all vertices in each piece whose radial distance is at least $2^{ck}$ from the boundary vertices of the piece  for a specific constant $c$ in time linear in the complexity of the piece. We show that all such vertices are irrelevant for $(G,T,k)$, and
  the resulting graph has treewidth $2^{O(k)}$.

\subsection{Our Methods}\label{sec:methods}
In this subsection, we give a more detailed sketch of our algorithm. 
We first apply the irrelevant technique in $2^{O(k)}n$ time 
so that the resulting graph $G$ has treewidth $2^{O(k)}$. This algorithm is described in Section~\ref{sec:irrelevant}. 
In the following, we assume that the treewidth of $G$ is $2^{O(k)}$. 
A key idea of our algorithm is summarized in the following lemma. The definitions of frames, forests and reference paths
will be given in Section~\ref{sec:intro-skeleton} and Section~\ref{sec:cutting}. 
Section~\ref{sec:cutting}, Section~\ref{sec:week_linkages} and Section~\ref{sec:weak_linkage_construction} 
are devoted to prove the following lemma. 

\begin{lemma}\label{lem:enumeration}
We can enumerate $2^{O(k^2)}$ weak linkages $\mathcal W$ in $2^{O(k^2)}n$ time  one of which is discretely homotopic to
a $T$-linkage of $G$ such that
\begin{itemize}
    \item the edges of the walks of $\mathcal W$ lie on the \emph{frames}, \emph{skeleton forests}, and \emph{reference paths}, 
    \item each edge of the frames and skeleton forests is traversed by the walks of $\mathcal W$ $2^{O(k)}$ times, and
    \item each edge of the the reference paths is traversed by the walks of $\mathcal W$ at most once. 
\end{itemize}
\end{lemma}

\subsubsection{Frames, Skeleton Forests, and Reference Paths}\label{sec:intro-skeleton}
To prove Lemma~\ref{lem:enumeration}, we first define the \emph{frames} and \emph{skeleton forests} in Section~\ref{sec:cutting}. 
A frame is a noose of complexity $2^{O(k)}$, and we construct $O(k)$ \emph{concentric} frames. 
They subdivide $G$ into $O(k)$ \emph{framed rings} each of which is a subgraph of $G$ lying between two consecutive frames. 
A framed ring containing a terminal and a framed ring not containing a terminal appear alternatively. For illustration, see Figure~\ref{fig:frames}. 

\paragraph{Framed Ring Containing a Terminal.}
For a framed ring containing a terminal, the radial distance between its two boundary frames is $2^{O(k)}$.
This allows us to encode the interaction between $\mathcal P$ and a framed ring $\textsf{Ring}^\textsf{o}$ containing a terminal
in a compact way. 
%
To do this, we cut $\ring^\textsf{o}$ along the radial curves from each terminal to the outer boundary of $\ring^\textsf{o}$ 
so that the terminals lie on the boundary of the resulting subgraph.
This approach is also used in~\cite{erickson2011shortest} for computing shortest non-crossing walks in a planar graph. 
For illustration, see Figure~\ref{fig:poly_schema}. 
The union of the radial curves forms a forest, and we call it the \emph{skeleton forest} of $\ring^\textsf{o}$. 
In Section~\ref{sec:week_linkages}, we push the parts of the paths of $\mathcal P$ contained in $\textsf{Ring}^\textsf{o}$
onto the skeleton forest to obtain the \emph{canonical $T$-linkage}. 
Then as in~\cite{erickson2011shortest,lokshtanov2020exponential}, we can encode the interaction between the canonical weak $T$-linkage
and $\ring^\textsf{o}$ using $O(k\cdot k_i)$ bits, where $k_i$ denotes the number of terminals contained in $\ring^\textsf{o}$. 

\paragraph{Framed Ring Containing No Terminal.}
On the other hand, for a framed ring $\ring^\textsf{x}$ 
not containing a terminal, the radial distance between its two boundary frames can be as large as $\Theta(n)$. 
However, in this case, the number of maximal subpaths $\pi$ of the paths of $\mathcal P$ contained in $\ring^\textsf{x}$ is $2^{O(k)}$. 
For a maximum-cardinality set $\mathcal Q$ of disjoint paths in $\ring^\textsf{x}$ between the two boundary frames, 
the absolute value of the \emph{winding number} between $\pi$ and a path of $\mathcal Q$ is $2^{O(k)}$.\footnote{More precisely, there is a $T$-linkage, which we call the \emph{base} $T$-linkage, satisfying this property.}  
Using this observation, we precompute a maximum-cardinality set $\mathcal Q$ of disjoint paths between the two boundary frames.
We call the paths of $\mathcal Q$ the \emph{reference paths}. 
Then in Section~\ref{sec:week_linkages}, we 
push the parts of the paths of $\mathcal P$ contained in $\textsf{Ring}^\textsf{x}$ onto the frames and the paths of $\mathcal Q$.\footnote{In  fact, the story is more complicated here. It is not easy to apply the two techniques simultaneously (the technique for bounding the winding number between $\pi$ and a path of $\mathcal Q$, 
and the technique for pushing the paths onto the frames and paths of $\mathcal Q$). To handle this issue, we define subframes, which are cycles lying close to the frames, and then define $\mathcal Q$ as a maxmimum-cardinality set of paths connecting two subframes.}
We show that the resulting walks use each edge of the frames $2^{O(k)}$ times, and use each edge of the paths of $\mathcal Q$ at most once.
Then we can encode the interaction between the canonical weak $T$-linkage 
and $\ring^\textsf{x}$ using $O(k)$ bits.

\medskip
In summary, only $O(k^2)$ bits are required to encode the interaction between the framed rings and the canonical weak $T$-linkage.
Therefore, we can enumerate $2^{O(k^2)}$ weak linkages one of which is the canonical weak $T$-linkage as described in Section~\ref{sec:weak_linkage_construction}. 
Note that the canonical $T$-linkage satisfies the properties stated in Lemma~\ref{lem:enumeration}. 
Using these properties, we can show that given the canonical weak $T$-linkage, we can compute a $T$-linkage of $G$ in $O(n)$ time
in Section~\ref{sec:reconstructing}.

\subsubsection{Finding a Solution Using a Given Weak Linkage}\label{sec:intro-homology}
  Given a weak $T$-linkage $\mathcal W$ of $G$ satisfying the properties stated in Lemma~\ref{lem:enumeration}, 
  we first compute the flow function $\phi$ representing $\mathcal W$. Then 
  we let $\Pi$ be the set of all face-edge paths each consisting of faces and edges sharing a common vertex, 
  and let $\Gamma: \Pi\to \Sigma^*$ be the function with $\Gamma(\pi)=\Sigma\cup\Sigma^{-1}$ for all $\pi\in \Pi$. 
  By applying the algorithm in~\cite{schrijver1994finding}, we can compute a flow function $\psi$ 
  homologous to $\phi$ with $\psi(\pi)\in \Gamma(\pi)$.
  By the argument of~\cite{schrijver1994finding}, $\psi$ represents a unique $T$-linkage $\mathcal P$ of $G$, and thus we can obtain 
  $\mathcal P$ from $\psi$ in linear time. 
  
  As mentioned in Section~\ref{sec:overview-homology}, we can show that 
  the running time of the algorithm in~\cite{schrijver1994finding} is  $O(\chi^2 + n\cdot \eta \cdot \chi\cdot |\Pi|)$.
  Here, $\chi$ is the number of face-edge paths $\pi$ of $\Pi$ with $\phi(\pi)\notin \Gamma(\pi)$, and 
  $\eta$ is the maximum length of $\bar{f}$ for all \emph{initial} functions $f$.\footnote{The definition of the initial functions is given in Section~\ref{sec:homology_algorithm}.}
  The value of $\chi$ is $2^{O(k)}\cdot \textsf{MaxDeg}$, where $\textsf{MaxDeg}$ denotes
   the maximum degree in $G$ of the vertices on the frames and skeleton forests. 
  At this point, $\textsf{MaxDeg}$ can be as large as $\Theta(n)$. 
  
  \paragraph{Modification of $G$ to $\gmod$.}
  To handle this issue, we modify $G$ slightly so that $\textsf{MaxDeg}$ decreases to $2^{O(k)}$ in Section~\ref{sec:reducing_violation}. 
  We say an edge $e$ is \emph{linkage-edge} if a walk of $\mathcal W$ uses $e$. 
 For a vertex $v$ of $G$, let $d_W(v)$ be the number of linkage-edges 
 incident to $v$. Also, let $d_G(v)$ be the degree of $v$ on $G$.  
 We say an edge $e$ lies \emph{between $e_1$ and $e_2$} if $e$, $e_1$, and $e_2$ are incident to a common vertex $v$, and $e$ lies from $e_1$ to $e_2$ in clockwise direction around $v$. 
 Two linkage-edges $e$ and $e'$ incident to $v$ form a \emph{wedge at $v$} if there is no other linkage-edge between $e$ and $e'$. 
 Furthermore, the wedge contains all edges between $e$ and $e'$, and all wedges at $v$ are pairwise interior-disjoint. 
 See Figure~\ref{fig:gmod}. 
 Clearly, there are at most $d_W(v)$ wedges at a vertex $v$. We say a wedge at $v$ is \emph{empty} if 
 no edge of $G$ incident to $v$ is contained in the wedge. 
 
 For each non-empty wedge at a vertex $v$ in the frames and skeleton forests, we insert a new vertex $v'$ to $G$, 
 and add the edge between $v$ and $v'$. 
 Then we remove the edges in the wedge, and reconnect them to $v'$ instead of $v$. Note that, the edges incident to a new vertex $v'$ are not linkage-edges. 
 We do this for all non-empty wedges and all vertices in the frames and skeleton forests. 
 This takes $O(n)$ time in total since the total degree of all vertices in the frames and skeleton forests is $O(n)$. 
 Clearly, there is a one-to-one correspondence between two linkages (and weak linkage) on $G$ and on $\gmod$.  
 Due to the modification, the running time of~\cite{schrijver1994finding} is improved to $2^{O(k)}\cdot n\cdot \eta \cdot |\Pi|$.
 
 \paragraph{Further Improvement.}
 We improve further the running time of~\cite{schrijver1994finding} to $2^{O(k)}\cdot n\cdot \eta \cdot |\Pi|$ in Section~\ref{sec:pre_feasible}. 
 First, using a suitable data structure, we improve the running time to $2^{O(k)} n\cdot \eta$.
 More specifically, the algorithm of~\cite{schrijver1994finding} computes $2^{O(k)}$ initial functions $f$, and then computes a function $\bar{f}$, which is called 
 a \emph{smallest pre-feasible function}, for each initial function $f$. 
 The procedure for computing $\bar{f}$ consists of $O(n\eta)$ iterations, and in each iteration, we update the value of $f(F)$ for 
 a specific face $F$. 
 Here, each iteration takes $O(|\Pi|)$ time because we are required  to find a face-edge path $\pi \in \Pi$ with 
  $f(F)^{-1}\cdot \phi(\pi)\cdot \phi(F') \notin \Gamma(\pi)$ and $f(F)\neq\epsilon$, 
  where $F$ and $F'$ are the end faces of $\pi$. 
  Then we update the value of $f(F)$ or $f(F')$ accordingly. Instead of considering all face-edge paths in $\Pi$,
  we can find such a face-edge path $\pi$ in amortized constant time by maintaining a suitable data structure. 
  This improves the running time to $2^{O(k)} n\cdot \eta$.
  
  Finally, we show that $\eta=2^{O(k)}$ by introducing $O(k)$ new constraints to $\Pi$, one per frame.   
  Before doing that, we first observe that 
  for any weak $T$-linkage $\mathcal W$ represented by $f$,
  there is a weak $T$-linkage represented by $f'$ discretely homotopic to $\mathcal W$, 
  where $f'$ is a function obtained from $f$ during the process of computing $\bar{f}$. 
  This is simply because each iteration of the process is indeed a face operation.  
  Moreover, due to the new constraints, the 
  part of $\mathcal W$ contained in a framed ring $\textsf{Ring}$ is discretely homotopic to 
  the part of $\mathcal W'$ contained in ring $\textsf{Ring}$. 
  Using this property, we show that $\eta$ is $2^{O(k)}$ in the proof of Lemma~\ref{lem:length}. 
  
 \medskip 
 
 In summary, we first apply the irrelevant technique to $G$ in $2^{O(k)}n$ time 
 so that $G$ has treewidth of $2^{O(k)}$ in Section~\ref{sec:irrelevant}.
 We construct the frames, skeleton forests, and reference paths in $2^{O(k)}n$ time in Section~\ref{sec:cutting} and Section~\ref{sec:week_linkages}. 
 Using them, we enumerate $2^{O(k^2)}$ weak $T$-linkages satisfying the properties stated in Lemma~\ref{lem:enumeration}
 in $2^{O(k^2)}n$ time.
 Then for each weak $T$-linkage, we apply the algorithm of~\cite{schrijver1994finding} after modifying $G$ slightly
 in $2^{O(k)}n$ time in total as described in Section~\ref{sec:reconstructing}.
 In this way, we can solve the \textsf{Planar Disjoint Paths} problem in $2^{O(k^2)}n$ time.

\section{Irrelevant Vertex Technique}\label{sec:irrelevant}
In this section, we remove sufficiently many irrelevant vertices from $G$ to reduce the treewidth of $G$ to $2^{O(k)}$.
Using this algorithm as a preprocessing step, we may assume that $G$ has treewidth  $2^{O(k)}$
in the following sections. 
A \emph{$c$-punctured plane} $\boxdot$ is the region obtained by removing $c$ open holes from the plane. 
The boundary of $\boxdot$ is the union of the  boundaries of the open holes of $\boxdot$ and a vertex of $G$ lying on the boundary of $\boxdot$ is called a boundary vertex. 
Let $G$ be a plane graph, and $T$ be a set of terminal pairs. We consider a point where each terminal lies as a (trivial) hole. In this way,
we can obtain a $2k$-punctured plane $\boxdot$ 
such that the terminals of $T$ lie on the boundary of $\boxdot$. 

We say vertex-disjoint cycles $C_1,\ldots,C_\ell$ of $G$  are \emph{concentric} if $C_i$ is contained in the interior of $C_{i+1}$ for all $i\in[\ell-1]$. 
A cycle $C$ \emph{separates} a vertex  $u$ and a vertex set $W$
if $u$ is contained in the interior (or exterior) of $C$, and $W$ is contained in the exterior (or interior) of $C$.
Furthermore, a vertex $v$ is said to be $\ell$-\emph{isolated} if there are $\ell$ concentric cycles each separating $v$ and $\bar{T}$.

\begin{lemma}[Lemmas 1 and 10 in~\cite{adler2017irrelevant}]\label{lem:isolate_irr}
There exists a function $g(k)\in O(k^{1.5}2^k)$ which every $g(k)$-isolated vertex in $G$ is irrelevant for the instance $(G,T,k)$ of \textsf{Planar Disjoint Paths}. 
\end{lemma}


In the following, let $g(k)$ be the function in $O(k^{1.5}2^k)$. 
Our strategy is to decompose $G$ into
$2^{O(k)}$ \emph{nice} subgraphs $H$ each embedded on a 1-punctured plane after removing several $g(k)$-isolated vertices from $G$. 
A  subgraph $H$ of $G$ embedded on a $c$-punctured plane is said to be \emph{nice} if  
any cycle $C$ of $H$ separating a vertex $v$ of $H$ and the boundary vertices of $H$ also separates $v$ and $\bar{T}$ in $G$. 
Then a vertex in $H$ separated from the boundary vertices of $H$ by $g(k)$ concentric cycles is $g(k)$-isolated in $G$. Thus we are allowed to delete such vertices from $H$. 

\begin{figure}
    \centering
    \includegraphics[width=0.9\textwidth]{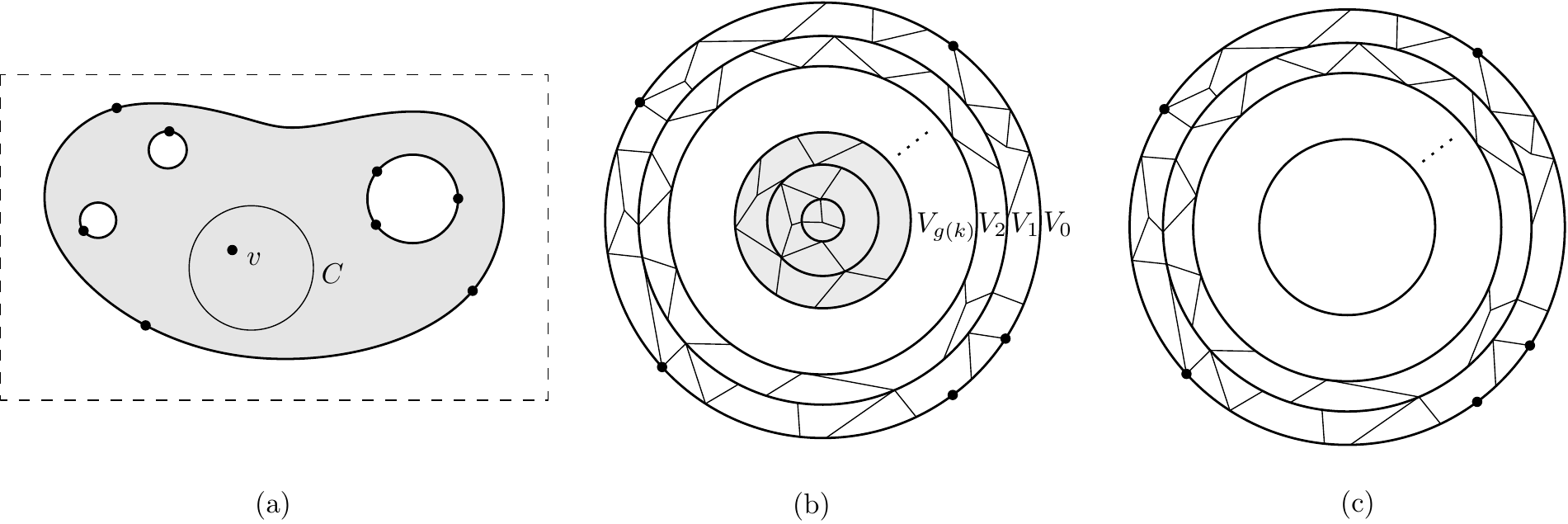}
    \caption{\small (a) 
    The white areas are open holes, and the points of their boundaries are terminals.
    The cycle $C$ separates $v$ from the boundary vertices, and it also separates $v$ from the terminals.
    (b) The gray area has vertices of $V_i$ for $i = g(k)+1, \ldots, n'$. 
    These vertices $g(k)$-isolated. 
    (c) A nice subgraph embedded on a $1$-punctured plane $\boxdot$ after removing all $g(k)$-isolated vertices has treeewidth at most $g(k)$.
    }
    \label{fig:nice_subgraph}
\end{figure}

We first show how to remove irrelevant vertices from a nice  subgraph $H$ embedded on a 1-punctured plane in Section~\ref{sec:disk_embedding}. 
Then we show how to decompose $G$ into nice subgraphs each embedded on a 1-punctured plane in Section~\ref{sec:cut=reduction}. 

\subsection{A Nice Subgraph Embedded on a 1-Punctured Plane}\label{sec:disk_embedding}
We are given a nice subgraph $H$ of $G$
embedded on a 1-punctured plane $\boxdot$.
Then we show how to remove $g(k)$-isolated vertices from $H$ so that $H$ has treewidth $2^{O(k)}$.
Our algorithm is simple: we remove all vertices sufficiently far from the boundary vertices of $H$ \emph{with respect to the radial distance}. 
Let $V_i$ be the set of vertices $v$ of $H$ such that the minimum radial distance between $v$ and a boundary vertex of $H$ is exactly $i$. 
That is, $V_0$ is the set of boundary vertices of $H$, and 
$V_{i+1}$ is the set of vertices not in $V_i\cup V_{i-1}$ lying on a face incident to a vertex of $V_i$. 
We can partition $V(H)$ into $V_0,\ldots V_{n'}$ in this way in time linear in the complexity of $H$
using a \emph{doubly connected edge list} of $H$, which a data structure for representing an embedding of a planar graph~\cite{CGbook}. 





\begin{lemma}\label{lem:disk_case_isolated}
A vertex in $V_i$ is $(i-1)$-isolated for $i\in [n']$.
\end{lemma}
\begin{proof}
We prove the lemma inductively. 
For the base case that $i=1$,
the lemma holds immediately since
every non-terminal vertex is $0$-isolated.


We assume that the lemma holds for vertices in $\cup_{j\leq i}V_j$, and consider a vertex $v$ in $V_{i+1}$ for $i\geq 1$. 
Note that no edge of $H$ connects two vertices each in $V_0$ and $V_j$, respectively, for any $j\geq 2$. Moreover, every path in $H$ connecting a boundary vertex of $V_0$ and $v$ intersects $V_1$. This means there exists a cycle $C$ in $H[V_1]$ which separates $V_0$ and $v$ in $H$. 
Imagine that we remove from $H$ all vertices separated by $C$ from $v$. Then the resulting graph $H'$ can be considered as a plane graph embedded on
a 1-punctured plane such that the vertices of $V(C)$
are the boundary vertices of $H'$. Therefore, there are $(i-1)$ concentric cycles
separating $v$ and $V_1$ in $H'$ (and thus in $H$) by the induction hypothesis. Thus, $v$ is $i$-isolated. \end{proof}

\begin{lemma}\label{lem:disk_case_treewidth}
    For a vertex $v$ in $V_i$,
    no $i$ concentric cycles separating $v$ and $V_0$ exist
    for $i\in [n']$.
\end{lemma}
\begin{proof}
We prove the lemma using the induction on the minimum radial distance $i$ between $v$ and the boundary vertices of $H$. 
For the base case, we consider a vertex $v$ in $V_1$. 
There is a face incident to both a boundary vertex $w$ of $H$ and $v$, and thus no cycle separates $w$ and $v$. Thus, the lemma holds.

We assume that there exists a set $\mathcal C$ of $(i+1)$ concentric cycles separating $v$ and the boundary vertices of $G$.  
Let $C$ be the cycle farthest from $v$ among them. 
As we did before, imagine that we remove from $G$ all vertices separated by $C$ from $v$. Then the resulting graph $H'$ can be considered as a plane graph embedded on 
a 1-punctured plane such that the vertices of $V(C)$
are the boundary vertices of $H'$. 
By the induction hypothesis, no $i$ concentric cycles separating $v$ and the boundary vertices of $H'$ exist 
if the radial distance between $v$ and $V(C)$ is $i$. Therefore, 
the radial distance between $v$ and $V(C)$ is less than $i$, and thus at least two cycles of $\mathcal C$ intersect $V_1$. 
This means there is a cycle which separates some vertex in $V_1$ and $V_0$. This contradicts to the definition of $V_1$.  
Thus, such a set $\mathcal C$ does not exists, and thus this completes the proof.
\end{proof}

Using the two lemmas stated above,
we 
remove all vertices in $V_i$ for $i=g(k)+1,\ldots,n'$ from $H$ in time linear in the complexity of $H$ in total.
By Lemma~\ref{lem:disk_case_isolated},
all removed vertices are $g(k)$-isolated. 
After removing them, 
$H$ has no $(g(k)+1)$ concentric cycles by Lemma~\ref{lem:disk_case_treewidth}. 
A planar graph without $(g(k)+1)$ concentric cycles 
has treewidth at most $g(k)$. Therefore, we have the following lemma. 

\begin{lemma}\label{thm:disk_case}
Given a nice subgraph $H$ of $G$ embedded on a 1-punctured plane, 
 we can remove $g(k)$-isolated vertices so that the resulting  subgraph $H$ has treewidth at most $g(k)$.
\end{lemma}

\subsection{General Case: Cut Reduction}\label{sec:cut=reduction}
In this section, we are given a graph $G$ embedded on a $2k$-punctured plane, and show how to remove $g(k)$-isolated vertices  so that the resulting graph has treewidth at most $g(k)$. 
 We apply the \emph{cut reduction} introduced in~\cite{reed1995rooted} to decompose $G$ into several nice subgraphs embedded on 1-punctured planes. Then we use the algorithm described in Section~\ref{sec:disk_embedding} to each nice subgraph embedded on a 1-punctured plane. 


\begin{lemma}[Lemma 2 in~\cite{reed1995rooted}]\label{lem:shcism}
Let $H$ be a nice subgraph of $G$ embedded on a $c$-punctured plane $\boxdot$. 
If $c\geq 3$, we can compute both a non-crossing closed curve $J$ contained in $\boxdot$ and a set $X$ of $g(k)$-isolated vertices in time linear in the complexity of $H$ such that 
\begin{itemize}
    \item {$J$ intersects $G$ only at vertices of $H$ and does not cross any edge of $G$,}
    \item{the number of vertices of $H-X$ intersected by $J$ is at most $6g(k)+6$, and}
    \item 
    $\boxdot\setminus J$ has at most three connected components each of which has less than $c$ holes. 
\end{itemize}
\end{lemma}

Lemma~2 of~\cite{reed1995rooted} shows how to decompose a nice subgraph $H$ of $G$ embedded on a $c$-punctured plane with $c\geq 3$, but 
similarly, we can do this in the case of $c= 2$. For details, see its proof of the following lemma. 

\begin{lemma}\label{lem:cut-reduction}
Let $G$ be a plane graph and $T$ be a set of  terminal pairs. After removing $g(k)$-isolated vertices from $G$, we can decompose $G$ into $2^{O(k)}$ nice subgraphs each embedded on a  1-punctured plane and each having $2^{O(k)}$ boudnary vertices  in linear time in total. 
\end{lemma}
\begin{proof}
First, by applying Lemma~\ref{lem:shcism} repeatedly, we decompose $G$ embedded on a $2k$-punctured plane into at most $3^{2k}$ 
nice subgraphs of $G$ each embedded on a $1$- or $2$-punctured plane. 
For each nice subgraph $H$ of $G$ embedded on a 2-punctured plane $\boxdot$, we do the following. 
Let $C_1$ and $C_2$ be the boundaries of 
two holes of $\boxdot$. 
Consider a simple arc contained in $\boxdot$ having one endpoint on $V(C_1)$ and one endpoint on $V(C_2)$ that intersects $G$ only at vertices of $H$. 
If there are more than one such arcs, we choose an arc $A$ that minimizes $|V(A)|$, where
$V(A)$ denotes the set of vertices of $G$ (and thus of $H$) intersected by $J$. Then let $v$ be the median vertex of  $V(A)$. 

If the size of $V(A)$ is at most $6g(k)+6$, then 
we cut $\boxdot$ 
along the non-crossing closed  curve $C_1\cup A \cup C_2$ into a 1-punctured plane $\boxdot'$,
and then we can consider $H$ as a nice subgraph of $G$
embedded on $\boxdot'$ and having $|V(C_1)\cup V(C_2)\cup V(A)|=2^{O(k)}$ boundary vertices. 
Otherwise, 
there are $g(k)$ concentric cycles separating $v$ and $V(C_1)\cup V(C_2)$, and thus $v$ is $g(k)$-isolated. 
Let $X$ be the set of the vertices lying on the ones, except for the most $g(k)$ cycles farthest from $v$, among these cycles. 
All vertices of $X$ are $g(k)$-isolated
since $H$ is a nice subgraph of $G$. Moreover, we can compute them in time linear in the complexity of $H$. 
The number of vertices in $H-X$ intersected by $A$ is $2g(k)$ by construction. We remove all vertices of $X$ from $H$, and then cut $\boxdot$ along the  closed non-crossing curve $C_1\cup A \cup C_2$ into a 1-punctured plane $\boxdot'$ as we did for the previous case. 
Then $H$ is a nice subgraph of $G$ embedded on $\boxdot'$ having at most $2g(k)=2^{O(k)}$ boundary vertices. 
\end{proof}



\begin{figure}
    \centering
    \includegraphics[width=0.8
    \textwidth]{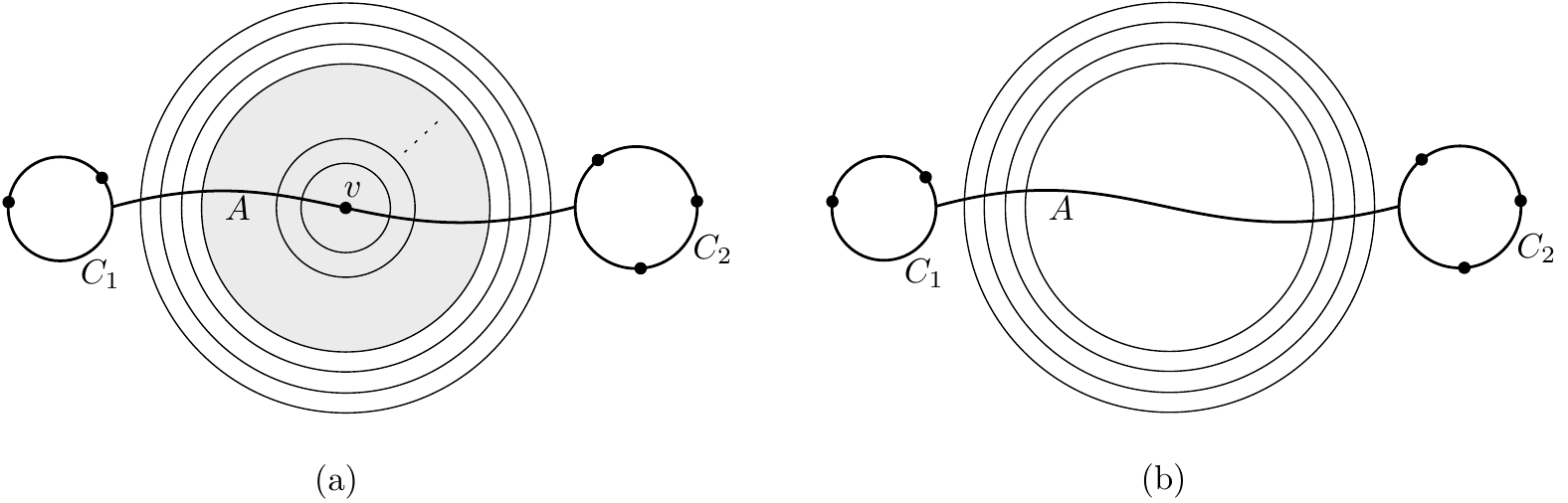}
    \caption{\small (a) 
    $C_1$ and $C_2$ are boundaries of two holes of $\boxdot$ and $A$ is an arc having two endpoints on $V(C_1)$ and $V(C_2)$.
    $v$ is a median vertex of $V(A)$.
    Concentric cycles except the outermost $g(k)$ cycles are colored gray.
    These vertices are $g(k)$-isolated.
    (b) $g(k)$-isolated vertices are removed from \ccheck{(a)}.
    $V(C_1)\cup V(C_2)\cup V(A)$ has $2^{O(k)}$ boundary vertices.
    }
    \label{fig:cut_reduction}
\end{figure}

After apply Lemma~\ref{lem:cut-reduction},
we are given $2^{O(k)}$ nice subgraphs each
embedded on a  1-punctured plane and each having $2^{O(k)}$ boundary vertices.
For each nice subgraph, we apply the procedure for removing irrelevant vertices described in Section~\ref{sec:disk_embedding}.
Then each of the resulting subgraphs has
treewidth $2^{O(k)}$.
This implies that $G$ has treewidth $2^{O(k)}$.
That is, we can construct a tree decomposition $(\mathcal T,\beta)$ of $G$ of width $2^{O(k)}$ from tree decompositions of the nice subgraphs. More specifically,  
let $B$ be the set of all boundary vertices of all nice subgraphs. 
For each tree decomposition $(\mathcal T_H,\beta_H)$ of width $2^{O(k)}$ of a nice subgraph $H$, 
we add all vertices of $B$ to all bags of $\beta_H(\cdot)$. 
For the root $t$ of $\mathcal T$, we let $\beta(t)=B$, and then 
let $t$ point to the roots of $\mathcal T_H$'s  in $\mathcal T$ so that $t$ has the roots as its children. 
In this way, we can obtain a tree decomposition $(\mathcal T, \beta)$ of $G$ of width $2^{O(k)}$. 
Therefore, we have the following theorem. Furthermore, this bound is tight by Theorem~3 in~\cite{adler2019lower}.

\begin{theorem}
Given an instance  $(G,T,k)$ of \textsf{Planar Disjoint Paths}, we can remove irrelevant vertices from $G$ in  $2^{O(k)}n$ time in total
so that the resulting graph $G$ has treewidth of $2^{O(k)}$.
\end{theorem}

\section{Cutting the Plane into $O(k)$ Framed Rings}\label{sec:cutting}
For two concentric cycles $I$ and $I'$ such that $I$ is contained in the interior of $I'$, we let 
$\textsf{Ring}(I,I')$ be the subgraph of $G$
contained in the exterior of $I$ and the interior of $I'$ (including $I$ and $I'$). 
A ring $\ring(I,I')$ is said to be \emph{terminal-free} if no terminal is a vertex of $\ring(I,I')$. 
Let $\mathcal P$ be a $T$-linkage of $G$.  
A \emph{segment} of $\mathcal P$ in $\ring(I,I')$ is a maximal subpath of a path $P$ of $\mathcal P$ contained in $\textsf{Ring}(I,I')$. 
Note that both endpoints of a segment of $\mathcal P$ lie on $V(I)\cup V(I')$ 
if $\textsf{Ring}(I,I')$ is terminal-free. 
A segment in a terminal-free ring 
is called a \emph{traversing} segment if 
it has one endpoint in $V(I)$ and one endpoint in $V(I')$. 
A segment of $\mathcal P$ is called a \emph{visitor}, otherwise.
That is, a visitor has both endpoints in $V(I)$ 
    or has both endpoints in $V(I')$. 
    We call a visitor an \emph{inner visitor} if its endpoints lie in $V(I)$, and an \emph{outer visitor}, otherwise. 
    See Figure~\ref{fig:layers}(c).

\medskip 
In this section, we subdivide the plane with respect to a set $\mathcal C$ of $O(k)$  concentric cycles of $G$, which will be called the \emph{subframes}, such that there is a $T$-linkage $\mathcal P$  satisfying the following: For any two consecutive cycles $C$ and $C'$ of $\mathcal C$ such that
$\ring(C,C')$ is terminal-free, 
a segment of $\mathcal P$ in $\textsf{Ring}(C,C')$ is a traversing segment. Furthermore, the number of the segments of $\mathcal P$ in $\textsf{Ring}(C,C')$ is $2^{O(k)}$.
However, the total complexity of the subframes might be large. To handle this issue, 
we compute two curves $B$ and $B'$ of total complexity $2^{O(k)}$ in $\textsf{Ring}(C,C')$, which we will call the \emph{frames}, so that 
the radial distance between $C$ and $B$ (and $C'$ and $B'$) is $2^{O(k)}$. 
For illustration, see Figure~\ref{fig:frames}.

\subsection{Rings and Tight Concentric Cycles from the Outer Face}\label{sec:rings}
In this subsection, we construct a sequence $\mathcal I=\langle I_1,\ldots, I_{n'} \rangle$ of concentric 
cycles containing $F^*$ in their interiors in a \emph{breadth-first} fashion,
which will be used to define the frames and subframes in the following subsections,
where
$F^*$ denotes an innermost face of $G$, that is, $G$ is a face farthest from the outer face with respect to the radial distance. 
We can compute $F^*$ in linear time using a variant of the breadth-first search on $G$. 

To do this, we first decompose $V$ into $V_0,\ldots, V_{n'}$ such that $V_i$ consists of all vertices of $G$ whose radial distance from $F^*$
is exactly $i$. 
For each index $i\in[n']$, 
consider the subgraph of $G$ induced by $V_0\cup V_1\cup \ldots V_i$. 
The subgraph is not necessarily connected, but it has a unique outer face bouned by a single cycle.  
We let $I_i$ be the boundary (cycle) of the outer face of the subgraph. See Figure~\ref{fig:layers}(a). 
Let $\mathcal I=\langle I_1,\ldots, I_{n'}\rangle$. 
We say a terminal-free ring $\textsf{Ring}(I_i,I_j)$ a \emph{maximal terminal-free ring} if 
$\textsf{Ring}(I_{i'},I_{j'})$ contains a terminal for any two indices $i',j'$ with 
$[i,j]\subsetneq [i',j']$.
Also, a ring $\textsf{Ring}(I_i,I_j)$ is said to be \emph{thick} if  $|i-j|>2^{10k}$. 


\begin{figure}
    \centering
    \includegraphics[width=0.9\textwidth]{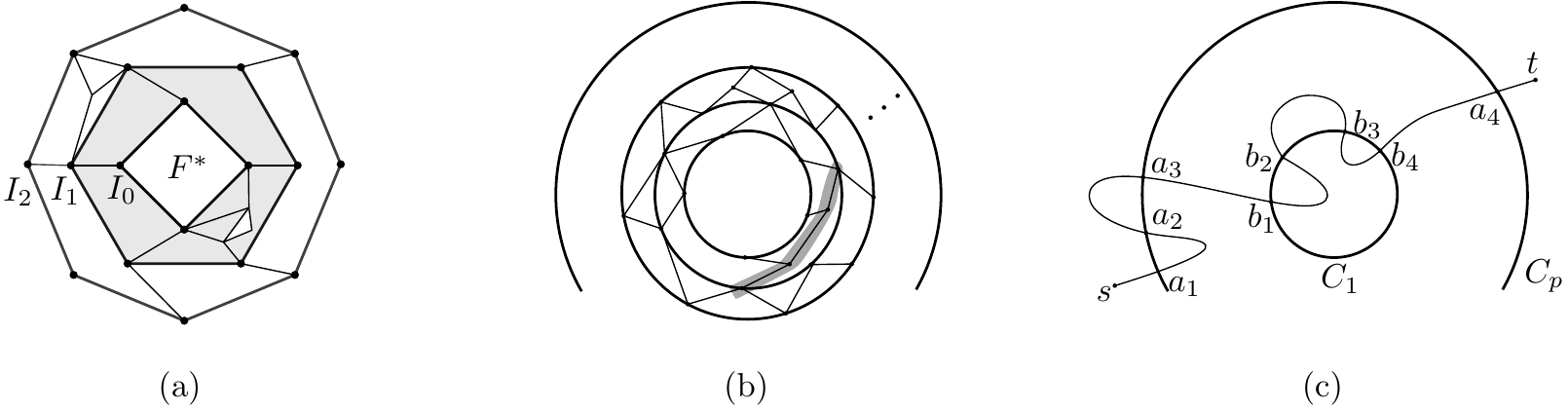}
    \caption{\small (a) 
    The vertices of $I_0, I_1$ and $I_2$ are marked with disks.
    (b) This sequence of concentric cycles is not tight. The witness of non-tightness is the path colored gray. The path violates the third condition of the tightness. 
    (c) The path in the figure that connecting $s$ and $t$ has
    four segments in $\textsf{Ring}(C_1,C_p)$: the outer visitor connects $a_1$ and $a_2$, the inner visitor connects $b_2$ and $b_3$, and the two traversing segments connecting $b_1$ and $a_3$, and $a_4$ and $b_4$, respectively.      }
    \label{fig:layers}
\end{figure}

For each thick maximal terminal-free ring  $\textsf{Ring}(I_i,I_j)$, we compute two 
cycles $C$ and $C'$ in $\textsf{Ring}(I_i,I_j)$, which will be called the subframes,
such that every segment of $\mathcal P$ in $\textsf{Ring}(I_i,I_j)$
intersects $C$ and $C'$ exactly once 
for a specific $T$-linkage $\mathcal P$. 
For this purpose, we define a \emph{tight} sequence of concentric cycles in 
a thick maximal terminal-free ring $\textsf{Ring}(I_i,I_j)$. 
Recall that $G$ is drawn in the plane, and thus a cycle forms a closed curve in the plane. 
For a cycle $C$ of $G$, we use $\textsf{cl}(C)$ to denote the region in the plane enclosed by $C$ including $C$ itself. 
Also, we use $\textsf{int}(C)$ to denote the open region
enclosed by $C$ excluding $C$.  
We say a sequence $\mathcal C=\langle C_1,\ldots, C_{p}\rangle$ 
of concentric cycles in $\textsf{Ring}(I_i,I_j)$ is \emph{tight} if
\begin{itemize}
    \item $C_1 = I_i$, $C_p= I_j$, 
    \item the radial distance between $C_1$ and $C_p$ is exactly $p-1$, and 
    \item $\textsf{cl}(C_{i+1})\setminus \textsf{cl}(C_i)$ does not contain any cycle $C$ 
with $\textsf{cl}(C_i) \subsetneq  \textsf{cl}(C) \subsetneq \textsf{cl}(C_{i+1})$
for $i\in[1, 2^{5k}]$, 
    \item $\textsf{int}(C_{i+1})\setminus \textsf{int}(C_i)$ does not contain any cycle $C$ 
with $\textsf{cl}(C_i) \subsetneq  \textsf{cl}(C) \subsetneq \textsf{cl}(C_{i+1})$
for $i\in[p-2^{5k},p]$.
\end{itemize}

For illustration, see Figure~\ref{fig:layers}(b). 
Note that the third and fourth conditions are symmetric.

\begin{lemma}
    For each thick maximal terminal-free ring, there is a tight sequence of concentric cycles in the ring.
    Moreover, it can be constructed in time linear in the complexity of the ring. 
\end{lemma}
\begin{proof}
    We construct tight concentric cycles by using a breadth-first search through faces in two opposite directions:
    outwards from $C_1=I_i$ and inwards from $C_p=I_j$. 
    For the traversal going outwards from $C_1$, we  construct the concentric cycles starting from $C_1$: 
    collect all the faces incident to $C_i$ but not visited so far, and 
    let $C_{i+1}$ be the outer boundary cycle of the union of these faces. 
    For the inward step, we can choose the cycles symmetrically from $C_p$. 
    The process terminates when all faces are visited. 

    Now we show that the cycles satisfy the above conditions. 
    For $i\in[1, 2^{5k}]$, we assume that $\textsf{cl}(C_{i+1})\setminus \textsf{cl}(C_i)$ contains a cycle $C$ 
    with $\textsf{cl}(C_i) \subsetneq  \textsf{cl}(C) \subsetneq \textsf{cl}(C_{i+1})$. 
    Then $\textsf{cl}(C_{i+1}) \setminus \textsf{cl}(C)$ and $\textsf{cl}(C_i)$ is separated by the inclusive relation. 
    Thus, the radial distance between a vertex $v \in \textsf{cl}(C_{i+1}) \setminus \textsf{cl}(C)$ and $C_i$ is larger than one. 
    However, all vertices of $C_{i+1}$ are derived from the faces incident to $C_i$, so there exists a vertex $u\in C_i$ with $\textsf{rdist}(u, v) = 1$.
    This makes a contradiction, thus there is no vertex $v \in \textsf{cl}(C_{i+1}) \setminus \textsf{cl}(C)$.
    This means that there is no cycle $C$ with $\textsf{cl}(C_i) \subsetneq  \textsf{cl}(C) \subsetneq \textsf{cl}(C_{i+1})$. 
    For $i \in [p-2^{5k}, p]$, we can define $C_i$ symmetrically so that satisfies the fourth condition.
    
    For the outward traversal, a vertex $v \in C_{i}$ is visited in the $(i-1)$ steps of the traversal. 
    Thus, a vertex $v \in C_{i}$ has the radial distance $(i-1)$ from $C_1$. 
    Symmetrically, any vertex $v$ of a cycle $C_{i}$ constructed from the inward traversal 
    has the radial distance $(p-i)$ from $C_p$. 
    Also, there is a face incident to the two last concentric cycles, one from the inward traversal and one from outward traversal.  
    Therefore, the radial distance between the two last concentric cycles is one, and 
    the radial distance between $C_1$ and $C_p$ is $p-1$. 
    
    Clearly, this process takes time linear in the complextiy of $\ring(I_i,I_j)$. 
\end{proof}


\subsection{Monotonicity of a $T$-Linkage in Each Thick Maximal Terminal-Free Ring}\label{sec:monotonicity}
We say a $T$-linkage $\mathcal P$ is $\mathcal{C}$-\emph{cheap} if 
it uses a smallest number of edges of $E(G)\setminus (\cup_{C\in \mathcal C} E(C))$
among all $T$-linkages. 
Let $\mathcal P$ be a $\mathcal{C}$-{cheap} $T$-linkage. 
In this section, we analyze the interaction between $\mathcal P$ and $\mathcal C$ as stated in Lemma~\ref{lem:visitors} and Corollary~\ref{lem:monotone}. 
To do this, we need the following lemmas. Here, for two indices $\alpha$ and $\beta$ in $[p]$,  
a \emph{cut} for $V(C_\alpha)$ and $V(C_\beta)$ in a thick maximal terminal-free ring $\textsf{Ring}(I_i,I_j)$
is defined as a set $W$ of vertices of a thick maximal terminal-free ring $\textsf{Ring}(I_i,I_j)$ such that
the removal of $W$ from $\textsf{Ring}(I_i,I_j)$ disconnects 
every pair $(v_\alpha,v_\beta)$ of vertices with $v_\alpha\in C_\alpha$ and $v_\beta\in C_\beta$. 
Note that $W$ forms a \emph{noose}: a simple closed
curve in the plane intersecting $G$ only at the vertices of $W$. 
In the proof of the following lemma, we use the fact that the treewidth of $G$ (and thus the treewidth of $\textsf{Ring}(I_i,I_j)$) is $2^{O(k)}$.

\begin{figure}
    \centering
    \includegraphics[width=0.8\textwidth]{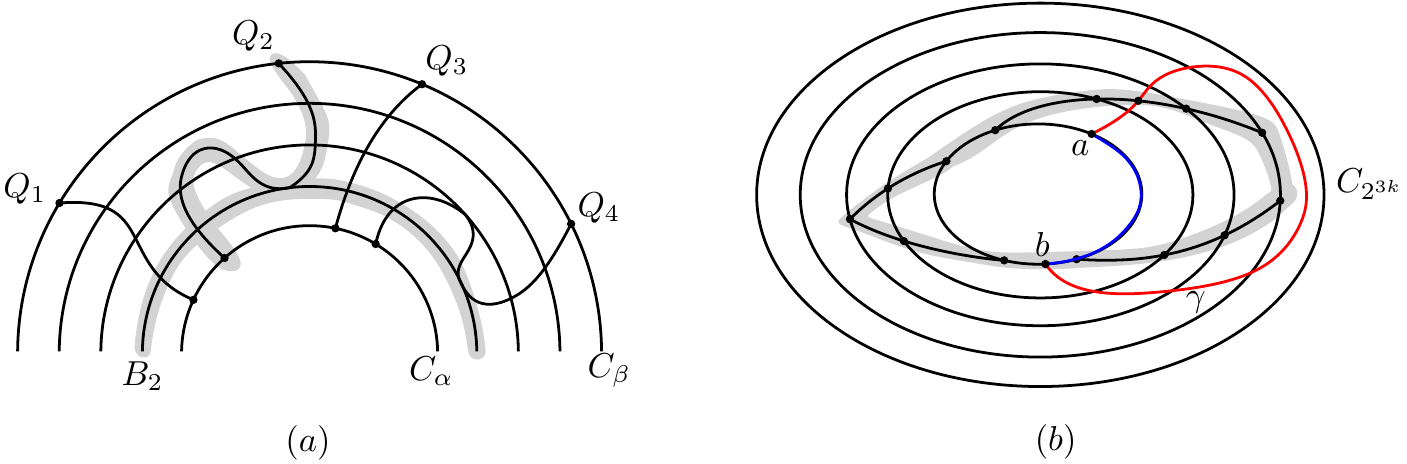}
    \caption{\small (a) $B_2$ is the union of $Q_2$ and $C_{\alpha+2}$, which is the subgraph of $G$ colored gray. (b) $S^o$ is the noose colored gray. It passes through at most $2^{3k}$ vertices of $G$. An inner visitor $\gamma$ (colored red) has two endpoints $a$ and $b$, and its base arc is  colored blue.}
    \label{fig:bramble}
\end{figure}
\begin{lemma}\label{lem:min-cut}
    A minimum cut for $V(C_\alpha)$ and $V(C_\beta)$ in $\textsf{Ring}(C_\alpha,C_\beta)$ has size at most $2^{3k}$
    if $|\alpha-\beta|> 2^{3k}$.
\end{lemma}
\begin{proof}
    Assume to the contrary that a minimum cut for $V(C_\alpha)$ and $V(C_\beta)$ in $\textsf{Ring}(C_\alpha,C_\beta)$ is larger than $2^{3k}$ for two indices $\alpha$ and $\beta$ with $|\alpha-\beta| >2^{3k}$.
    Let $\mathcal Q=\langle Q_1,\ldots, Q_\ell\rangle$ be a maximum-cardinality set of 
    internally vertex-disjoint paths between vertices of $C_\alpha$ and vertices of $C_\beta$. 
    By the Menger's theorem, the size of $\mathcal Q$ is equal to the size of a minimum cut for $V(C_\alpha)$ and $V(C_\beta)$, which is larger than $2^{3k}$ by the assumption.
    Then we can construct a \emph{bramble} of order $2^{3k}$ in $\textsf{Ring}(I_i,I_j)$ as follows.
    A bramble for an undirected graph $G$ is a family of connected vertex sets of $G$ 
    that all touch each other. 
    Two vertex sets $A$ and $B$ \emph{touch} if they have a vertex in common, or 
    a vertex of $A$ is adjacent to a vertex of $B$ in $G$. 
    For each index $s\in [1,2^{3k}]$, 
    we let $B_{s}$ be the union of $V(Q_s)$ and $V(C_{s+\alpha-1})$. 
    See Figure~\ref{fig:bramble}(a). 
    Note that $B_{s}$ is a connected vertex set since 
    each path of $\mathcal Q$ crosses each cycle in $\{C_\alpha,\ldots,C_\beta\}$ at least once. 
    Moreover, any two such vertex sets $B_{s}$ and $B_{s'}$ touch each other.
    To see this, notice that $Q_s$ intersects $C_{s'+\alpha-1}$, and thus $Q_s$ and $C_{s'+\alpha-1}$
    have a vertex in common. 
    Since $V(Q_s) \subseteq B_{s}$ and $V(C_{s'+\alpha-1}) \subseteq B_{s'}$ by construction, 
    $B_{s}$ touches $B_{s'}$. Therefore, the family $\mathcal B$ of $B_{s}$ for all indices 
     $s\in [1,2^{3k}]$ is a bramble.
    
     The \emph{order} of a bramble is defined as the smallest size of a hitting set of the elements of the bramble. To show that the order of $\mathcal B$ is $2^{3k-1}$, 
    observe that no three elements $B_s$, $B_{s'}$, and $B_{s''}$ in $\mathcal B$ 
    share a common vertex. 
    This observation holds because 
    $Q_s, Q_{s'}$ and $Q_{s''}$ are pairwise vertex-disjoint paths, and $C_{s+\alpha-1}, C_{s'+\alpha-1}$ and $C_{s''+\alpha-1}$ are pairwise vertex-disjoint cycles. 
    Therefore, any vertex of $G$ hits at most two elements of $\mathcal B$. 
    This implies that the order of $\mathcal B$ is at least the half of $|\mathcal B|$, which is $2^{3k-1}$. 
    
    It is known that $G$ has a bramble of order at least $k'$ if and only if it has treewidth at least $k'-1$~\cite{PBook}. In our case, $\mathcal B$ has order $2^{3k-1}$, which implies that
    $G$ has treewidth at least $2^{3k-1}-1$. However, due to the preprocessing we made in Section~\ref{sec:irrelevant},  $G$ has treewidth $O(k^{1.5}2^{k})$. 
    This makes a contradiction, and thus the size of $\mathcal Q$ is at most $2^{3k}$, and the size of a minimum cut for $V(C_\alpha)$ and $V(C_\beta)$ in $\textsf{Ring}(I_i,I_j)$ is at most $2^{3k}$.
\end{proof}

\begin{lemma}\label{lem:visitors} 
For any index $r\in (2^{4k}, p-2^{4k})$,
 no visitor of $\mathcal P$ in $\textsf{Ring}(I_i,I_j)$ intersects $C_r$.
\end{lemma} 
\begin{proof}
    By Lemma~\ref{lem:min-cut}, 
   a minimum cut $S$ for $V(C_1)$ and $V(C_{2^{3k}})$ has size
   at most $2^{3k}$. 
    Then there is a  \emph{noose} $S^{o}$
    with $S^{o}\cap V=S$.
    A noose is a closed simple curve in the plane intersecting 
     $G$ only at vertices. 
    Similarly, let $R$ be a minimum cut for 
     $V(C_{p-2^{3k}})$ and $V(C_{p})$ in  
    $\textsf{Ring}(I_i,I_j)$, and let $R^o$ be a noose with $R^o\cap V=R$. See Figure~\ref{fig:bramble}(b). 

We prove the lemma
    only for the inner visitor since the outer visitor can be proved symmetrically. (Note that the tightness of a sequence of concentric cycles is defined symmetrically
    for $C_1$ and $C_p$.) 
    For an inner visitor not intersecting $S^o$,
    the claim holds immediately since $S^o$ lies in the interior of $C_{2^{4k}}$. 
    The number of inner visitors intersecting $V(S^o)$ 
    is at most $2^{3k}$
    since the visitors are internally vertex-disjoint, and the size of $V(S^o)$ is $2^{3k}$. 
    
    
    For an inner visitor $\gamma$ intersecting $S^o$, 
    let $a$ and $b$ be its endpoints in $V(C_1)$. There are two arcs of $C_1$  connecting $a$ and $b$. Exactly one of them forms a bounded region not containing the interior of $C_1$ together with $\gamma$. See Figure~\ref{fig:bramble}(b). We call such an arc the \emph{base arc} of $\gamma$. 
    We define the \emph{order} of an inner visitor $\gamma$
    intersecting $S^o$ as the number of other inner visitors
    intersecting $S^o$ 
    whose base arcs are contained in the base arc of $\gamma$. 
    
    \begin{figure}
        \centering
        \includegraphics[width=0.8\textwidth]{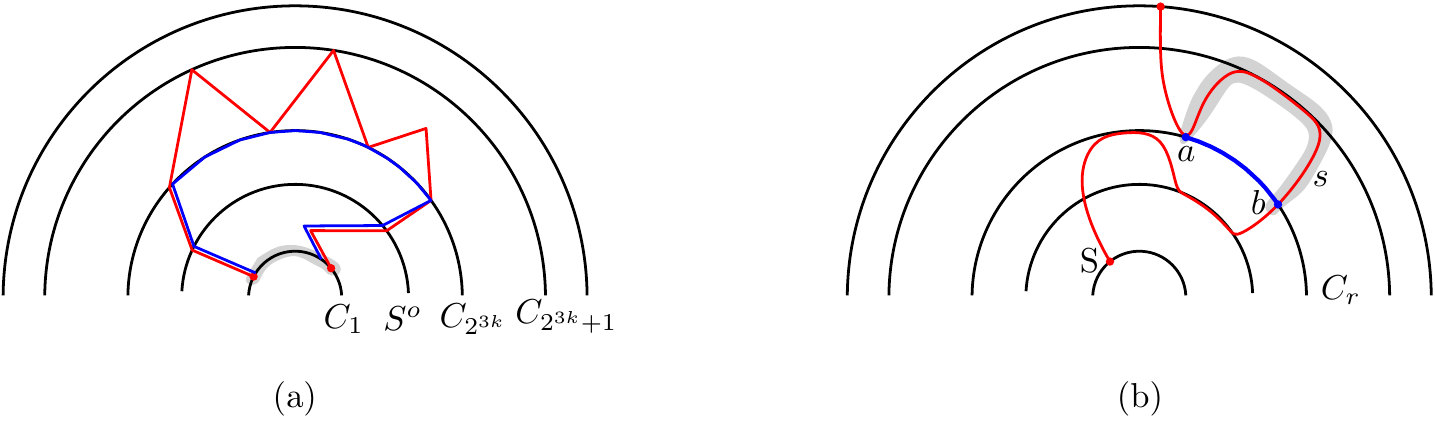}
        \caption{\small (a) The red path $\gamma$ is an inner visitor. If its order is zero, no segment of $\mathcal P$ intersecting $S^o$ has endpoints on the base of $\gamma$. Therefore, the blue path does not intersect any other segments of $\mathcal P$, and thus we can replace the red path with the blue path.
        (b) $S$ is a traversing segment. The subpath of $S$ lying between $a$ and $b$ is an $r$-horn. The base of the $r$-horn is colored blue.}
        \label{fig:rerouting}
    \end{figure}

    To prove the lemma,
    we show that an inner visitor intersecting $S^o$ of
    order $x$ does not intersect any of $C_{2^{3k}+x+1}, C_{2^{3k}+x+2},\ldots, C_p$. We prove this using induction on the order of the inner visitors. 
    Note that the order of an inner visitor is at most $2^{3k}$  since the number of visitors intersecting $S^o$ is at most $2^{3k}$. 
    As a base case, consider an inner visitor $\gamma$ of order $0$. 
    The region bounded by $\gamma$ and its base arc
    does not intersect by any other segments intersecting $S^o$.
    Moreover, a segment not intersecting $S^o$ does not intersect $C_{2^{3k}}$. 
    If $\gamma$ intersects $C_{2^{3k}+1}$, 
    we can reroute it so that it uses strictly fewer edges not contained in the cycles of $\mathcal C$:
    for every maximal part of $\gamma$ 
    lying outside of the interior of $C_{2^{3k}}$,
    we replace it with an arc of $C_{2^{3k}}$. See Figure~\ref{fig:rerouting}(a). 
    The resulting $\gamma$ does not intersect any other
    path in $\mathcal P$. 
    Therefore, the claim holds for the base case. 
    
    Now we assume that no inner visitor intersecting $S^o$ of order $x-1$ 
    intersects $C_{2^{3k}+x}$. Then let $\gamma$ be an inner
    visitor intersecting $S^o$ of order $x$. 
    If $\gamma$ intersects $C_{2^{3k}+x+1}$, 
    the region bounded by $\gamma$ and $C_{2^{3k}+x+1}$
    is not intersected  by any other visitors by the induction hypothesis. Similarly to the previous case,
    we can reroute it along arcs of $C_{2^{3k}+x}$ using fewer edges not contained in the cycles of $\mathcal C$ in this case. 
    This violates the fact that $\mathcal P$ is $\mathcal C$-cheap, and thus
    $\gamma$ does not intersect $C_{2^{3k}+x+1}$, and thus
    it does not intersect any of $C_{2^{3k}+x+1}, C_{2^{3k}+x+2},\ldots, C_p$.
    This proves the claim. 
    
    Since the order of an inner visitor is at most $2^{3k}$,
    no inner visitor intersects any of  $C_{2^{4k}+1} ,\ldots, C_p$. Symmetrically, 
    no outer visitor intersects any of 
    $C_{1}, C_{2},\ldots, C_{p-2^{4k}-1}$.
    Therefore, the lemma holds. 
\end{proof}

Now we focus on the traversing segments of $\mathcal P$ in $\ring(I_i,I_j)$.
We show that for a traversing segment $S$ of $\mathcal P$ and an index $r\in(2^{4k}, 2^{5k})\cup (p-2^{5k},p-2^{4k})$, the vertices of $V(S) \cap V(C_r)$ appear consecutively along $S$. 
    We only prove this for indices $r\in(2^{4k}, 2^{5k})$. The other case, that is, $r\in(p-2^{5k}, p-2^{4k})$, can be handled symmetrically. 
    To prove this claim, we define a \emph{horn} of a traversing segment $S$.
    A subpath $s$ of $S$ is called an $r$-\emph{horn} 
     if
     its endpoints lie on $C_r$, and no internal vertices of $s$ lie in $\textsf{cl}(C_r)$,
     where $\textsf{cl}(C_r)$ denotes the region in the plane bounded by $C_r$ including $C_r$. 
     See Figure~\ref{fig:rerouting}(b). 
    In this case, 
    the \emph{base} of $s$ is defined as 
    the subpath of $C_r$
    which, together with $s$, forms a bounded region lying outside of the interior of $C_r$. 
    For illustration of the following lemma, see Figure~\ref{fig:horn}(a). 
    
    \begin{figure}
        \centering
        \includegraphics[width=0.9\textwidth]{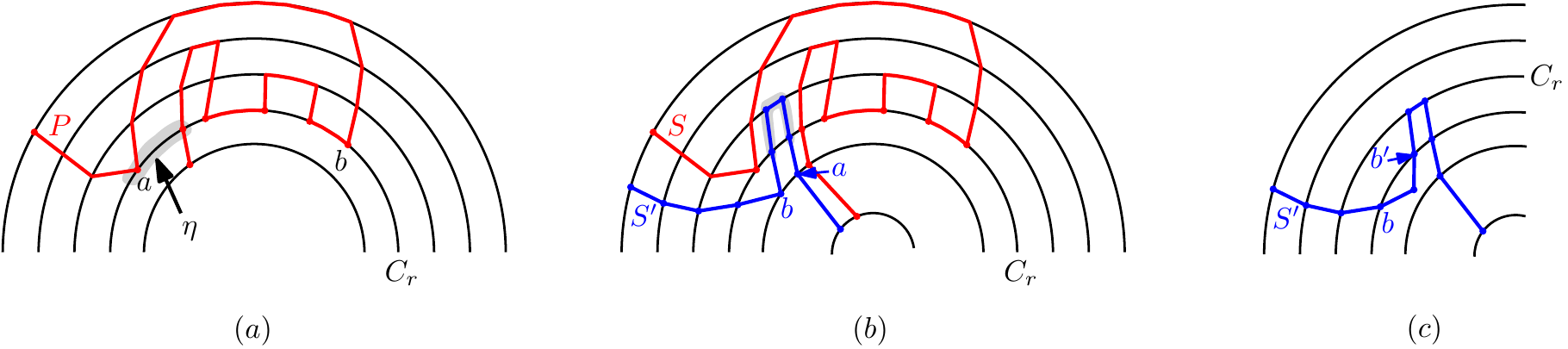}
        \caption{\small (a) The red segment $S$ is a traversing path of $\mathcal P$, and the subpath $s$ of $S$ between $a$ and $b$ is an $r$-horn. Then $\eta$ is the part of the base of $s$ neither intersecting $P$ nor contained in the bases of any other $r$-horns coming from $S$.
        (b) Since $S$ is $\mathcal C$-cheap, there is a segment $S'$ intersecting $\eta$. (c) If $b$ is contained in $V(C_r)$, the tightness of $\mathcal C$ is violated because of the subpath of $S'$ between $b$ and $b'$.}
        \label{fig:horn}
    \end{figure}

\begin{lemma}\label{lem:horn}
    Let $s$ be an $r$-horn of a traversing segment $S$ of $\mathcal P$ for 
 $r\in(2^{4k}, 2^{5k})$. Let $x$ and $y$ be the first and last
 vertices of $S$ on the base of $s$ (including the endpoints of $s$) along $S$ from the  end vertex of $S$ lying on $V(C_1)$. 
    Let $\eta$  be the subpath of the base of $s$ connecting $x$ and $y$. 
    Then there is   a traversing segment of $\mathcal P$ other than $S$
    which has an $(r-1)$-horn intersecting $\eta$. 
\end{lemma}
\begin{proof}
    Since $\mathcal P$ is $\mathcal C$-cheap,
    $\eta$ is intersected by a segment $S'$ of $\mathcal P$. 
    Otherwise, we can reroute $S$ along $\eta$ so that
    it uses strictly fewer edges not contained in the cycles of $\mathcal C$, which contradicts the fact that
    $\mathcal P$ is $\mathcal C$-cheap. 
    Moreover, $S'$ is a traversing segment since no visitor intersects $C_r$ by Lemma~\ref{lem:visitors}. 
    Since the paths of $\mathcal P$ are vertex-disjoint, and $S'$ intersects
    the bounded region $R$ formed by
    $\eta$ and the subpath of $S$ connecting $x$ and $y$, $S'$ must intersect $\eta$ at least twice. 
    Then consider a maximal subpath $s'$ of $S'$ contained in $R$ and intersecting 
     $\eta$ only at its endpoints. See Figure~\ref{fig:horn}(b). 
    Let $a$ be the first vertex of $S'$ lying before $s'$ (from the end vertex of $S'$ lying on $V(C_1)$) that lies on the cycles in $\mathcal C$. Similarly, let $b$ be the first vertex
    of $S'$ lying after $s'$ (from the end vertex of $S'$ lying on $V(C_1)$) that lies on the cycles in $\mathcal C$.
     Note that $a$ and $b$ are  vertices of $V(C_r)\cup V(C_{r-1})$ since $S'$ is a path in $G$ connecting a vertex in $C_1$ and a vertex in $C_p$. 
    Moreover, we claim that $a$ and $b$ are vertices of
    $C_{r-1}$. 
    If $b$ is not a vertex of $C_{r-1}$, then it is a vertex of $C_r$. 
    See Figure~\ref{fig:horn}(c). 
    Let $b'$ be the endpoint of $s'$ appearing before the other endpoint of $s'$ along $S'$ from $b$.  
    Then a path $\gamma$ connecting $b$ and $b'$ along $C_r$,
    together with
    the subpath $\gamma'$ of $S'$ between $b$ and $b'$, 
    forms a bounded region not intersecting the interior of $C_{r-1}$. 
    This violates the tightness of $\mathcal C$: By replacing $\gamma$ with $\gamma'$ from $C_r$, we can obtain a cycle $C$ 
    contained in $\textsf{cl}(C_{r})\setminus \textsf{cl}(C_{r-1})$ 
with $\textsf{cl}(C_r) \subsetneq  \textsf{cl}(C) \subsetneq \textsf{cl}(C_{r-1})$. 
Therefore, both $a$ and $b$ are vertices of $C_{r-1}$. 
    This implies that $S'$ has an $(r-1)$-horn
    containing $s'$ and having $a$ and $b$ as the endpoints of its base. Therefore, the lemma holds.
\end{proof}

\begin{lemma}\label{lem:no-horn}
    No $r$-horn of a traversing segment of $\mathcal P$ exists
    for $r\in (2^{5k}, 2^{6k})$.
\end{lemma}
\begin{proof}
Assume to the contrary that there is an $r$-horn $s_1$ of a traversing segment of $\mathcal P$ for $r\in (2^{5k}, 2^{6k})$.
Then we can define a sequence $\mathcal S=\langle s_1,\ldots, s_q\rangle$ of horns inductively as follows. 
Given an $r'$-horn $s_t$, we define $s_{t+1}$ as follows. 
Let $S$ be the segment of $\mathcal P$ containing $s_t$. 
Then we define $x,y$ and $\eta$ as defined in the statement of Lemma~\ref{lem:horn}: Let $x$ and $y$ be the first and last
 vertices of $S$ on the base of $s_t$ (including the endpoints of $s$) along $S$ from the end vertex of $S$ lying on $V(C_1)$. 
    Let $\eta$  be the subpath of the base of $s$ connecting $x$ and $y$. 
If there is an $r'$-horn (other than $s_t$) whose base is contained in $\eta$, we let $s_{t+1}$ be the one with a longest base among all such $r'$-horns. 
Otherwise, we let $s_{t+1}$ be the one 
with a longest base among all $(r'-1)$-horns intersecting  $\eta$. 
See Figure~\ref{fig:horn-seq}(a). 
Note that if $r'\geq 2^{4k}$, $s_{t+1}$ is well-defined
by Lemma~\ref{lem:horn}. Therefore, the size of $\mathcal S$ must be at least $2^{5k}-2^{4k}>2^{3k}$. 
By construction, the following properties hold.

\begin{itemize}
\item[(1)] No two consecutive horns in $\mathcal S$ come from the same segment of $\mathcal P$.  
\item[(2)] If an $r$-horn $s_t$ and $r'$-horn $s_{t'}$ come from  the same segment of $\mathcal P$ with $t<t'$, 
then $r> r'$. 
\end{itemize}


\begin{figure}
    \centering
    \includegraphics{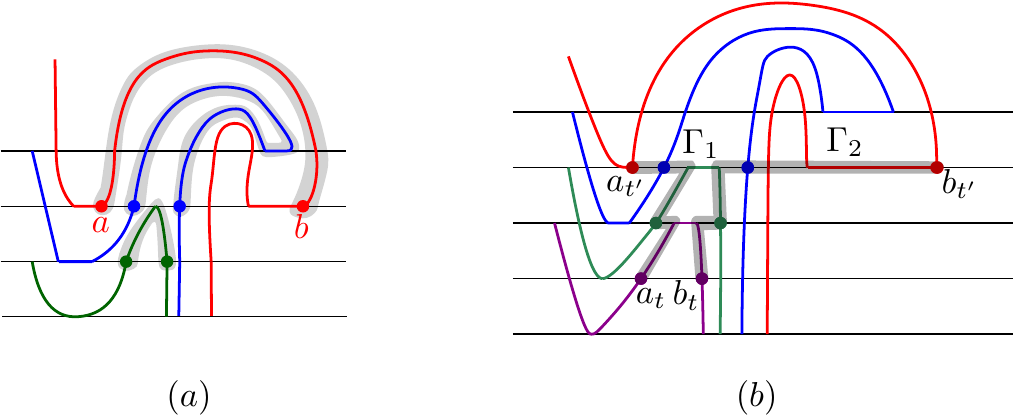}
    \caption{\small (a) The sequence of horns starting from the horn with endpoints $a$ and $b$
    consists of three horns colored gray. 
    (b) The left path colored gray is $\Gamma_1$ that connects $a_{t'}$ and $a_{t}$.
    The right path colored gray is $\Gamma_2$ that connects $b_{t'}$ and $b_{t}$.
    }
    \label{fig:horn-seq}
\end{figure}
In the following, we claim that 
the size of $\mathcal S$ is at most $2^{3k}$, which makes a contradiction. 
That is, this claim implies that no segment of $\mathcal P$ has a $r$-horn for $r\in(2^{5k}, 2^{6k})$. 
    If every horn in $\mathcal S$ comes from different segments of $\mathcal P$, the claim holds immediately. This is because the number of traversing segments of $\mathcal P$ is at most the size of a minimum cut for $C_1$ and $C_p$, which is at most $2^{3k}$. 
     
    Thus assume that there are two horns of $\mathcal S$
    coming from the same traversing segment of $\mathcal P$. 
    If there are more than two such pairs, we choose a pair 
    $(s_t,s_{t'})$ coming from the same traversing segment $S$ of $\mathcal P$  with $t<t'$ that minimizes $t'-t$. 
    Recall that $t'-t \geq 2$ as stated in Property~(1). 
    For an index $x\in[q]$, let $a_x$ and $b_x$ be the counterclockwise and clockwise endpoints of the base of $s_{x}$ along the cycle of $\mathcal C$ containing its base, respectively. 
    Let $\Gamma_1$ be the path of $G$ connecting
    $a_{t'}, a_{t'-1}, \ldots, a_t$ in order such that 
    $a_{x}$ and $a_{x-1}$ are connected by a part of the base of $s_x$ and 
    a subpath of $s_{x-1}$ for $x\in[t,t']$. 
    Similarly, let $\Gamma_2$ be the path of $G$ connecting 
    $b_t, b_{t+1},\ldots, b_{t'}$ in order such that 
    $b_{x}$ and $b_{x-1}$ are connected by a part of the base of $s_x$ and 
    a subpath of $s_{x-1}$ for $x\in[t,t']$. 
    See Figure. 
        
        
    
    
    Let $r$ be the index such that $s_t$ is an $r$-horn,
    and $r'$ be the index such that $s_{t'}$ is an $r'$-horn. 
    By Property~(2), $r>r'$,
    and therefore, 
    there is an $(r-1)$-horn $s_{t''}\in\mathcal S$ not coming from the segment $S$ but  intersecting the base of $s_t$ with $t<t''<t'$. 
    Let $S''$ be the traversing segment of $\mathcal P$ 
    containing $s_{t''}$. 
    In the following, we show that $S''$ has a horn in $\mathcal S$ lying between $s_t$ and $s_{t'}$.
    This contradicts the choice of $s_t$ and $s_{t'}$. 
    Since $S''$ connects a vertex of $V(C_1)$ and a vertex of  $V(C_p)$, 
    $S''$ intersects $\Gamma_1$ (and $\Gamma_2$) at a vertex lying outside of $s_{t''}$. 
    Consider a minimal subpath $\pi$ of $S''$ having one endpoint on $\Gamma_1$ and one endpoint on $\Gamma_2$. 
    Let $\ell$ be the index such that
    one endpoint of $\pi$ lies between $a_\ell$ and $a_{\ell+1}$ along $\Gamma_1$,
    and the other endpoint of $\pi$ lies between $b_\ell$ and $b_{\ell+1}$ along $\Gamma_2$.
    Note that the subscript $(\ell)$ of $a$ and $b$ must be the same since the paths of $\mathcal P$ are pairwise vertex-disjoint. 
    If $s_{\ell}$ comes from $S''$, then we are done. Thus we assume that $s_\ell$ comes
    from a segment other than $S''$. Then by the construction, $\pi$ must be included in $\mathcal S$
    before $s_{\ell+1}$ is included in $\mathcal S$. This makes a contradiction. 
    Therefore, the size of $\mathcal S$ is at most $2^{3k}$, and the lemma holds. 
\end{proof}

Lemma~\ref{lem:no-horn} implies the following corollary. 
\begin{corollary}\label{lem:monotone}
    The vertices of $V(S) \cap V(C_{r})$ appear consecutively along $S$
    for each traversing segment $S$ of $\mathcal P$ for each index $r\in(2^{4k}, 2^{5k})\cup (p-2^{5k},p-2^{4k})$. 
\end{corollary}

    

\subsection{Subframes and Frames for a Thick Maximal Terminal-Free Ring}
Now we are ready to define two frames and two subframes for each thick maximal terminal-free ring $\ring(I_i,I_j)$.
Let $\cup\mathcal C$ be the union of the tight sequence of concentric cycles for all thick terminal-free rings. 
We call a $\cup\mathcal C$-cheap linkage $\mathcal P$ a \emph{cheap} linkage. 
Since all terminal-free rings are pairwise disjoint, a cheap $T$-linkage $\mathcal P$ is also a $\mathcal C$-cheap $T$-linkage for the tight sequence $\mathcal C$ of concentric cycle 
we computed for a thick maximal terminal-free ring. 
Then Corollary~\ref{lem:monotone} implies the following corollary. 
We call $C_{2^{5k}}$ and $C_{p-2^{5k}}$ the \emph{subframe} in $\ring(I_i,I_j)$. In this way,
we can obtain at most $4k$ subframes in total for all thick maximal terminal-free rings.
Moreover, note that the subframes are concentric. 

\begin{corollary}
    For a cheap $T$-linkage $\mathcal P$, every segment of $\mathcal P$ in $\textsf{Ring}(I_i.I_j)$ has at most one traversing subsegment and no visitor in $\ring(C_{2^{5k}},C_{p-2^{5k}})$.
\end{corollary}

Although $C$ and $C'$ cut the paths of $\mathcal P$ in a structured way, the complexity of a subframe might be large.
To handle this issue, we deal with two noose $B$ and $B'$ of total complexity $2^{O(k)}$ 
in $\textsf{Ring}(C,C')$ lying close to $C$ and $C'$, respectively, instead of dealing with $C$ and $C'$,
where $C=C_{2^{5k}}$ and $C'=C_{p-2^{5k}}$.
Recall that $\mathcal C=\langle C_1,\ldots, C_{p}\rangle$ is a tight sequence 
of concentric cycles in $\textsf{Ring}(I_i,I_j)$ with $p\geq 2^{10k}$. 
By Lemma~\ref{lem:min-cut}, a minimum cut for $V(C_{2^{5k}})$ and $V(C_{2^{6k}})$ in $\textsf{Ring}(C_{2^{5k}},C_{2^{6k}})$ has size at most $2^{3k}$. 
Moreover, such a cut forms a noose $B$ in $\textsf{Ring}(C_{2^{5k}},C_{2^{6k}})$. 
We can compute it in $2^{O(k)}n$ time using the Ford–Fulkerson algorithm for computing a maximum flow in $O(fN)$ time, where $f$ denotes the size of a maximum flow, and $N$ denotes the complexity of a given  graph. In our case, $f=2^{O(k)}$ and $N=O(n)$. 
Similarly, we compute a noose $B'$ in $\textsf{Ring}(C_{p-2^{5k}},C_{p-2^{6k}})$ of complexity $2^{O(k)}$ in time $2^{O(k)}n$. 
In the following, we call the two nooses $B$ and $B'$ the \emph{frames}. 

\begin{figure}
    \centering
    \includegraphics{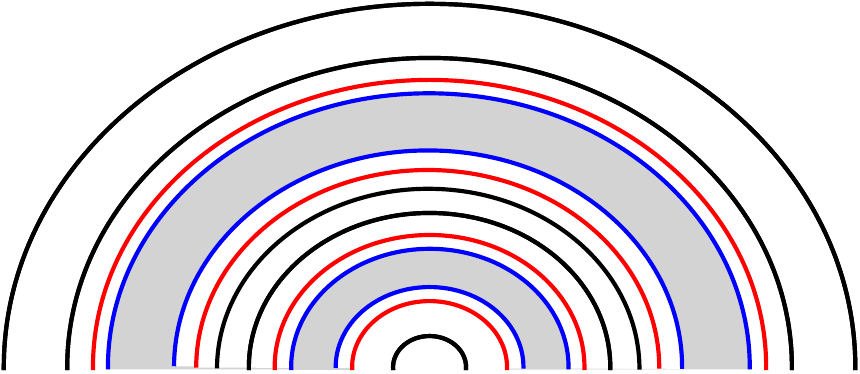}
    \caption{\small The black cycles are the boundaries of the thick maximal terminal-free rings,
    the red cycles are the subframes, and the blue cycles are the frames. The terminal-free framed rings are colored gray.}
    \label{fig:frames}
\end{figure}
Let $\mathcal B= \langle B_1, B_1', B_2, B_2', \ldots, B_{k'},B_{k'}' \rangle$ be the sequence of 
(concentric) frames such that $B_i$ and $B_i'$ are frames from the $i$th thick maximal terminal-free ring. See Figure~\ref{fig:frames}(a). 
Note that the radial distance between $V(B_i)$ and $V(B_i')$ is at least $2^{k}$, and the radial distance between $V(B_i')$ and $V(B_{i+1})$ is at most $2^{O(k)}$ by construction for all indices $i\in[k')$.
Moreover, no terminal is contained in $\textsf{Ring}(B_i,B_{i}')$. 
A ring defined by two consecutive nooses in $\mathcal B$ is called a \emph{framed ring}. 
A framed ring containing a terminal is called a \emph{terminal-containing} framed ring,
and a framed ring not containing a terminal is called a \emph{terminal-free} framed ring.

\paragraph{Substructure of a Terminal-Containing Framed Ring.}
Let $\oring=\textsf{Ring}(B_{i-1}',B_{i})$ be a terminal-containing framed ring.
We construct a substructure of $\oring$, which we will call the \emph{skeleton forest}.  
For each terminal $t$ of $\bar{T}$ lying in $\oring$, we connect $t$ and a vertex of $B_i$ by a radial curve (a shortest path in $\radgraph$) of complexity $2^{O(k)}$. 
By construction of $\mathcal B$, the radial distance between $B_{i}$ and $t$ is $2^{O(k)}$. 
In addition to this, we connect $B_{i-1}'$ and of $B_i$ by a radial curve of complexity $2^{O(k)}$. Recall that the radial distance between the two radial curves is at most $2^{O(k)}$.
We can compute all such radial curves in $2^{O(k)}n$ time in total. 
Let $\Gamma_i$ be the union of these radial curves, which forms a forest. We call $\Gamma_i$ the \emph{skeleton forest} of $\oring$. 
The number of leaf nodes of $\Gamma_i$ is $O(k_i)$, where $k_i$ denotes the number of terminals in $\oring$. 
Also, the number of nodes of $\Gamma_i$ of degree at least three is $O(k_i)$. 
A maximal path of $\Gamma_i$ consisting of degree-2 vertices 
is called a \emph{tree-path}. 
Also, a maximal subpath of $B_{i}$ (and $B_{i-1}'$) consisting of vertices not contained in $\Gamma_i$ is called a \emph{boundary-path}. 
Then there are $O(k_i)$ tree-paths and boundary-paths, and the total complexity of the tree-paths and boundary-paths is $2^{O(k)}$. 

\medskip
The following lemma summarizes this section. 

\begin{lemma}
    We can construct $O(k)$ framed rings in $2^{O(k)}n$ time in total. 
    Also, we can compute the skeleton forests for all terminal-containing framed rings in $2^{O(k)}n$ time in total.
\end{lemma}

\section{Crossing Pattern and Abstract Polygonal Schema}\label{sec:week_linkages}
In this section, we define a canonical encoding of a cheap $T$-linkage $\mathcal P$, which
we will call the crossing pattern, with respect to the frames.
Let $\mathcal B= \langle B_1, B_1', B_2, B_2', \ldots, B_{k'},B_{k'}' \rangle$ be the sequence of 
(concentric) frames such that $B_i$ and $B_i'$ are frames from the $i$th thick maximal terminal-free ring. 
Also, let $\mathcal C= \langle C_1, C_1', C_2, C_2', \ldots, C_{k'},C_{k'}' \rangle$ be the sequence of 
(concentric) subframes such that $C_i$ and $C_i'$ are subframes from the $i$th thick maximal terminal-free ring. 
Recall that the radial distance between $B_{i-1}'$ and $B_i$ is $2^{O(k)}$. 
That is, there is a radial curve between $B_{i-1}'$ and $B_i$ of complexity $2^{O(k)}$. 
We call the endpoints of the radial curve lying on $B_{i-1}'$ and $B_i$ the \emph{origins} of $B_{i-1}'$ and $B_i$, respectively.

We first define the \emph{winding number} of two walks in $\textsf{Ring}(D,D')$ for two nooses $D$ and $D'$ of $G$.
We say two walks are \emph{aligned} if the pairs of the endpoints of the two walks are the same. 
For two aligned walks $\pi$ and $\pi'$ traversing $\textsf{Ring}(D,D')$ 
oriented from the endpoints on $D$ to the endpoints on $D'$,
we record the signed numbers of crossings of $\pi$ along $\pi'$:
for every common vertex of $V(\pi)\cap V(\pi')$ excluding their endpoints, we record +1 if $\pi$ crosses $\pi'$ from left to right, record -1 if $\pi$ crosses $\pi'$ from right to left, and record 0 if $\pi$ does not cross $\pi'$ at that vertex. 
Also, if the first and last edges of $\pi$ lie in the same direction of $\pi'$, we record 0. 
If the first edge of $\pi$  lies to the left of $\pi'$, but the last edge of $\pi$ on $D'$ lies to the right of $\pi'$, then we record $-1$. Otherwise, we record $+1$. 
Then the winding number between $\pi$ and $\pi'$, denoted by $\textsf{WindNum}(\pi,\pi')$, is defined as the sum of all recorded numbers. 
See Figure~\ref{fig:winding}(a-b). 
If two walks $\pi$ and $\pi'$ are not aligned, we contract all the edges of $D_\textsf{sub}$ into a single vertex, and contract all the edges of $D_\textsf{sub}'$ into a single vertex 
so that the two walks become aligned, where
$D_\textsf{sub}$ (and $D_\textsf{sub}'$) is the subnoose of $D$ (and $D'$) 
lying from the endpoints of $\pi$ to the endpoints of $\pi'$ in clockwise direction. 
 Then the winding number between $\pi$ and $\pi'$
is defined as their winding numbers after the contraction. 
See Figure~\ref{fig:winding}(c). 
\begin{figure}
    \centering
    \includegraphics{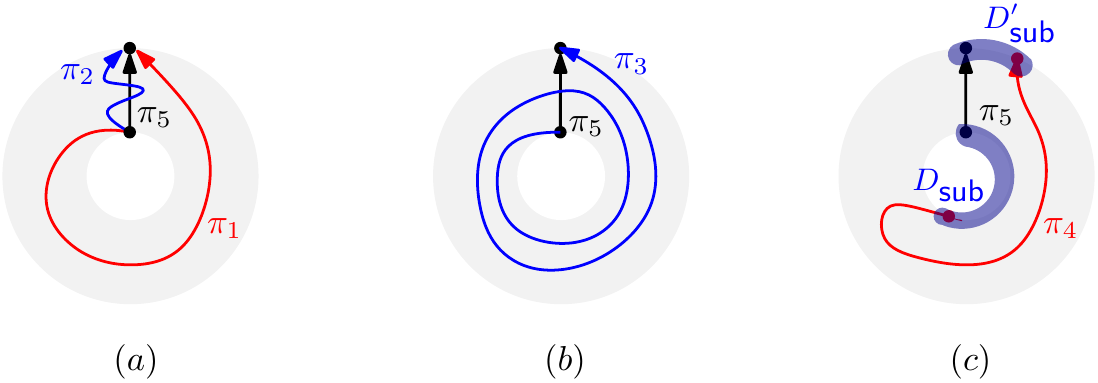}
    \caption{\small $\windnum(\pi_1,\pi_5)=-1$, $\windnum(\pi_2,\pi_5)=0$, $\windnum(\pi_3,\pi_5)=-2$, and $\windnum(\pi_4,\pi_5)=0$}
    \label{fig:winding}
\end{figure}

The following lemma holds due to Proposition~7.1 and Proposition~7.2 of~\cite{lokshtanov2020exponential}. 
\begin{lemma}\label{lem:winding}
     Let $\textsf{Ring}(D,D')$ be a ring defined by two nooses $D$ and $D'$,
     and let $\mathcal P$ and $\mathcal Q$ be two linkages traversing $\textsf{Ring}(D,D')$.  
     Then there exists a traversing linkage $\mathcal P'$  aligned with $\mathcal P$ such that  
     for a fixed path $Q\in\mathcal Q$, $|\textsf{WindNum}(P', Q)| \leq 7$ for all paths $P'$ of $\mathcal P'$. 
\end{lemma}
\begin{proof}
By Proposition~7.2 in~\cite{lokshtanov2020exponential}, there exists a linkage $\mathcal P'$ aligned with $\mathcal P$ so that each $P'$ in $\mathcal P'$ satisfies $|\textsf{WindNum}(P',  Q')| \leq 6$ with a path $Q'$ in $\mathcal Q$.
Furthermore, for a fixed path $Q\in \mathcal Q$, Proposition~7.1 in~\cite{lokshtanov2020exponential} implies that $\textsf{WindNum}(P',  Q')$ and $\textsf{WindNum}(P',  Q)$ differ by at most one since $Q$ and $Q'$ are exactly same or vertex-disjoint paths. This completes the proof. 
\end{proof}

\subsection{Reference Paths and a Base $T$-Linkage}
Given a cheap $T$-linkage $\mathcal P$, we reroute each path of $\mathcal P$ to
control the winding number of each path of $\mathcal P$ and a precomputed path, called a \emph{reference path}, using Lemma~\ref{lem:winding}.
For each two consecutive subframes $C_i$ and $C_i'$ of $\mathcal C$, 
we compute a  maximum-cardinality set $\mathcal Q_i$ of vertex-disjoint paths from $C_i$ to $C_i'$. 
The paths in $\mathcal Q_i$ are called the \emph{reference paths} in $\ring(C_i,C_i')$. 
Let $Q$ be a fixed path $\mathcal Q_i$, which we call the \emph{initial} path of $\mathcal Q_i$. 
By Lemma~\ref{lem:min-cut}, the size of $\mathcal Q_i$ is $2^{O(k)}$, and thus  we can compute $\mathcal Q_i$ in $2^{O(k)}n$ time for all $i\in[k']$ 
using the Ford–Fulkerson algorithm. 

Also, the size of $\mathcal Q_i$ is at least the number of traversing segments of $\mathcal P$ in $\ring(C_i,C_i')$ since
the traversing segments of $\mathcal P$ in $\ring(C_i,C_i')$ are also vertex-disjoint paths  from $C_i$ to $C_i'$. 
By Lemma~\ref{lem:winding}, there is a linkage $\mathcal P'$,  
aligned with the set of all traversing segments of $\mathcal P$ in $\textsf{Ring}(C_i,C_i')$ 
with $|\textsf{WindNum}(P', Q)| \leq 7$ for all paths $P'\in\mathcal P'$. 
We replace each traversing segment of $\mathcal P$ in $\ring(C_i,C_i')$ with its corresponding segment of $\mathcal P'$. We do this for all two consecutive subframes defining a terminal-free ring.
Since $\ring(C_i,C_i')$ has no visitor, the resulting paths are pairwise vertex-disjoint. 
We call the resulting linkage $\mathcal P$ the \emph{base $T$-linkage}. 
Note that the base $T$-linkage always exists if $(G,T,k)$ is a \textsf{YES}-instance.

\paragraph{Oscillating Subsegments and Traversing Subsegments.}
Now let $\mathcal P$ be a base $T$-linkage. 
Let $\mathcal S$ be the set of the traversing  segments of $\mathcal P$ in $\ring(C_i,C_i')$, and let $\mathcal Q_i$ be the set of all reference paths in $\ring(C_i,C_i')$.
Since $\textsf{Ring}(C_i,C_i')$ is terminal-free, 
every subsegment of a segment in $\mathcal S$ having both endpoints on $V(C_i)$ (or $V(C_i')$)
is fully contained in $V(C_i)$ (or $V(C_i')$).
Also, we may assume that every path of $\mathcal Q_i$ intersects $V(C_i)$ and $V(C_i')$ only at its endpoints; if it is not the case, we can take its maximal subpath
having one endpoint on $V(C_i)$ and $V(C_i')$. 
However, it is possible that a path of $\mathcal S\cup \mathcal Q_i$ crosses $B_i$ (and $B_i'$) more than once. 

\begin{figure}
    \centering
    \includegraphics{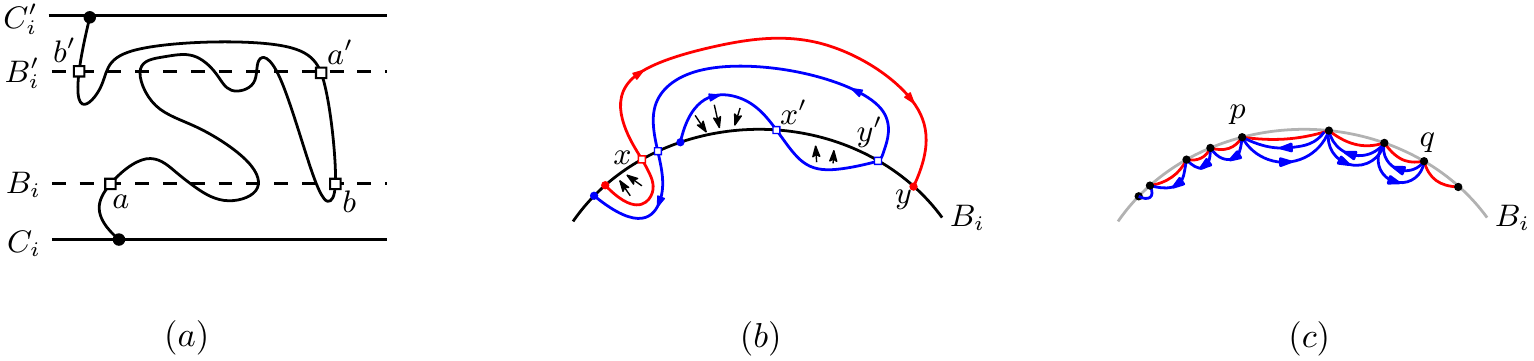}
    \caption{\small (a) The subpath of $S$ lying between $a$ and $b$ is the subsegment oscillating around $B_i$,
    and the subpath of $S$ lying between $b'$ and $b'$ is the subsegment oscillating around $B_i'$.
    (b) The red oscillating subsegment has two pieces, and the blue one has four pieces.
    The piece with endpoints $x$ and $y$ has order three, and the piece with endpoints $x'$ and $y'$ has order zero.
    We push the pieces towards their base arcs. See the arrow. 
    (c) After pushing the pieces, we obtain two walks traversing along $B_i$. Then we remove the six blue edges lying between $p$ and $q$ during the U-turn elimination process. 
    }
    \label{fig:oscillating}
\end{figure}

To handle this, we decompose each path $S$ in $\mathcal S\cup \mathcal Q_i$ into five pieces: a prefix, two oscillating subsegments, a traversing subsegment, and a suffix. 
First, we consider $S$ as oriented 
from its endpoint on $C_i$ to its endpoint on $C_i'$.
Let $a$ (and $b$) is the first (and last) vertex of $s$ lying on $B_i$. Then, let $a'$ (and $b'$) be the first (and last) vertex of $S$ lying on $B_i'$ among all vertices of $S$
lying after $b$ along $S$. See Figure~\ref{fig:oscillating}(a). 
We say the subpath lying between $a$ and $b$ (and between $a'$ and $b'$) is the \emph{subsegment oscillating around} $B_i$ (and $B_i'$).
Also, the subpath lying between $b$ and $a'$ is called the \emph{traversing subsegment} of $S$. 
The other pieces are called the prefix and suffix of $S$, respectively. 
Also, the subsegment between $a$ and $b'$ is called the \emph{middle subsegment} of $S$. 

\begin{lemma}
    For each path $S$ of $\mathcal S$, 
    the absolute value of the 
    winding number between the middle subsegments of $S$ and $Q$ is  $2^{O(k)}$.
\end{lemma}
\begin{proof}
Let $\bar S$, $S_\textsf{pre}$, and $S_\textsf{su}$ be the middle subsegment, prefix, and suffix of $S$, respectively. Analogously, we let $\bar Q$, $Q_\textsf{pre}$, and $Q_\textsf{su}$
be the middle subsegment, prefix, and suffix of $Q$, respectively. 
Our goal is to give an upper bound of $|\textsf{WindNum}(\bar S, \bar Q)|$. 
It is at most the sum of $|\textsf{WindNum}(S,Q)|$ and the absolute value of the winding numbers between the two middle subsegments,
the two prefixes, and the two suffixes. We show that the absolute value of the winding number between each pair of pieces is $2^{O(k)}$. 
Note that, by construction $|\textsf{WindNum}(S,Q)|\leq 7$. 

We first show that $|\textsf{WindNum}(S_\textsf{pre}, Q_\textsf{pre})|$ is at most $2^{O(k)}$. By the definition, two prefixes are segments of traversing segments in $\ring(C_i,B_i)$. Note that, 
By the tightness, there is a radial curve $\mu$ connecting $C_i$ and $B_i$ whose length is at most $2^{O(k)}$ in $\ring(C_i,B_i)$. Note that    $|\textsf{WindNum}(S_\textsf{pre}, \mu)|$ and $|\textsf{WindNum}(Q_\textsf{pre}, \mu)|$ are at most the complexity of $\mu$, which is $2^{O(k)}$. Furthermore, by  Proposition~7.1 in~\cite{lokshtanov2020exponential}, $|\textsf{WindNum}(S_\textsf{pre}, Q_\textsf{pre})|$ differs by at most one from the sum of $|\textsf{WindNum}(S_\textsf{pre}, \mu)|$ and $|\textsf{WindNum}(Q_\textsf{pre}, \mu)|$. Therefore, $|\textsf{WindNum}(S_\textsf{pre}, Q_\textsf{pre})|$ is $2^{O(k)}$, and  $|\textsf{WindNum}(S_\textsf{su}, Q_\textsf{su})|$ is $2^{O(k)}$.

Now we show that $|\textsf{WindNum}(\bar S, Q_\textsf{pre})|$ is $2^{O(k)}$. 
The crossing between $\bar S$ and $Q_\textsf{pre}$ appears in $\ring(C_i,B_i)$. The part of $\bar S$ contained in $\ring(C_i,B_i)$ has both end vertices in $V(B_i)$. Furthermore, the number of such subsegments is $2^{O(k)}$ since $|B_i|\in 2^{O(k)}$. 
The absolute value of the winding number between such a subsegment and a traversing subsegment is at most one. Thus, $|\textsf{WindNum}(\bar S, Q_\textsf{pre})|$ is $2^{O(k)}$. Analogously, $|\textsf{WindNum}(\bar S, Q_\textsf{su})|$, $|\textsf{WindNum}(S_\textsf{pre}, \bar Q)|$, and $|\textsf{WindNum}(S_\textsf{su}, \bar Q)|$ are also $2^{O(k)}$. Therefore, we can conclude that the absolute value of the winding number  between the middle subsegments is $2^{O(k)}$.
\end{proof}

\subsection{Canonical Weak Linkage for a Base $T$-Linkage $\mathcal P$}
In this subsection, we define a \emph{canonical weak linkage} $\mathcal W$ which is discretely homotopic to the base $T$-linkage $\mathcal P$. 
To work with discrete homotopy, we use the radial completion $G^\textsf{rad}$ of $G$. 
Note that $\mathcal P$ is also a $T$-linkage of $G^\textsf{rad}$.
The frames of $G$ are nooses of $G$, thus we can consider them as cycles in $G^\textsf{rad}$.  
For any two walks in the canonical weak linkage $\mathcal W$ we construct in this subsection, 
an edge traversed by the walks of $\mathcal W$ by more than once lies 
on the frames or skeleton forests only. Moreover, no walk of $\mathcal W$ uses an edge of $E(G^\textsf{rad})\setminus E(G)$ 
which does not lie on any frame or skeleton forest. These two properties are crucial for reconstructing a $T$-linkage of $G$ from $\mathcal W$  in Section~\ref{sec:reconstructing}.

To obtain $\mathcal W$, 
we first \emph{push} some subsegments of the segments of $\mathcal P$ in each framed ring 
 in a specific order 
so that no two paths of $\mathcal P$ cross during the pushing procedure. 
For a terminal-free framed ring,  we push the subsegments of the paths of $\mathcal P$ oscillating around the frames onto the frames. 
For a terminal-containing framed ring, we push the minimal subsegments $\pi$ of the paths of $\mathcal P$ such that 
the end vertices of $\pi$ lie on the same tree-path of the skeleton forest defined in the ring. 


\paragraph{Pushing Segments on a tree-path $\tau$.}
For a terminal-containing ring, imagine that we cut each segment of $\mathcal P$ in the ring with respect to the skeleton forest. 
An endpoint of the subsegments lies on the skeleton forest or the boundary of the ring. For the subsegments having both endpoints on the same tree-path, we push them to the tree-path. 
For a tree-path $\tau$, let $L$ be the set of 
 all minimal subsegments $\pi$ having both endpoints on $\tau$. 
Note that $|L|=2^{O(k)}$ since the length of $\tau$ is $2^{O(k)}$, and the paths of $\mathcal P$ are pairwise vertex-disjoint. 
The \emph{base arc} of a subsegment $\pi\in L$ is defined as the subpath of $\tau$ having the same endpoints as $\pi$. 
The \emph{order} of $\pi$ is defined as the number of the other subsegments in $L$ contained in the region bounded by $\pi$ and its base arc. 
We push the subsegments in $L$ in the increasing order by applying the face operations. 
In this way, the base $T$-linkage $\mathcal P$ is transformed into a weak $T$-linkage. 
Since we transform $\mathcal P$ to the weak $T$-linkage using the face operations, they are homotopic. 

\paragraph{Pushing Segments on $B_i$.} For each terminal-free ring $\xring$, we first cut each subsegment oscillating around $B_i$ into several \emph{pieces} 
at the vertices intersected by $B_i$. 
Let $S_\textsf{os}$ be the set of all subsegments oscillating around $B_i$, and 
let $\Pi$ be the set of all pieces obtained from the subsequences of $S_\textsf{os}$ by cutting them at the vertices intersected by $B_i$. 
Each piece $\pi$ has its \emph{base arc} on $B_i$ and its \emph{order}.
The definition of base arcs and orders are similar to their definitions in the proof of Lemma~\ref{lem:visitors}. 
For a piece $\pi$ oriented from $x$ to $y$, there are two paths between $x$ and $y$ in $B_i$. Furthermore, exactly one of them forms,  together with $\pi$, a bounded region not containing the interior of $C_i$ in its interior. 
We call this arc the \emph{base arc} of $\pi$. 
We consider it as oriented from $x$ to $y$. See Figure~\ref{fig:oscillating}(b). 
We define the order of $\pi$ as the number of other pieces (which could come from the same subsegment) whose base arcs are  contained in the base arc of $\pi$.  
Then we push all pieces of $\Pi$ to their base arcs in the increasing order of their orders. 
We can push all pieces of $\Pi$ by applying face operations. 
Therefore, the set $\mathcal W_\textsf{pushed}$ of the walks obtained in this way 
is homotopic to $S_\textsf{os}$. 

\paragraph{Eliminating ``U-turns''.}  At this point, a walk of  $\mathcal W_\textsf{pushed}$ moves along $B_i$ back and forth. We further simplify the walks of
 $\mathcal W_\textsf{pushed}$ 
by eliminating ``U-turns''. 
To make the description easier, 
for an edge of $B_i$ used by the walks of $\mathcal W_\textsf{pushed}$ more than once,
say $N$ times, we make $N$ copies of the edge so that they become parallel edges in $B_i$. 
The complexity of $\radgraph$ increases in this way, 
but our algorithm does not compute these parallel edges. In fact, all arguments in this subsection show the existence of a weak linkage of $G$
satisfying certain properties, which be used only for the analysis of our algorithm. 
Due to the duplication of edges, we can consider $\mathcal W_\textsf{pushed}$ as a weak \emph{edge-disjoint} linkage. 

Whenever two edges of a walk $W$ of $\mathcal W_\textsf{pushed}$ lying consecutive along $W$ form a face $F$, 
 we apply the face operation for $W$ on $F$ so that the two consecutive edges of $W$ are eliminated. 
Note that this happens only when the two edges are copies of the same edge of $G$, and $W$ traverses them in the opposite directions. 
See Figure~\ref{fig:oscillating}(c). 
We repeat this until no such two consecutive edges exist. Recall that we do this only for the purpose of analysis, and thus we do not need to
care about the running time of this procedure. 
Let $\mathcal W_{\textsf{elim}}$ be the walks we obtained from $\mathcal W_\textsf{pushed}$ by eliminating all U-turns.


\begin{lemma}\label{lem:oscillating}
    The walks in $\mathcal W_{\textsf{elim}}$ are non-crossing. Moreover, either they are vertex-disjoint paths, or they turn around $B_i$ in the same direction. 
    Moreover, the total complexity of $\mathcal W_{\textsf{elim}}$ is $2^{O(k)}$. 
\end{lemma}
\begin{proof}
Since pushing and eliminating U-turns keeps a homotopic relation, the walks in $\mathcal W_\textsf{elim}$ are non-crossing.
To make the description easier, we represent the interaction between the copies of an edge $e$ of $B_i$ and the walks of $\mathcal  W_{\textsf{elim}}$
as a string over the alphabet $\Sigma=\{1,2,\ldots, r'\}$, where $r'=|\mathcal W_\textsf{elim}|$.  Each symbol $\ell$ in the alphabet represents
the $\ell$th walk of $\mathcal W_\textsf{elim}$ for $\ell\in[r']$.
Assume that all edges of $B_i$ is oriented in clockwise direction along $B_i$. 
A copy of $e$ has a symbol $\ell$ if the $\ell$th walk traverses the copy
of $e$ in clockwise order, and a symbol $\ell^{-1}$
if the $\ell$th walk traverses $e$ in counterclockwise order. 
Since each copy of an edge of $G$ is traversed by the walks in $\mathcal W_\textsf{elim}$ only once, each copy of $e$ has at most one symbol. 
For an edge $e$ of $G$, we define $\sigma(e)$ as a sequence of symbols 
of the copies of $e$ in the order. 

We first claim that $\sigma(e)$ does not contain $\ell\ell^{-1}$ or
$\ell^{-1}\ell$ as its substring for any $\ell\in\Sigma$.
Assume to the contrary that there exists  an edge $e$ whose string $\sigma(e)$ has a substring $\ell\ell^{-1}$.
Let $W$ be the $\ell$th walk of $\mathcal{W}_\textsf{elim}$.
Imagine that we walk along $W$ from the copy of $e$ contributing to $\ell$ to the copy of $e$ contributing to $\ell^{-1}$. Then there must be an edge $e'$ of $G$ such that $W$ traverses two copies of $e'$ in the opposite directions consecutively. Moreover, since the copies of $e$ contributing to $\ell$ and $\ell^{-1}$ are consecutive among all copies of $e$, 
the copies of $e_1$ traversed by $W$ are also consecutive among all copies of $e_1$. Therefore, the two copies of $e'$ form a U-turn. 
The elimination step removes all U-turns, and thus this makes a contradiction.

Moreover, $\sigma(e)$ contains at most one of $\ell$ and $\ell^{-1}$ for all $\ell\in\Sigma$. 
To see this, assume to the contrary that a copy of $e$ has symbol $\ell$ (or $\ell^{-1}$) 
and another copy of $e$ has symbol $\ell^{-1}$ (or $\ell$). 
Among all substrings of $\sigma(e)$ starting with $\ell$ (or $\ell^{-1}$) and ending at $\ell^{-1}$ (or $\ell$) 
for a symbol $\ell\in\Sigma$, we choose a shortest one $\sigma'(e)$. 
If $\sigma'(e)$ has length two, it means $\sigma'(e)=\ell\ell^{-1}$ or $\ell^{-1}\ell$, and it contradicts the claim we proved earlier. 
Thus, we assume that $\sigma'(e)$ has length larger than two. 
This means that there is another symbol $\ell_1$ or $\ell_1^{-1}$ in $\sigma'(e)$. 
Without loss of generality, assume that $\sigma(e)$ 
starts with $\ell$ and ends with $\ell^{-1}$. Furthermore, $\ell_1\in \sigma(e)$. 
The other cases can be handled analogously. 

\begin{figure}
    \centering
    \includegraphics[width=0.9
    \textwidth]{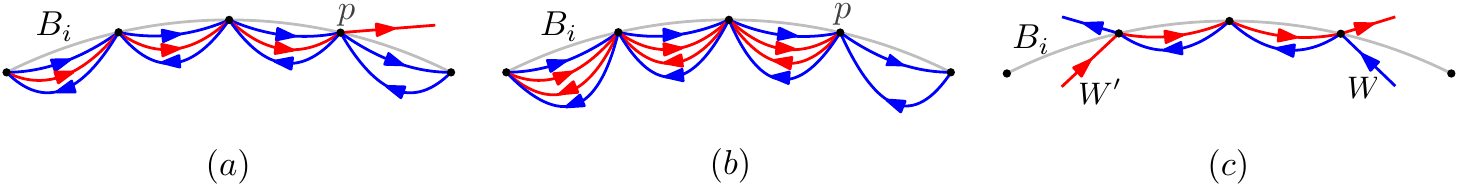}
    \caption{\small (a) The blue line is $W$ and the red line is $W_1$.
    $W_1$ leaves the frame $B_i$ at $p$ that must cross with $W$.
    (b) The case where $W_1$ has the counterclockwise direction at $p$ makes a U-turn.
    (c) The blue line is $W$ and the red line is $W'$.
    $W$ traverse $B_i$ in counterclockwise direction and $W'$ traverse $B_i$ in clockwise direction.
    The $W'$ must cross with $W$ for leaving or entering the frame.
    }
    \label{fig:canonical_weak}
\end{figure}

Let each $W$ and $W_1$ be the $\ell$th and $\ell_1$th walk of $\mathcal W_\textsf{elim}$. Similarly to the earlier proof, we walk along $W_1$ from a copy of $e$. 
The copy of $e$ contributing to $\ell_1$ lies between two copies of $e$ contributing to each $\ell$ and $\ell^{-1}$. Thus, the copies of $e_1$ traversed by the walk along $W_1$ lies between two copies of $e_1$ contributing to each of $\ell$ and $\ell^{-1}$.
Since $W$ and $W_1$ are non-crossing walks, $W_1$ ends $p$ which has 4 incident edges in $W$ or  traverses another copy of $e$ with opposite direction between $W$. The second case means that $\sigma(e)$ has $\ell_1^{-1}$ also. This contradicts that $\sigma(e)$ is the shortest one.
We consider the first case. By the definition of $\mathcal W_{\textsf{elim}}$, the path corresponded by $W_1$ leaves the frame $B_i$ at $p$. 
Since the walk $W_1$ is surrounded by $W$ at $p$, the path must cross with $W$ to leave the frame. This contradicts to the non-crossing condition.
See Figure~\ref{fig:canonical_weak}(a,b).



We have shown that, each walk $W$ in $\mathcal W_\textsf{elim}$ traverses $B_i$ in clockwise or counterclockwise direction. 
Now we show that two walks $W$ and $W'$ traverse $B_i$ in the same direction if they traverse two copies of a common edge $e\in B_i$. We  suppose that $W$ traverses $B_i$ in counterclockwise direction and $W'$ traverses $B_i$ in  clockwise direction. 
Let $\ell$ and $\ell'$ be the indices of $W$ and $W'$ in $\mathcal W_\textsf{elim}$. We consider the maximal substring $\sigma'(e)$ of $\sigma(e)$ which starts with $\ell^{-1}$ or $\ell'$. Without loss of generality, we assume that it starts with $\ell^{-1}$. 
This means that the subwalk of the reverse of $W'$ must cross the walk $W$. This contradicts to the fact that
$W$ and $W'$ are non-crossing. 
See Figure~\ref{fig:canonical_weak}(c).

The remaining task is to bound the complexity of $\mathcal W_\textsf{elim}$. 
The complexity of $\mathcal W_\textsf{elim}$ is at most the sum of length of the base arcs of $S_\textsf{os}$. The number of subsegments in $S_\textsf{os}$ is at most $|B_i|$. Furthermore, the length of each base arc is also at most $|B_i|$. Thus, the complexity of $\mathcal W_\textsf{elim}$ is at most $|B_i|^2\in 2^{O(k)}$. This proves the lemma.
\end{proof}

We replace each oscillating segment of $\mathcal S_\textsf{os}$ 
with its corresponding walk in $\mathcal W_\textsf{elim}$. We do this for $B_i'$, and 
We also do this for all rings $\ring(C_i,C_i')$ with $i\in[k']$. 
The resulting weak $T$-linkage is called the \emph{canonical} weak $T$-linkage for $\mathcal P$. 
Note that 
an edge (and a vertex) of $G$ traversed by the walks  in the canonical weak linkage $\mathcal W$ more than once lies on frames of $\mathcal B$.
Also, no walk of $\mathcal W$ uses an edge of $E(G^\textsf{rad})\setminus E(G)$ 
which does not lie on any frame. 

\begin{lemma}\label{lem:no-loop in tree-path}
    No maximal subwalk of the walks of $\mathcal W$ contained in $\oring \setminus (E(B_{i-1}')\cup E(B_i))$ has their endpoints on the  same boundary-path or the same tree-path. 
\end{lemma}

\subsection{Crossing Pattern with respect to Abstract Polygonal Schema} 
Now we are ready to define the crossing pattern of the canonical weak linkage $\mathcal W$ of $\mathcal P$. To do this, we encode the information
on the interaction between $\mathcal W$ and each framed ring. 
In the following, we let $\ring_i^\textsf{o}=\ring(B_{i-1}', B_i)\setminus (E(B_{i-1}')\cup E(B_i))$ and $\ring_i^\textsf{x}=\ring(B_{i},B_i')$. 
In this way, each edge of $E(G^{\textsf{rad}})$ belongs to exactly one of the framed rings.
Note that $\ring_i^\textsf{o}$ 
contains a terminal of $\bar{T}$, but $\ring_i^{\textsf{x}}$ contains no terminal of $\bar{T}$. 
Moreover, the radial distance between $B_{i-1}'$ and $B_i$ is $2^{O(k)}$
by construction. That is, the shortest path in $\radgraph$ between a vertex of $B_{i-1}'$ and a vertex in $B_i$ is $2^{O(k)}$.

\paragraph{Interaction between $\mathcal W$ and $\textsf{Ring}_i^\textsf{o}$.} 
A maximal subwalk $W$ of a walk in $\mathcal W$ contained in $\textsf{Ring}_i^\textsf{o}$ has endpoints on
the frames or the terminals. Also, no edge of the frames
are included in $\oring$.

\begin{figure}
    \centering
    \includegraphics[width=0.9\textwidth]{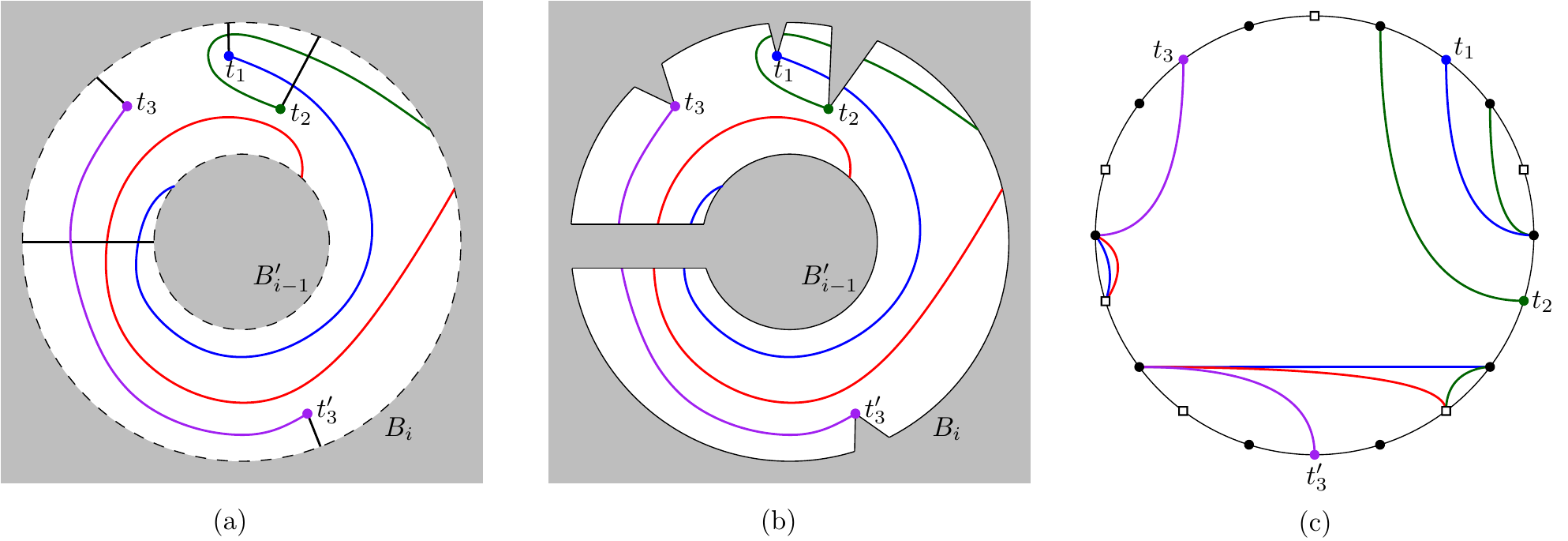}
    \caption{\small (a) 
    The black lines are tree-paths in $\Gamma_i$, and the dashed lines are boundary-paths. 
    The colored vertices of the leaf nodes in $\Gamma_i$ are terminals.
    (b) By cutting $\oring$ along $\Gamma_i$, every terminal lies on the boundary, and each tree-path appears exactly twice.
    (c) The black vertices are tree-paths, the squares are boundary-paths, and the colored vertices are terminals in $\oring$.
    }
    \label{fig:poly_schema}
\end{figure}

We first cut $\oring$ along the skeleton forest $\Gamma_i$ of $\oring$ so that every terminal in $\oring$ lies on the boundary of $\oring$ and each tree-path appears exactly twice. See Figure~\ref{fig:poly_schema}. 
Since the complexity of $\Gamma_i$ is $2^{\Theta(k)}$ in the worst case, 
we represent it in a more compact way as an \emph{abstract polygonal schema} $\Delta_i$: a convex polygon with $O(k_i)$ vertices, where $k_i$ denotes the number of terminals in $\oring$. This idea also used in~\cite{erickson2011shortest} for computing
shortest non-crossing walks in a planar graph. 
Each vertex of $\Delta_i$ corresponds to a tree-path, a boundary-path, or 
a vertex of $\Gamma_i$ of degree one or degree at least three. 
A maximal subwalk $W$ of a walk of $\mathcal W$ contained in $\oring$ corresponds to a diagonal of $\Delta_i$ due to Lemmas~\ref{lem:visitor-oring} and~\ref{lem:no-loop in tree-path}.  
Note that two distinct maximal subwalk of a walk of $\mathcal W$ might correspond to the same diagonal of $\Delta_i$. We define the \emph{weight} of a diagonal of $\Delta_i$ as the number of distinct maximal subwalks of the walks of $\mathcal W$ corresponding to the diagonal. Note that each diagonal has weight $2^{O(k)}$. 
In this way, 
the interaction between $\mathcal W$ and $\oring$ is encoded as a weighted triangulation of $\Delta_i$. 
 
\paragraph{Interaction between $\mathcal W$ and $\textsf{Ring}_i^\textsf{x}$.} 
A maximal subwalk $\omega$ of a walk of $\mathcal W$ 
contained in $\xring$ 
has  one endpoint on $V(B_{i})$ and the other endpoint on $V(B_{i}')$. 
But in this case, the edges of the frames $B_{i}$ and $B_i'$ are included in $\xring$. Therefore, 
$\omega$ first rotates around $B_{i}$ several times, and then it departs from $B_{i}$. Then it hits $B_i'$, and then rotates around $B_i'$ several times. Since $\xring$ is terminal-free, it suffices to record the winding numbers of the maximal walks as the crossing pattern.

Let $\mathcal W^i= \langle \omega^1,\ldots,\omega^\ell\rangle$ be the sequence of maximal walks of the walks of $\mathcal W$ in $\xring$
sorted in sorted along the their endpoints in $V(B_i)$. 
Also, let $\mathcal W^i_{\textsf{mid}}$ be the sequence of the subsegments of $\omega^j$'s excluding the prefixes and suffixes of
the walks of $\mathcal W^i$. 
Also, let $Q_\textsf{mid}$ be the part of the initial path of $\mathcal Q_i$ excluding its prefix and its suffix. 
As the crossing pattern, we first record $\ell$, the size of $\mathcal W^i$.
Then we record $\textsf{WindNum}(\omega, Q_\textsf{mid})$ 
for all subsegments $\omega$ of $\mathcal W^i_\textsf{mid}$ in a \emph{compact} way.
There exists a $\omega_j$ in $\mathcal W^i_{\textsf{os}}$ (and $\mathcal W^i_{\textsf{os}}$) such  that $\textsf{WindNum}(\omega_{j'},Q_\textsf{mid})$ is the same for all $j'\leq j$, and 
$\textsf{WindNum}(\omega_{j''},Q_\textsf{mid})$ is the same for all $j''>j$ 
by the proof of Lemma~\ref{lem:oscillating}. Therefore, it is sufficient to record the index $j$,  $\textsf{WindNum}(\omega_{j},Q_\textsf{mid})$, and 
$\textsf{WindNum}(\omega_{j+1},Q_\textsf{mid})$.

\medskip

The crossing pattern of $\mathcal W$ consists of 
the weighted triangulations for all terminal-containing framed rings  $\textsf{Ring}_i^\textsf{o}$ and 
the two components of the encoding for all  terminal-free framed rings  $\textsf{Ring}_i^\textsf{x}$.

\begin{lemma}\label{lem:num-crossing}
We can enumerate $2^{O(k^2)}$ different crossing patterns in $2^{O(k^2)}n$ time one of which
is the crossing pattern of the canonical weak linkage
of a base $T$-linkage. 
\end{lemma}
\begin{proof}
    To enumerate crossing patterns, we first compute $\Delta_i$ for all terminal-containing framed rings $\oring$ in $2^{O(k)}n$ time. 
    Also, we compute the set  $\mathcal Q_i$ of the reference paths 
    for all terminal-free framed rings $\xring$ in $2^{O(k)}n$ time.
    Then for each index $i$, we decompose
    the initial path of $\mathcal Q_i$ into five pieces: the prefix, suffix, oscillating subsequences, and traversing subsequences. 
    This also takes $2^{O(k)}n$ time. 
 
    We first analyze the number of different crossing patterns for the terminal-containing rings.
    The crossing pattern for each ring $\oring$ is a weighted triangulation of $\Delta_i$. Recall that the complexity of $\Delta_i$ is $O(k_i)$, where $k_i$ denotes the number of terminals contained in $\oring$. Also, the weight of each diagonal of $\Delta_i$ is  $2^{O(k)}$.
    The number of distinct triangulations of $\Delta_i$ is $2^{O(k_i)}$.
    For a fixed triangulation, 
    the number of distinct combinations of the weights of the diagonals
    is $(2^{O(k)})^{O(k_i)}$ since $\Delta_i$ has $O(k_i)$ diagonals.
    The total number of different crossing patterns for all terminal-containing rings is $\prod_{i}(2^{O(k)})^{O(k_i)}=2^{O(k^2)}$ since
    the sum of $k_i$ for all terminal-containing rings is $k$.

    Now we analyze the number of different crossing patterns for the terminal-free framed rings. Recall that the number of terminal-free framed rings is at most $k$.
    As the crossing pattern, we record the number of walks, and a specific index of the walks and the winding numbers of two specific walks, which are all integers of $2^{O(k)}$. The number of combinations of the four integers  for 
    all terminal-free rings is $(2^{O(k)})^k=2^{O(k^2)}$. 
    Therefore, the total number of distinct crossing patterns for all terminal-free framed rings is $2^{O(k^2)}$.

    Therefore, the total number of different crossing patterns is $2^{O(k^2)}$, and we can enumerate each crossing pattern in time linear in the complexity of the crossing pattern. Therefore, in total,
    we can enumerate all crossing patterns in $2^{O(k^2)}n$ time. 
\end{proof}

\section{Reconstructing a Weak Linkage from a Crossing Pattern}\label{sec:weak_linkage_construction}
In this subsection, given a crossing pattern $\sigma$, we show how to compute
a weak $T$-linkage $\mathcal W$ 
homotopic to a $T$-linkage $\mathcal P$ whose crossing pattern is $\sigma$. 
The weak linkage $\mathcal W$ which we will construct in this subsection
might traverse the same vertex more than once. But in this case, such a vertex 
lies on a frame. 
Recall that the total complexity of the frames is $2^{O(k)}$, which will be crucial to reconstruct a $T$-linkage from a weak $T$-linkage efficiently in Section~\ref{sec:reconstructing}. 

We compute the subwalks of the walks of $\mathcal W$ contained in each framed ring,
and then we merge the subwalks from two framed rings so that they become a weak $T$-linkage.
To make the merge procedure easier, for a frame $B_i$ (and $B_i'$), we contract all edges on $B_i$ (and $B_i'$) so that $B_i$ (and $B_i'$) consists a single loop
with both endpoints on the origin of $B_i$ (and $B_i'$).  Let $G_\textsf{contract}$ be the plane graph obtained from $G$ by contracting all edges on the frames. 
A weak linkage of $\gcontract$ corresponds to a weak linkage of $G$ in a straightforward way. 
Moreover, if two weak linkages of $\gcontract$ are homotopic, then their corresponding weak linkages in $G$ are also homotopic. 

In the following, we consider $\ring_i^{\textsf{o}}$ and $\ring_i^{\textsf{x}}$ as the framed rings of $\gcontract$.
Also, we let $\mathcal W$ be the weak $T$-linkage we are going to construct in this section. 

\begin{figure}
    \centering
    \includegraphics{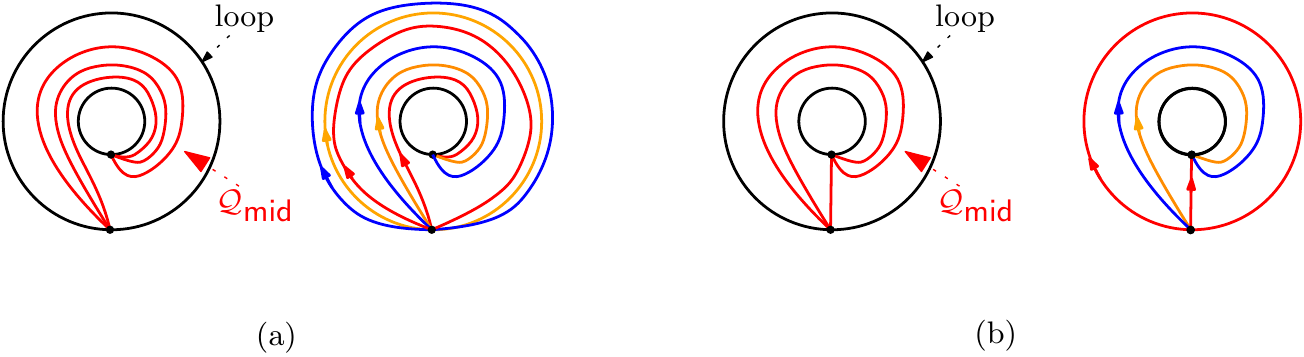}
    \caption{\small (a) The reconstruction of $\mathcal W$ in the case that  $\windnum(\omega, Q_\textsf{mid})=-2$
    for all $\omega\in\mathcal W_\textsf{mid}$,
    and $\windnum(Q, Q_\textsf{mid})=0$ for all $Q\in\mathcal Q_\textsf{mid}$.
    (b) The reconstruction of $\mathcal W$ in the case that  $\windnum(\omega, Q_\textsf{mid})=0$
    for all $\omega\in\mathcal W$, 
    and $\windnum(Q, Q_\textsf{mid})=-1$ for exactly one middle subsequences $Q$ of $\mathcal Q_\textsf{mid}$.    
    }
    \label{fig:reconstructing}
\end{figure}

\paragraph{Reconstructing a Weak Linkage lying on $\ring_i^\textsf{x}$.}
Let $\mathcal Q_i$ be the set of reference paths of $\xring$, and let $\mathcal Q_\textsf{mid}$ be the middle subsequences of 
the reference paths. 
Let $Q_\textsf{mid}$ be the middle subsegment of the initial path of $\mathcal Q_i$. 
Also, let $\mathcal W_\textsf{mid}=\langle \omega_1,\ldots,\omega_\ell\rangle$ be the set of middle subsegments
of the walks of $\mathcal W$. 

As the crossing pattern for $\xring$, we are given the size of $\mathcal W_\textsf{mid}$, and 
the winding numbers of $\mathcal W_\textsf{mid}$ with respect to $Q_\textsf{mid}$. We can construct $\ell$ walks each of which
traverses $B_i$ (which is a loop) several times, and then traverses
a path of $\mathcal Q_\textsf{mid}$ due to Lemma~\ref{lem:winding}. 

\begin{lemma}\label{lem:winding}
For any crossing pattern in $\xring$, we can reconstruct a weak linkage $\mathcal W_{\textsf{mid}}$ having the given crossing pattern in time linear in $2^{O(k)}n$ time. 
\end{lemma}
\begin{proof}
Our goal is to compute $\omega_x$ for all $x\in[\ell]$, where $\mathcal W_\textsf{mid}=\langle \omega_1,\ldots,\omega_\ell\rangle$. 
Note that $\mathcal W_{\textsf{mid}}$ and $\mathcal Q_i$ are sorted along its end vertices in clockwise order.
There are two tasks to construct $\omega_x$: first, we are required to decide which walk of $\mathcal Q_{\textsf{mid}}$ 
$\omega$ traverses. Next, we are required to determine how many times $\omega$ traverses around $B_i$. These canbe determined by looking
at the winding numbers of $\omega_1,\dots,\omega_{\ell}$. Without loss of generality, we assume that the winding number for the walks  $\omega_1,\dots \omega_j$ is $N$ and for the others $\omega_{j+1},\dots \omega_\ell$ is $(N-1)$. Otherwise, we can rotate $\mathcal W$ 
and $\mathcal W_{\textsf{mid}}$ (Recall that $\mathcal W$ is a cyclic sequence). 
Analogously, we assume that the last path $Q_{-1}$ of $\mathcal Q_\textsf{mid}$ has the largest winding number with respect to the initial path $Q_1$ of $\mathcal Q_{\textsf{mid}}$. 

We construct the $x$th walk $\omega_x$ of $\mathcal W_{\textsf{mid}}$ as follows: we first traverse $B_i$ in clockwise direction  $(N_{x}-q_x)$-times and then traverse $Q_x$, where $N_{x}$ and $q_x$ be the winding numbers of $\omega_x$ and $Q_x$, respectively, with respect to the initial path of $\mathcal Q_\textsf{mid}$. For illustration, see Figure~\ref{fig:reconstructing}. 
The constructed weak linkage $\mathcal W_{\textsf{mid}}$ has the give the crossing pattern. 
\end{proof} 

In this way, we can construct $\ell$ walks in $\xring$ such that the winding numbers with respect to $Q_\textsf{mid}$ 
are determined in the same way as the crossing pattern for $\xring$.

\paragraph{Reconstructing a Weak Linkage Lying on $\ring_i^\textsf{o}$.}
Recall that the crossing pattern for  $\ring_i^\textsf{o}$ is a weighted triangulation of $\Delta_i$. 
To obtain a weak linkage on $\ring_i^\textsf{o}$, we use Lemma~\ref{lem:boundary-instance}. 
We construct $G'$ and $T'$ satisfying the condition of Lemma~\ref{lem:boundary-instance} as follows. 
Recall that every vertex of $\Pi_i$ corresponds to a boundary-path, a tree-path, or a degree-1 vertex of the forest $\Gamma$. 
We first compute $\Gamma$, and cut $\oring$ along $\Gamma$ so that the tree-paths of $\Gamma$ appear exactly twice, and they become incident to the outer face of the resulting graph. Let $\ring_i^{\textsf{cut}}$ be the resulting graph. 
Then we contract all the edges in  each tree-path of $\Gamma$ on  $\ring_i^{\textsf{cut}}$. Let $G'$ be the resulting graph. 
A vertex of $G'$ lying on the outer face corresponds to a vertex of $\Delta_i$.
Then a diagonal of $\Delta_i$ corresponds to a vertex pair or $G'$. 
We let $T'$ be the set of all diagonals (pairs of vertices of $G'$) of the weighted triangulation of $\Delta_i$ in the crossing pattern.
If the diagonal (vertex pair) of the triangulation of $\Delta_i$ has weight $w$, we add $w$ copies of the vertex pair to $T'$. 

By applying Lemma~\ref{lem:boundary-instance} to $(G',T')$, we can compute a weak $T'$-linkage of $G'$  
such that the paths of the weak linkage are vertex-disjoint and edge-disjoint, except for their endpoints, in $2^{O(k)}n$ time.
Then we ``uncontract'' the vertices of $G'$ to reconstruct $\ring_i^{\textsf{cut}}$. 
In this way, the walks of the weak $T'$-linkage of $G'$ becomes the walks of $\ring_i^{\textsf{cut}}$ 
accordingly. Then we reconstruct $\oring$ from $\ring_i^{\textsf{cut}}$ by ``uncutting'' the boundary edges of $\ring_i^{\textsf{cut}}$ 
along $\Gamma$. They form a weak linkage $\mathcal W_i^{\textsf{o}}$ of  $\oring$ such that 
any two such walks traverse the same vertex only at the origins of the frames and $\Gamma$. 

\begin{lemma}\label{lem:boundary-instance}
Let $G'$ be a planar graph with $n$ vertices, and $T'$ be a (multi)set of terminal pairs such that the terminals are incident to the outer face of $G$. Then we can find a weak $T'$-linkage 
such that the walks of the weak linkage are vertex-disjoint (and also edge-disjoint), except for their endpoints, 
in $O(n|T'|)$ time, if it exists.
\end{lemma}
\begin{proof}
    We assume that $G'$ is connected. Otherwise, we consider each component independently. Let $C$ be the boundary cycle which incident to the outer face. Every terminals are contained in $C$.
    We say two terminal nodes $t$ and $t'$ in $T'$ is consecutive if $t'$ be the very first terminal node from $t$ in clockwise (or counterclockwise) direction along $C$. 

    Our algorithm works as follows: 
    choose a terminal pair $(t, t') \in T'$ which consecutive, and compute 
    the path between $t$ and $t'$ along the outer face of $G'$ in the  clockwise direction. 
    If there is not such consecutive terminal pair, then we return \textsf{NO}.
    Then we remove internal vertices and edges of the path from $G'$, and remove $(t,t')$ from $T'$. 
    Furthermore, if a vertex $t$ has no pair in $T'$, we delete $t$ also.
    Even if the obtained $G'$ is not connected, we can consider each connected component separately. 
    We repeat this procedure until $T'$ becomes empty. Let $\mathcal W$ be the weak $T'$-linkage obtained in this way. 
    Then the walks of $\mathcal W$ are vertex-disjoint and edge-disjoint, except for their endpoints. 

    We show that, if the instance is a \textsf{YES}-instance, then there exists a pair $(t,t')$ in $T'$ 
    such that $t$ and $t'$ lie consecutively along $C$.  
    We assume to the contrary that there is a solution $\mathcal W$ to the instance, but there is no such terminal pair in $T'$. 
    Then for any pair $a=(t,t')$, $G'$ is separated by two subgraphs with respect to the path of $\mathcal W$ connecting $t$ and $t'$. 
    Let $\pi_a$ and $\pi_a'$ be the parts of $C$ contained in the two subgraphs. 
    We define the \emph{order} of a pair $a=(t,t')\in T'$ as the minimum of $|V(\pi_a)\cap \bar{T}|$ and $|V(\pi_a')\cap \bar{T}|$.  
    Let $a=(t,t')$ be the pair of $T'$ with a smallest order. 
    Without loss of generality, we assume that $\pi_a$ contains a smaller number of terminals than $\pi_a'$. 
    Let $W$ be the path in $\mathcal W$ connecting $t$ and $t'$. 
    If there exists a terminal pair in $T'$ consisting of two terminals contained in different paths among $\pi_a$ and $\pi_a'$, 
    then any path connecting them crosses $W$. 
    Thus, for every pair in $T'$, the two terminals of the pair are contained in the same path among $\pi_a$ and $\pi_a'$. 
    If there is no terminal pair in $T'$ whose terminals are contained in $\pi_a$, this contradicts to the assumption that $t$ and $t'$ does not lie consecutively along $C$. If it is not the case, there is a terminal pair whose terminals are contained in $\pi_a$, this contradicts that $(t,t')$ is a pair with a smallest order. Thus, there exists a terminal pair in $T'$ whose terminals lie consecutively along $C$
    if the given instance is a \textsf{YES}-instance.
    
    Now we show that the algorithm returns a feasible solution if there exists a feasible weak $T'$-linkage $\mathcal W$. Let $(t,t')$ be a  terminal pair in $T'$ such that $t$ and $t'$ lie on $C$ consecutively.
    We claim that $\mathcal W$ remains to be a weak $T'$-linkage even after we replace $\omega$ with $\omega'$, 
    where $\omega$ is the walk of $\mathcal W$ between $t$ and $t'$, and $\omega'$ is the walk connecting $t$ and $t'$ returned by our algorithm. 
    Note that $\omega$ and $\omega'$ form a bounded face whose boundary and interior have no terminal node, except $t$ and $t'$. Thus, if there exists another walk in $\mathcal W$ which intersects $\omega'$, it would touch $\omega'$ first. This contradicts
    the fact that the walks in $\mathcal W$ are vertex-disjoint (except for their endpoints). Thus,  $\mathcal W$ remains to be a weak $T'$-linkage even after we replace $\omega$ with $\omega'$.

    In each iteration, we find a consecutive terminal pair and a path connecting them along $C$ in $O(n)$ time.  
    The number of iterations is $O(|T'|)$ since we we remove one terminal pair in each iteration.
    Therefore, the total running time is $O(n|T'|)$. 
\end{proof}

\paragraph{Merging the Subwalks Obtained from Each Framed Ring.}
We simply merge all subwalks obtained from the framed rings.
Since each frame forms a loop in $\gcontract$, we can merge them
in a straightforward way. Also, clearly, the resulting walks, say $\mathcal W$, are
non-crossing. Then we ``uncontract'' the loops of $\gcontract$ 
so that each loop becomes a cycle of $\radgraph$. 
Then $\mathcal W$ are ``uncontracted'' accordingly, and it becomes
a weak $T$-linkage of $G$ such that 
an edge traversed by the resulting weak linkage more than once
lies on the frames and the backbone forest. 

This section is summarized in the following lemma. 
\begin{lemma}\label{lem:enumeration-time}
    Given a crossing pattern $\sigma$, we can
    compute a weak $T$-linkage $\mathcal W$ of $\radgraph$ in $2^{O(k)}n$ time whose crossing pattern is $\sigma$ such that an edge of $\radgraph$
    traversed by $\mathcal W$ more than once lies on a frame or a skeleton forest, and an edge of $E(\radgraph)\setminus E(G)$
    traversed by $\mathcal W$ lies on a frame. 
    Moreover, if $\sigma$ is the crossing pattern of a base $T$-linkage $\mathcal P$, 
    then $\mathcal W$ is homotopic to $\mathcal P$. 
\end{lemma}

\section{Reconstructing a Linkage from a Weak Linkage}\label{sec:reconstructing}
In this section, we reconstruct a linkage from each weak linkage which we have constructed in Section~\ref{sec:weak_linkage_construction}. 
We use the the algorithm by Schrijver~\cite{schrijver1994finding} using suitable data structures. 
Given a weak $T$-linkage $\mathcal W$ of a planar graph, the algorithm by by Schrijver~\cite{schrijver1994finding} computes
a $T$-linkage of $G$ homotopic (and $F^*$-homologous) to $\mathcal W$ in $O(n^6)$ time, where $F^*$ denotes the innermost face of $G$. 
In our case, we improve the running time of this algorithm to $2^{O(k)}n$ by using suitable data structures and 
the properties of the weak linkage $\mathcal W$ we constructed in Section~\ref{sec:weak_linkage_construction}. 
Let $\gparallel$ be the planar graph obtained from $\radgraph$ by adding $2^{O(k)}$ copies of each edge of the frames and skeleton forests  
so that the walks of $\mathcal W$ become pairwise edge-disjoint. 


In Section~\ref{sec:homology_algorithm}, 
we briefly describe the algorithm of Schrijver~\cite{schrijver1994finding} so that its description is compatible with our notations.
Then we analyze the running time of this algorithm in a fine-grained manner with respect to 
structural parameters of $\gparallel$ as well as the complexity of the graph. Using a suitable data structure, we can implement the
algorithm of of Schrijver~\cite{schrijver1994finding} in $\textsf{poly}(\chi)\cdot n$ time, where $\chi$ is a structural parameter we used in the analysis. We will show that the structural parameter $\chi$ has value $2^{O(k)}$ in our case.  
To obtain a good bound on $\chi$, we slightly modify $G$ into another planar graph $G'$  in Section~\ref{sec:reducing_violation}. Then we analyze an upper bound on the structural parameter 
in Section~\ref{sec:pre_feasible}. 
By combining these arguments, we can conclude that the \textsf{Planar Disjoint Paths} can be solved in $2^{O(k^2)}n$ time.

\subsection{Summary of the Algorithm in~\cite{schrijver1994finding} and Analysis of Its Running Time}\label{sec:homology_algorithm}
  Since the algorithm in~\cite{schrijver1994finding} works with directed graphs, we replace each undirected edge of $\gparallel$ into two edges
  directed in the opposite directions. Also, we consider each walk $W_i$ of $\mathcal W$ as directed from $s_i$ to $t_i$. 
  Let $\mathcal F$ be the set of the faces of $\gparallel$. Note that there is a face of $\mathcal F$ whose boundary consists of 
  the copies of the edges of the innermost face $F^*$ of $G$. Recall that all frames contain the innermost face $F^*$ in their interiors. 
  To make the description easier, we also use $F^*$ to denote the face of $\gparallel$ corresponding to the innermost face of $G$. 
  
  \paragraph{Flow Function and Linkages.}
  A weak $T$-linkage $\mathcal W$ can be considered as a \emph{flow function} 
   $\phi : E \to \Sigma^*$,  
   where $\Sigma^*$ denotes the set of all strings consisting of the symbols in $\Sigma$
  for an alphabet $\Sigma$. Here, we use $\epsilon$ to denote the empty string. Let $\Sigma=\{1,\dots, k\}$. 
  We let $\sigma^{-1}$ be the string obtained from a string $\sigma$ in the reversed way. 
  A \emph{product} of two strings $\sigma, \sigma'$ of $\Sigma^*$,
  denoted by $\sigma\cdot\sigma'$, is defined as a string obtained by concatenating them and then by deleting all appearances of  $xx^{-1}=x^{-1}x$ for $x\in\Sigma$. Note that $\sigma \cdot \sigma^{-1}=\epsilon$. 
  Let $W_i$ be the path of $\mathcal W$ connecting $s_i$ and $t_i$ for $i\in[k]$. 
  For each edge $e\in E(\gparallel)$, let $\phi(e)=i$ if it is used in the path connecting $s_i$ to $t_i$, and 
  let $\phi(e)=\epsilon$ if it is not used by any path of $\mathcal W$. 
  
  In general, a flow function  $\psi:E(\gparallel) \rightarrow \Sigma^*$  is defined as follows. 
  For a vertex $v$, let $e_1,\dots,e_\ell$ be its incident edges in  $\gparallel$ sorted in the clockwise direction.
  Then $h(v)$ be the product of  $\psi(e_r)^{v(e_r)}$'s for all $r=1,2,\ldots,\ell$,
  where the sign $v(e_r)$ is positive if $e_r$ is an outgoing edge from $v$, and negative, otherwise. 
  A function $\psi$ is called a \emph{flow function}
  if $h(s_i)=i$ for all $i\in[k]$, $h(t_i)=i^{-1}$ for all $i\in[k]$,
  and $h(v)=\epsilon$ for $v\notin \bar{T}$.  

  \paragraph{Homology.}
  For a directed edge $e$ of $\gparallel$, we let $L_e$ denote the face lying to the left of $e$, and let $R_e$ denote
  the face lying to the right of $e$. 
  We say two flow functions $\phi$ and $\psi$ are \emph{homologous} if there exists a \emph{homology} function $f: \mathcal{F}\to \{1,1^{-1},\dots, k,k^{-1}\}^*$, where $\mathcal F$ is the set of all faces of $\gparallel$, such that 
  \begin{itemize}
      \item {$f(F^*)=\epsilon$, and }
      \item {$f(L_e)^{-1}\cdot \phi(e)\cdot f(R_e)=\psi(e)$ for every directed edge $e$.}
  \end{itemize}
  
    \paragraph*{Homology Feasibility Problem.}
  Our goal in this section is to compute a flow function $\psi$ homologous to a given flow function $\phi$ such that
  $\psi$ represents a $T$-\emph{linkage}. 
  To control the number of edges with $\psi(\cdot)\neq \epsilon$ incident to a common vertex,
  Schrijver~\cite{schrijver1994finding} introduced a more general problem setting: the \emph{homology feasibility problem}.
  In this problem, we are given a flow function $\phi$, and a candidate function $\Gamma: {\Pi} \to 2^{\Sigma^*}$,
  where $\Pi$ is a set of \emph{face-edge paths}.
  Here, a face-edge path in $\gparallel$ is an alternating sequence of edges and faces such that 
  two consecutive edge and face are adjacent in $\gparallel$. 
  Moreover,  $\Gamma(\pi)$ must be \emph{hereditary}, that is,
  for every string $x\in \Gamma(\pi)$, all its prefixes and $x^{-1}$ also belong to $\Gamma(\pi)$.
  In the following, for a flow function $\psi$, we let 
  $\psi(\pi)$ denote the product of $\psi(e)^{\mu(e)}$'s for all edges $e$ in $\pi$,
  where the sign $\mu(e)$ is positive if $L_e$ is the face lying previous to $e$ in $\pi$, and negative, otherwise. 
  
  Then the goal is to compute a flow function $\psi$ homologous to a given flow function $\phi$ such that 
  $\psi(\pi)\in \Gamma(\pi)$. 
  Schrijver showed that a flow function $\psi$ with $\psi(\pi)\in \Gamma(\pi)$ for all $\pi\in\Pi$ corresponds to a $T$-linkage
  if $\Pi$ is the set of all face-edge paths each consisting of faces and edges sharing a common vertex,
  and if $\Gamma(\pi)=\Sigma\cup\Sigma^{-1}$ for all $\pi\in \Pi$.
  Therefore, it suffices to present an algorithm for the homology feasibility problem. 
  
  \paragraph{Pre-Feasible Function.} 
  Recall that our goal is compute a homology function $f$ such that $\psi(\pi)\in\Gamma(\pi)$ for all face-edge paths $\pi\in\Pi$
   for a given $\Pi$, $\Gamma$, and $\phi$, where $\psi(e)=f(L_e)^{-1}\cdot \phi(e) \cdot f(R_e)$ for an edge $e\in E(\gparallel)$.
  In this case, we call $f$ a \emph{feasible (homology) function}.
  Also, a function $f:\mathcal F\to \Sigma^*$ with $f(F^*)=\epsilon$ 
  is called a \emph{pre-feasible (homology) function} (with respect to $\phi$) if 
  for each face-edge path $\pi$ of $\Pi$, either $\psi(\pi)\in \Gamma(\pi)$, or $f(F)=f(F')=\epsilon$. 

  For two strings $\sigma$ and $\sigma'$ in $\Sigma^*$, we define the \emph{join} of them, denoted by $\sigma\vee \sigma'$, as the smallest string which starts with both $\sigma$ and $\sigma'$. If no such string exists, we set $x\vee y$ is infinite. Furthermore, two functions $f$ and $g: \mathcal F\to \Sigma^*$,  the join $f\vee g$ is defined as $(f\vee g)(F)=f(F)\vee g(F)$ for all faces $F$ of $\mathcal F$. If $(f\vee g)(F)$ is infinite, then we say $f\vee g$ is infinite. Note that, if $f \vee g$ is finite, then $(f\vee g)(F)$ is equal to $f(F)$ or $g(F)$. Furthermore, if $f\vee g=g$, then we say $f$ is \emph{smaller} than $g$.

  \subsubsection{Finding a Smallest Pre-Feasible $\bar f$} \label{sec:subroutine}
   For a function $f: \mathcal F\to \Sigma^*$ with $f(F^*)=\epsilon$, Schrijver~\cite{schrijver1994finding} showed that
  a \emph{smallest} pre-feasible function $\bar f$ larger than $f$ can be well-defined. That is, any finite pre-feasible function $f^*$ larger than $f$ is also larger than $\bar f$. 
  Furthermore, Section 2.4 of ~\cite{schrijver1994finding} describes an algorithm for finding a smallest pre-feasible function $\bar f$ for a given function $f$ with $f(F^*)=\epsilon$. 
  It will be used as a subroutine for computing a $T$-linkage homotopic to $\mathcal W$. 
  The subroutine is based on the following lemma. 

  \begin{definition}
 We say a face-edge path $\pi \in \Pi$ is a \emph{violating path} if 
  $f(F)^{-1}\cdot \phi(\pi)\cdot f(F') \notin \Gamma(\pi)$, and $\phi(e)\neq\epsilon$, 
     where $F$ and $F'$ are the end faces of $\pi$, and $e$ is the first or last edge in $\pi$. 
  \end{definition}

  
 \begin{lemma}[Proposition~2 in~\cite{schrijver1994finding}]\label{lem:violating_face_pair}
    Assume that $f(F^*)=\epsilon$ and $\bar f$ is finite. For a violating face-edge pair $\pi\in \Pi$, 
    \begin{itemize}
         \item $f(F)$ is smaller than $\phi(\pi)f(F')$, or 
         \item $f(F')$ is smaller than  $\phi(\pi)^{-1}f(F)$, but not both, 
         \end{itemize}
    where $F$ and $F'$ are the end faces of $\pi$. 
  \end{lemma}

  The algorithm for finding a smallest pre-feasible function $\bar f$ described in~\cite{schrijver1994finding} consists of several iterations. 
  In each iteration, we find a violating path $\pi$ and reset the string $f(F')$ as the shortest string $x$ which makes $x^{-1}\cdot \phi(\pi) \cdot f(F) \in \Gamma(\pi)$ if $\pi$ belongs to the first case of Lemma~\ref{lem:violating_face_pair}.
  Otherwise, we reset $f(F)$ as the as the shortest string $x$ which makes $x^{-1}\cdot \phi(\pi^{-1}) \cdot f(F') \in \Gamma(\pi)$. 
  We repeat this procedure until $f(\cdot)$ becomes pre-feasible. 
  The analysis of~\cite{schrijver1994finding} shows that 
  one can find a smallest pre-feasible function $\bar f$ of $f$  
  in $O(n^4)$ time (assuming $\mathcal A$ has size $O(n^2)$, the size of $|\Gamma(\cdot)|$ is constant.) 
  However, in Section~\ref{sec:pre_feasible}, we show that 
  this procedure takes $2^{O(k)}n$ time in our case 
  due to the properties of $\mathcal W$.

  \subsubsection{Main Algorithm of~\cite{schrijver1994finding}}
  The algorithm of~\cite{schrijver1994finding} first computes several \emph{initial} functions $f$'s, compute their smallest pre-feasible functions, and then take the \emph{join} of them. Schrijver showed how to choose
  such initial functions $f$'s so that the join of their smallest pre-feasible functions becomes a flow function corresponding to a $T$-linkage
  homotopic to $\mathcal W$. 
  In Section~\ref{sec:pre_feasible}, we will show that $|\bar{f}(\cdot)|$ has length $2^{O(k)}$ for all 
  functions $\bar{f}$ we consider in the course of the algorithm.
  Therefore, to analyze the running time of this algorithm,
  it suffices to analyze the number of functions $f$'s we compute. 

  To do this, we briefly describe how the algorithm of~\cite{schrijver1994finding} chooses initial functions $f$'s. 
  To make the description easier, we consider $\pi$ as an oriented face-edge path. Then let $\pi^{-1}$ be the face-edge path
  obtained from $\pi$ by reversing its orientation.  
  Let $\phi$ be the flow function of the given weak $T$-linkage of $\gparallel$. 
 For each violating path $\pi\in \Pi$, we let $f_\pi$ be the function that maps the first face to $\phi(e)$ 
 and maps all other faces to $\epsilon$, where $e$ be first edge in $\pi$. 
 Furthermore, we say two paths $\pi$ and $\pi'$ are \emph{co-viable} if $\bar f_{\pi} \vee \bar f_{\pi'}$ is finite.
 Recall that $\bar{f}$ is a smallest pre-feasible function of $f$. 
 
 Let $\mathcal X$ be the set of all face-edge paths $\pi$ such that $\phi(\pi)\notin \Gamma(\pi)$. 
 Note that, if $\pi \in \mathcal X$, then $\pi^{-1}\in \mathcal X$. 
 We compute $\bar f_\pi$ for each pair $\pi\in \mathcal X$, and check 
 if $\pi$ and $\pi'$ are co-viable for each pair $(\pi,\pi')$ of face-edge paths in $\mathcal X$. 
 Then, we find a set $X\subset \mathcal X$ such that \textsf{(i)} For each path $\pi$ of $\mathcal X$, $X$ contains at least one of $\pi$ or $\pi^{-1}$, \textsf{(ii)} every pair of face-edge paths in $X$ is co-viable. 
 Then we return $\bigvee_{\pi\in X} \bar f_\pi$. 
 For its correctness, refer to~\cite{schrijver1994finding}. 

    \paragraph*{Time Complexity.} 
    In Section~\ref{sec:reducing_violation}, we modify $G$ and construct $\Gamma$ so that the size of $\mathcal X$ is $2^{O(k)}$. 
    (Also, we will represent each path of $\mathcal X$ in a compact way.) 
    Also, in Section~\ref{sec:pre_feasible}, we show that the subroutine for computing $\bar{f}_\pi$ takes $2^{O(k)}n$ time in our case. 
    This means that $\bar{f}_\pi$ has length $2^{O(k)}$. Thus we can find all pairs of co-viable face-edge paths of $\mathcal X$
    in $2^{O(k)}$ time in total.
    Then computing $X$ is indeed a variant of the 2-SAT problem. Thus, this takes 
    polynomial time in the size of $\mathcal X$, which is $2^{O(k)}$. 
    Finally, we join $2^{O(k)}$ strings each of length $2^{O(k)}$. Therefore, the total running time for computing a $T$-linkage homotopic to $\mathcal W$ is $2^{O(k)}$.

\medskip   
  The remaining tasks are to modify $\gparallel$ and construct $\Gamma$ so that the size of $\mathcal X$ is $2^{O(k)}$,
  and to show that the running time of the subroutine is $2^{O(k)}n$.
   
 \subsection{Modification of $\gparallel$ and Construction of $\Gamma$}\label{sec:reducing_violation}
 We first modify $\gparallel$ so that every vertex of the frames and the skeleton forests has degree $2^{O(k)}$.
 Then the weak linkage $\mathcal W$ changes accordingly. We use $\gmod$ to denote the resulting graph, and $\wmod$ to denote the resulting weak linkage. 
 Then we will see that 
 there is a $T$-linkage of $\gparallel$ homotopic to $\mathcal W$ if and only if there is a $T$-linkage of $\gmod$ 
 homotopic to $\wmod$. 
 
 \begin{figure}
     \centering
     \includegraphics{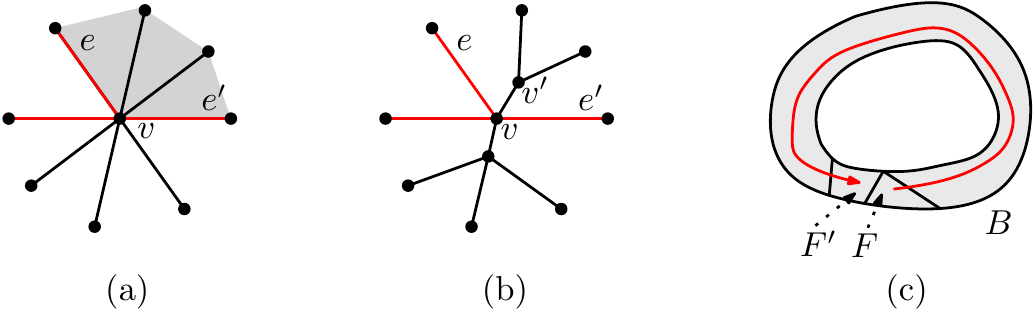}
     \caption{\small (a) The wedge at $v$ defined by $e$ and $e'$ is colored gray. The linkage-edges are colored red. 
     (b) We add two new vertices, one per wedge at $v$. (c) For a frame $B$, we construct a face-edge path (red curve) 
     consisting of the edges and faces
     incident to $B$ and lying in the interior of $B$. Here, the end faces $F$ and $F'$ are chosen arbitrary.}
     \label{fig:gmod}
 \end{figure}
 
 \paragraph{Modification of $\gparallel$ to $\gmod$.}
 We say an edge $e$ is \emph{linkage-edge} if a walk of $\mathcal W$ uses $e$.
 For a vertex $v$ of $\gparallel$, let $d_W(v)$ be the number of linkage-edges 
 incident to $v$. Also, let $d_G(v)$ be the degree of $v$ on $G$.  
 We say an edge $e$ lies \emph{between $e_1$ and $e_2$} if $e$, $e_1$, and $e_2$ are incident to a common vertex $v$, and $e$ lies from $e_1$ to $e_2$ in clockwise direction around $v$.
 Two linkage-edges $e$ and $e'$ incident to $v$ form a \emph{wedge at $v$} if there is no other linkage-edge between $e$ and $e'$. 
 Furthermore, the wedge contains all edges between $e$ and $e'$, and all wedges at $v$ are pairwise interior-disjoint. 
 See Figure~\ref{fig:gmod}.
 There are at most $d_W(v)$ wedges at a vertex $v$. We say a wedge at $v$ is \emph{empty} if
 no edge of $\gparallel$ incident to $v$ is contained in the wedge. 
 
 For each non-empty wedge at a vertex $v$ in the frames and skeleton forests, we insert a new vertex $v'$ to $\gparallel$,
 and add the edge between $v$ and $v'$. 
 Then we remove the edges in the wedge, and reconnect them to $v'$ instead of $v$. Note that, the edges incident to a new vertex $v'$ are not linkage-edges. 
 We do this for all non-empty wedges and all vertices in the frames and skeleton forests. 
 This takes $O(n)$ time in total since the total degree of all vertices in the frames and skeleton forests is $O(n)$. 
 Clearly, there is a one-to-one correspondence between two linkages (and weak linkage) on $\gparallel$ and on $\gmod$.
 Thus, there is a weak linkage $\mathcal W_\textsf{mod}$ in $\gmod$ corresponds the canonical weak linkage $\mathcal W$ in $\gparallel$ constructed in Section~\ref{sec:weak_linkage_construction}. We call the linkage a canonical weak linkage in $\gmod$. Furthermore, from now on, the flow function $\phi$ is corresponding the canonical weak linkage $\mathcal W_\textsf{mod}$ in $\gmod$.


\paragraph{Construction of $\Gamma$ and $\Pi$.}
Let $\Pi$ be the set of all face-edge paths such that 
all faces and edges in $\pi$ are incident to a common vertex and sorted in clockwise direction.
Let $\Pi_\textsf{frame}$ be the set of face-edge paths, one per frame, such that 
all incident faces and edges to a frame $B$ are in $\pi$,
and no edge of $\pi$ lies on $B$. 
Here, the number of paths of $\Pi$ can be $\Omega(n^2)$ in the worst case, and each path of $\pi$ has length $\Theta(n)$ 
in the worst case.
As we will see later, we do not need to construct $\Pi$ explicitly.

For a face-edge path $\pi$, we let $\Gamma(\pi)=\emptyset$ if $\pi$ has exactly one edge of $\gmod$, and it is is constructed due to the radial completion of $G$, 
and let $\Gamma(\pi)=\Sigma \cup \Sigma^{-1} \cup \{\epsilon\}$, otherwise. Recall that the edges of $\gmod$ constructed due to the radial
compleition of $G$ lie on the frames only, and thus the number of such edges is $2^{O(k)}$. 
This is the exactly same was as the definition of $\Gamma$ in~\cite{schrijver1994finding}.
Then Schrijver showed that a flow function $\psi$ with $\psi(\pi)\in\Gamma(\pi)$ corresponds to a linkage. 
Here, we additionally set $\Gamma(\pi)$ for a face-edge path $\pi$ in $\Pi_\textsf{frame}$.
Recall that the crossing pattern tells us the sequence of the crossings between each frame and the paths of our target linkage $\mathcal P$ 
sorted along the frame. In fact, we construct the weak linkage $\mathcal W$ in Section~\ref{sec:weak_linkage_construction}
so that the walks of $\mathcal W$ and $\mathcal W_\textsf{mod}$ cross each frame in the same order as the paths of $\mathcal P$.
See Figure~\ref{fig:gmod}(c). 
For each face-path $\pi$ in $\Pi_\textsf{frame}$, we let $\Gamma(\pi)=\{\phi(\pi)\}$. Recall that the flow function $\phi$ is corresponding the canonical weak linkage $\mathcal W_\textsf{mod}$ in $\gmod$.

\begin{lemma}
 The set $\mathcal X$ of all face-edge paths $\pi$ in $\Pi\cup \Pi_\textsf{frame}$ with $\phi(\pi)\notin \Gamma(\pi)$ has size $2^{O(k)}$.
\end{lemma}
\begin{proof}
By construction, no face-edge path of $\Pi_\textsf{frame}$ is contained in $\mathcal X$ by definition. 
For a face-edge path $\pi$ of $\Pi$, 
all faces and edges in $\pi$ are incident to a common vertex $v$ and sorted in clockwise direction.
If $\wmod$ traverses $v$ exactly once, $\phi(\pi)$ has length at most one, and thus $\phi(\pi)\in\Gamma(\pi)$ by construction.
Thus $\pi$ is contained in $\mathcal X$ only when $v$ is traversed by $\wmod$ more than once.
Such a vertex lies on a frame or a skeleton forest. For each vertex in a frame or a skeleton forest, its degree in $\gmod$ is $2^{O(k)}$
by construction.
Therefore, for each vertex $v$, there are $2^{O(k)}$ face-edges paths whose faces and edges are incident to $v$. 
Therefore, in total, the number of face-edge paths in $\mathcal X$ is $2^{O(k)}$. 
\end{proof}

The argument of~\cite{schrijver1994finding} shows that the resulting flow function represents to a $T$-linkage $\mathcal P_\textsf{mod}$ of $\gmod$, and Lemma~\ref{lem:homology-homotopy} implies that the algorithm of~\cite{schrijver1994finding} always finds
a $T$-linkage if the given weak $T$-linkage is discretely homotopic to a $T$-linkage.  
Moreover, we can reconstruct $G$ and $\mathcal P$ from $\gmod$ and $\mathcal P_\textsf{mod}$, respectively, in the reversed way.
    The following lemma summarizes this subsection.
    \begin{lemma}\label{lem:without-subroutine}
     Given a weak $T$-linkage of $G$ homotopic to a $T$-linkage of $G$,
     we can find a $T$-linkage of $G$ in $2^{O(k)}\cdot T(n) +2^{O(k)}n$ time, where $T(n)$ denotes the running time of the subroutine
     described in Section~\ref{sec:subroutine}. 
    \end{lemma}
    
 For convenient to describe, $\mathcal W$ refers $\mathcal W_\textsf{mod}$ if there is no explanation.
 Here, we need one more observation. 
 For any weak $T$-linkage $\mathcal W$ represented by $f$, 
 there is a weak $T$-linkage represented by $f'$ discretely homotopic to $\mathcal W$, 
 where $f'$ is a function obtained from $f$ during the process of computing $\bar{f}$. 
 This is simply because each iteration of the process described in Section~\ref{sec:subroutine} is indeed a face operation.  
 Due to $\Pi_\textsf{frame}$, the following observation holds. This will be used in Lemma~\ref{lem:length}. 
 
\begin{observation}\label{obs:winding} 
 Let $\psi$ be a flow function homologous to $\phi$, and $\xring$ be a terminal-free framed ring of $\gmod$.
 The absolute value of the winding number between a maximal walk of the weak linkage of $\psi$ contained in $\gmod$
 and a maximal walk of $\mathcal W$ contained in $\gmod$ is $2^{O(k)}$.
 Moreover, the number of maximal walks of the weak linkage of $\psi$ contained in $\gmod$ is $2^{O(k)}$.
\end{observation}


\subsection{Analysis of the Subroutine}\label{sec:pre_feasible}
In this section, we analyze the running time of the algorithm for finding a smallest pre-feasible function $\bar f$. There are two tasks: First, we show that the string size $|\bar f(F)|$ is $2^{O(k)}$ for every face $F$ of $\gmod$. 
Second, we show how to compute a violating face-edge path in constant time in each iteration.  

\begin{lemma}\label{lem:length}
The length of $\bar f_\pi (F)$ is $2^{O(k)}$ for every face $F$ of $\gmod$. 
\end{lemma}
\begin{proof}
Recall that $\phi$ is the flow function of the input weak $T$-linkage $\mathcal W$, and $\bar{f}_\pi(F^*)=\epsilon$ by the definition of the pre-feasible function. Let $\psi$ be the flow function $E\to \Sigma^*$ such that $\bar{f}_\pi(L_e)^{-1}\cdot \phi(e) \cdot \bar{f}_\pi(R_e)=\psi(e)$. 
By definition, a violating face-edge path with respect to $\bar{f}_\pi$ is also a violating face-edge path with respect to $f_\pi$. 
Also, let $\mathcal W'$ be the weak linkage of $\psi$.  

First, we show that $\bar{f}_\pi (F)$ has length $2^{O(k)}$ for all faces $F$ in a terminal-containing ring $\oring$
assuming that $\bar{f}_\pi (F')$ has length $2^{O(k)}$ for all faces incident to its inner boundary.
The radial distance from $F^*$ to the boundary of the ring is $2^{O(k)}$ by construction.
We construct an arbitrary face-edge path $\gamma_i$ from a face, say $F'$, incident to its inner boundary to $F$. 
The number of edges of $\gamma_i$ used by $\mathcal W'$ or $\mathcal W$ is $2^{O(k)}$. This is because 
the degree in $\gmod$ of a vertex traversed by $\mathcal W'$ or $\mathcal W$ more than once is $2^{O(k)}$.  
Therefore, $|\psi(\gamma_i)|=2^{O(k)}$ $|\phi(\gamma_i)|\in 2^{O(k)}$. 
Since $|\bar{f}_\pi(F')^{-1}\phi(\gamma_i) \bar{f}_\pi(F)|=\psi(\gamma_i)$, we have $|\bar{f}_\pi(F)|\in 2^{O(k)}$.

Then we show that $\bar{f}_\pi (F)$ has length $2^{O(k)}$ for all faces $F$ in a terminal-free ring $\xring$ 
assuming that $\bar{f}_\pi (F')$ has length $2^{O(k)}$ for all faces incident to its inner boundary.
By the construction of $\mathcal W$,
there is a face-edge path $\gamma_i$ in $\xring$ from a face $F'$ incident to the inner boundary of $\xring$
to $F$ such that no edge of $\gamma_i$ is traversed by $\mathcal W$. 
However, it is possible that lots of edges of $\gamma_i$ are traversed by $\mathcal W'$.
But even in this case, $\psi(\gamma_i)$ has length $2^{O(k)}$ since $\phi$ and $\psi$ are homologous. 
More specifically, for a maximal subwalk $\omega'$ of $\mathcal W'$ contained in $\xring$ 
and one endpoint on the outer boundary of $\xring$, the number of symbols in $\psi(\gamma_i)$ induced by $\omega'$ 
is at most $|\windnum(\omega',\gamma_i)|+1$. 
For a maximal subwalk $\omega$ of $\mathcal W$, we have $|\windnum(\omega, \gamma_i)|=2^{O(k)}$ by construction.
Moreover, $|\windnum(\omega,\omega')|= 2^{O(k)}$, and the number of such maximal walks $\omega'$ is $2^{O(k)}$ 
by Observation~\ref{obs:winding}.
Therefore, the length of $\psi(\gamma_i)$ is $2^{O(k)}$, and thus the length of $\bar{f}_\pi(F)$ is $2^{O(k)}$.

Since $\bar{f}_\pi(F^*)=\epsilon$, and the number of framed rings is at most $k$,
this implies that $\bar{f}_\pi(F)$ has length $2^{O(k)}$ for all faces $F$ of $\gmod$. 
\end{proof}

\begin{corollary}
    The number of iterations of the subroutine described in~\ref{sec:subroutine} is $2^{O(k)}n$. 
\end{corollary}
At each iteration of the algorithm, we find a violating face-edge path $\pi\in \Pi$, and
update the string of $f$ for exactly one face of $\gmod$. 
Once we have $\pi$, then the update takes time linear in the complexity of $\bar{f}_\pi(F)$, which is $2^{O(k)}$ by Lemma~\ref{lem:length}. 
Thus it is sufficient to show that we can find $\pi$ in amortized constant time. 
To do this, we use a data structures for computing $O(1)$ candidates for a violating face-edge path maintained dynamically in the course of updates of $f$. 


\paragraph*{Dynamic Data Structure for Finding a Violating Face-Edge Pair.} 
We let $\psi$ be the flow function with $\psi(e)$ as $f(L_e)^{-1}\cdot \phi(e)\cdot f(R_e)$. 
By definition, $\pi$ is a violating face-edge path if and only if $\psi(\pi)\notin \Gamma (\pi)$ and $f(F)\neq\epsilon$,
where $F$ is a first or last face of $\pi$. 
Also, the length of $f(F)$ is at most $2^{O(k)}$ by Lemma~\ref{lem:length}. 

First, the number of face-edge paths in $\Pi_\textsf{frame}$ is $O(k)$, and 
each face-edge path $\pi$ of $\Pi_\textsf{frame}$ has length $2^{O(k)}$, we maintain
the value $\psi(\pi)$ explicitly in the course of updates of $f$.
Similarly, the number of face-edge paths in $\Pi$ each consisting of faces and edges incident to a vertex in the frames
is $2^{O(k)}$ due to the modification of $\gparallel$ into $\gmod$. 
Also, each such path has length $2^{O(k)}$, we maintain
the value $\psi(\pi)$ explicitly in the course of updates of $f$. 

Therefore, in the following, it suffices to design a data structure for finding a violating face-edge pair
consisting of edges and faces incident to a vertex not contained in any frame. 
Moreover, in this case,
if $\psi(\pi)\notin \Gamma (\pi)$, then $f(F)\neq \epsilon$ for a first or last face $F$ of $\pi$. 
Therefore, it is sufficient to check if  $\psi(\pi)\notin \Gamma (\pi)$ only to determine if $\pi$ is a violating face-edge path or not. 

The data structure marks the vertices as \emph{``un-feasible''} to remember that 
a violating face-edge path lies \emph{near} these vertices. 
Also, it maintains the list of un-feasible vertices. 
For every vertex $v$ of $\gparallel$ not lying on any frame, 
the data structure stores $k$ linked lists. 
The $\ell$th list stores the edges $e$ of $\gmod$ incident to $v$ with $\ell\in \psi(e)$. Here, we maintain the edges in the list
in an arbitrary order. 
If $v$ has at least two non-empty linked lists, there is a violating face-edge path incident to $v$. 
Thus, we mark $v$ as \emph{``un-feasible''}. 

\paragraph{Query and Update Algorithms.}
To answer a query, we first check if there is a vertices marked as ``un-feasible''. 
If a vertex $v$ is marked as ``un-feasible'', there exist two non-empty linked lists associated with $v$. 
We select two edges $e$ and $e'$ stored in two different linked list associated with $v$. 
There are four faces incident to $e$ or $e'$. 
At least one face-edge path $\pi$ between them is a violating path by the definition. 
That is, we have $O(1)$ candidates for a violating path. 
For each candidate $\pi$, we check if it is a violating pair in $2^{O(k)}$ time as follows.
First, we compute $\psi(\pi)=f(F)^{-1}\phi(\pi)f(F')$,
where $F$ and $F'$ are the first and last faces of $\pi$. Recall that $\phi$ do not change during the subroutine, and we have computed $\phi(\pi)$. Since the lengths of $f(F)^{-1}$, $\phi(\pi)$, and $f(F')$ are all $2^{O(k)}$ by Lemma~\ref{lem:length}, we can compute
$f(F)^{-1}\phi(\pi)f(F')$ in $2^{O(k)}$ time. 
Therefore, we can check if $\pi$ is a violating path or not in the same time bound. 
In this way, we can answer the query in $2^{O(k)}$ time in total. 

During the algorithm, we make updating $f(\cdot)$. We are required to update the data structure accordingly.
If we update $f(F)$, then we update the linked lists for all vertices lying on the face of $F$, and we check if
each such vertex is required to be marked as  \emph{``un-feasible''}.
This process takes $O(k)$ time for each vertex. Thus, the total update time is $2^{O(k)}n$ time
The dynamic structure supports the candidate face-edge paths in $O(k)$ time. 
By Lemma~\ref{lem:length} and the argument of~\cite{schrijver1994finding},
$f$ is updated for each face $F$ in $2^{O(k)}$ times. This is simply because the value of $f(F)$ gets longer during the execution
of the subroutine, and the length of $\bar{f}(F)$ is $2^{O(k)}$.
Therefore, the total update time for a single execution of the subroutine is $2^{O(k)}n$.

Therefore, we have the following lemma.
\begin{lemma}\label{lem:time-subroutine}
 The subroutine described in~\ref{sec:subroutine} can be implemented to run in $2^{O(k)}n$ time. 
\end{lemma}

By combining Lemmas~\ref{lem:num-crossing},~\ref{lem:enumeration-time},~\ref{lem:without-subroutine} 
 and Lemma~\ref{lem:time-subroutine}, we can obtain the following theorem.
 
\begin{theorem}
    The \textsf{Planar Disjoint Paths} problem can be solved in $2^{O(k^2)}n$ time.
\end{theorem}

\bibliographystyle{plain}
\bibliography{paper}

\end{document}